\documentclass[12pt]{article}
\usepackage{amssymb}
\usepackage{amsfonts}
\usepackage{amsmath}
\usepackage{harvard}
\usepackage{chicago}
\usepackage{graphicx}

\newtheorem{theorem}{Theorem}

\newtheorem{condition}{Condition}

\newtheorem{lemma}{Lemma}

\newenvironment{proof}[1][Proof]{\textbf{#1.} }{\ \rule{0.5em}{0.5em}}
\input{tcilatex}
\renewcommand{\cite}{\citeasnoun}
\renewcommand{\thepage}{}
\renewcommand{\thefootnote}{\fnsymbol{footnote}}
\allowdisplaybreaks
\renewcommand{\baselinestretch}{1.3}

\begin{document}

\title{A More Robust t-Test\thanks{%
I am thankful for very helpful comments and advice from Angus Deaton, Hank
Farber, Bo Honor\'{e}, Karsten M\"{u}ller and participants at various
workshops. Financial support from the National Science Foundation through
grant SES-1919336 is gratefully acknowledged.}}
\author{Ulrich K. M\"{u}ller \\
Princeton University\\
Department of Economics\\
Princeton, NJ, 08544\vspace{-0.4cm}\vspace*{0.4cm}}
\date{July 2020}
\maketitle

\begin{abstract}
Standard inference about a scalar parameter estimated via GMM amounts to
applying a t-test to a particular set of observations. If the number of
observations is not very large, then moderately heavy tails can lead to poor
behavior of the t-test. This is a particular problem under clustering, since
the number of observations then corresponds to the number of clusters, and
heterogeneity in cluster sizes induces a form of heavy tails. This paper
combines extreme value theory for the smallest and largest observations with
a normal approximation for the average of the remaining observations to
construct a more robust alternative to the t-test. The new test is found to
control size much more successfully in small samples compared to existing
methods. Analytical results in the canonical inference for the mean problem
demonstrate that the new test provides a refinement over the full sample
t-test under more than two but less than three moments, while the
bootstrapped t-test does not.

\end{abstract}

\newpage \setcounter{page}{1} \renewcommand{\thepage}{\arabic{page}} %
\renewcommand{\thefootnote}{\arabic{footnote}} 
\renewcommand{\baselinestretch}{1.25} \small \normalsize%

\section{Introduction}

The usual t-test for inference about the mean of a population from an
i.i.d.~sample is a key building block of statistics and econometrics. Not
only does it have many direct applications, but also many other standard
forms of inference reduce to the application of a t-test applied to a
suitably defined population. For example, consider a linear regression with
scalar regressor, $Y_{i}=X_{i}\beta +\varepsilon _{i}$, $\mathbb{E}%
[X_{i}\varepsilon _{i}]=0$. A test of the null hypothesis $H_{0}:\beta
=\beta _{0}$ reduces to a test of $\mathbb{E}[W_{i}]=0$ for $%
W_{i}=(Y_{i}-X_{i}\beta _{0})X_{i}$, and the usual t-statistic computed from
the i.i.d.~sample $W_{i}$ amounts to a specific version of the usual
heteroskedasticity robust test suggested by \cite{White80}. Under clustering
that allows for arbitrary correlations between $\varepsilon _{j}$ for all $%
j\in \mathcal{C}_{i}$, $i=1,\ldots ,n$, the effective observations become $%
W_{i}=\sum_{j\in \mathcal{C}_{i}}(Y_{j}-X_{j}\beta _{0})X_{j}$. In the
presence of additional controls $Y_{i}=X_{i}\beta +Z_{i}^{\prime }\gamma
+\varepsilon _{i}$, the equivalence to the inference for the mean problem
holds approximately after projecting $Y_{i}$ and $X_{i}$ off $Z_{i}$. This
further extends to instrumental variable regression and parameters estimated
by GMM.

The asymptotic validity of standard t-statistic based inference relies on
two arguments. First, the law of large numbers implies that the variance
estimator in the denominator has negligible estimation error. Second, a
central limit theorem applied to the numerator yields approximate normality.
Underlying populations with heavy tails are a threat to both. Even if the
second moment exists, so that t-statistic based inference is asymptotically
justified, large samples might be required before these approximations
become accurate.

The effective sample size in empirical work is often considerably smaller
than the raw number of observations, and not all that large. This can arise
because researchers are interested in inference for smaller subgroups, or
because nonparametric kernel estimators are employed that effectively depend
only on relatively few observations, or because the relevant variation only
stems from a small fraction of observations, such as in studies about rare
events. It is also very common for standard errors to be clustered, reducing
the effective number of independent observations to the number of clusters,
which tends to be only moderately large. What is more, in many applications
clusters are of fairly heterogeneous size (think of the 50 states of the
U.S., say). Even if none of the variables under study are heavy-tailed, a
substantial portion of the parameter sampling variation will then stem from
the randomness of the large clusters, inducing a form of heavy-tailedness in
the resulting $W_{i}$ variables. See Section \ref{sec:clust} for an
illustration.

This paper develops an alternative to the t-test that performs more reliably
when the underlying population has potentially heavy tails. The focus is
exclusively on the case of moderately heavy tails, that is, the first two
moments of $W_{i}$ exist, so that asymptotically, standard t-statistic based
inference is valid. The aim is to devise an inference method that does not
overreject if the underlying population has moderately heavy tails, without
losing much in terms of efficiency if the underlying population has light
tails. The theoretical development only concerns the canonical inference for
the mean problem, but we show how to adapt the procedure to also obtain more
reliable inference about scalar parameters estimated by GMM, including the
case with clustering. In such more general contexts, the quality of standard
inference is poor if the induced $W_{i}$ has heavy tails, for which it is
neither necessary nor sufficient that the variables under study are
individually heavy-tailed.

To describe the key idea, consider the hypothesis test of $H_{0}:\mathbb{E}%
[W_{i}]=0$ against $H_{a}:\mathbb{E}[W_{i}]\neq 0$ based on an i.i.d.~sample 
$W_{i}$, $i=1,\ldots ,n$, from a population $W$ with cumulative distribution
function $F$. For expositional ease, suppose that $F$ has a thin left tail,
but a potentially heavy right tail. For some given $k$, let $\mathbf{W}%
^{R}=(W_{1}^{R},W_{2}^{R},\cdots ,W_{k}^{R})^{\prime }$ be the $k$ largest
order statistics, with $W_{1}^{R}$ the sample maximum. Conditional on $%
\mathbf{W}^{R}$, the remaining \textquotedblleft small\textquotedblright\
observations $W_{i}^{s}$, $i=1,\ldots ,n-k$ are i.i.d.~draws from the
truncated distribution with c.d.f. $F(w)/F(W_{k}^{R})$ for $w\leq W_{k}^{R}$%
. The mean of this truncated distribution under $H_{0}$ is no longer zero,
however, but is given by $-m(W_{k}^{R})<0$, where%
\begin{equation*}
m(w)=-\mathbb{E}[W|W\leq w]=\frac{\mathbb{P}(W>w)\mathbb{E}[W|W>w]}{1-%
\mathbb{P}(W>w)}.
\end{equation*}%
Note that $m(w)$ for $w$ large is determined by the properties of $F$ in its
right tail.

The idea now is to apply three asymptotic approximations. First, invoke
standard extreme value theory to obtain an approximation for the
distribution of $\mathbf{W}^{R}$ in terms of a (joint)\ extreme value
distribution governed by three parameters describing location, scale and
shape. Second, apply the central limit theorem to the conditional
i.i.d.~sample of remaining observations $W_{i}^{s}$ from the truncated (and
hence no longer heavy-tailed) distribution to argue that $%
(n-k)^{-1}\sum_{i=1}^{n-k}W_{i}^{s}$ is approximately normal with mean $%
-m(W_{k}^{R})$ under $H_{0}$ (and arbitrarily different mean under the
alternative $H_{a}$). Third, by the same arguments that justify extreme
value theory, obtain an approximation to $m(w)$ in terms of the three
parameters that govern the distributional approximation of $\mathbf{W}^{R}$.

These approximations lead to a \emph{parametric} approximate joint model of $%
k+1$ statistics: $\mathbf{W}^{R}$ is jointly extreme value, and $%
(n-k)^{-1}\sum_{i=1}^{n-k}W_{i}^{s}$ is normally distributed with a mean
that, under $H_{0}$, depends on $W_{k}^{R}$, and the parameters of the
extreme value distribution. For given $k$, this is a small sample
nonstandard parametric testing problem, and one can construct tests that are
exactly valid under the approximate parametric model. Specifically, we apply
computational techniques similar to those developed in \cite{Elliott15} to
determine a powerful that is of level $\alpha $ in this parametric model.
Once the test is applied to the original mean testing problem, it is no
longer of level $\alpha $ by construction. But the explicit modelling of the
potentially moderately heavy tail via extreme value theory might improve
performance over the usual t-test.

The main theoretical result of this paper corroborates this conjecture by
considering higher order improvements for populations with finite variance,
but that do not possess a third moment, and for which extreme value theory
applies. We consider asymptotics in which $k$ is a fixed number that does
not vary as a function of $n$. In this way, the asymptotics reflect that
moderately large samples only contain limited information about the tail
properties of the underlying population. We show that the approximation
error of the parametric model for $k$ fixed induces an error in the
rejection probability in the mean testing problem that converges to zero
faster than the error in rejection probability of the usual t-test. In that
sense, the new approach yields a refinement over the usual t-test and
provides theoretical support for the usefulness of the new perspective.

A natural alternative to obtain more accurate approximations is to consider
the bootstrap. \cite{Bloznelis03} show that the percentile-t bootstrap
provides a refinement whenever the underlying population has at least three
moments. A second, apparently new theoretical result shows that the
bootstrap \emph{does not} provide a refinement when the underlying
population has between two and three moments.

The new method readily generalizes to the case where both tails are
potentially heavy. The approximate parametric model then consists of $2k+1$
statistics, with $k$ joint extreme value observations from the left tail
governed by three parameters, $k$ extreme value observations from the right
tail governed by their own three parameters, and the conditionally normal
average of the middle observations. Since in most applications, there are no
compelling reasons to assume any constraints between the properties of the
left and right tail, the approximate parametric problem is thus indexed by a
six dimensional nuisance parameter. We use a version of the the algorithm of 
\cite{Elliott15} to numerically determine a powerful test in this parametric
problem for selected values of $k$. The large nuisance parameter space turns
this into a major computational challenge. Once a powerful valid test has
been determined, however, applying it in practice is entirely
straightforward and does not pose any significant computational burden. This
includes its use to obtain more reliable inference about scalar parameters
estimated by GMM with potentially clustered errors; see Section \ref{sec:GMM}
for details.

Our preferred default method uses $k=8$ and is appropriate when the sample
consists of at least 50 independent clusters or observations.\footnote{%
We also provide an alternative, even more robust test for $k=4$ that is
applicable to samples with as few as $25$ indepedent clusters or observation.%
} The tests were determined for various significance levels, enabling the
construction of confidence intervals at the 90\%, 95\% and 99\% level via
test inversion, and the computation of p-values. Corresponding tables and
STATA code is provided in the replication files.

Monte Carlo simulations show that the new approach leads to much better size
control in moderately large samples compared to existing methods, at fairly
small cost in terms of average confidence interval length for thin-tailed
populations. This is true in the canonical inference for the mean case, as
predicted by the theory, but also when comparing two means, and for
inference about regression coefficients under clustering. In one design, the
clusters are Metropolitan Statistical Areas, which are fairly heterogeneous
in size. As discussed above, this heterogeneity induces the resulting $W_{i}$
to be quite heavy-tailed, which leads to poor performance of standard
cluster robust inference. A moderately large number of heterogenous clusters
(say, no more than 100 or 200) is quite common in empirical work, making the
new approach a potentially attractive alternative in such settings.

The remainder of the paper is organized as follows. The next section
discusses the related literature and provides a review of known results from
extreme value theory and the approximate distribution of t-statistics.
Section 3 contains the new theoretical results in the inference for the mean
problem. Section 4 provides details on the construction of the new test.
Small sample simulation results are reported in Section 5. Section 6
concludes.

\section{\label{sec:bckgrnd}Background}

\subsection{Relationship to Literature}

The new method \textquotedblleft robustifies\textquotedblright\ the usual
t-test in the sense of providing more reliable inference under moderately
heavy tails. The classic robustness literature (see, for instance, \cite%
{Huber96} for an overview) is based on a very different notion; in this
literature, it is assumed that the observations are contaminated, with a
small fraction not stemming from the population of interest. This in turn
raises the question what type of estimands can still be reliably learned
about, and how to do so efficiently. In contrast, the estimands considered
here are defined relative to the distribution that generated the data. Which
of these views is appropriate depends on the application, and in particular
whether relatively extreme observations are part of the population of
interest, or rather are induced by measurement errors. It is also possible
to combine the approaches, such as applying the new test in a regression of
winsorized variables (with the estimand then defined relative to a
winsorized population).

Even under the assumption that the data is entirely uncontaminated,
informative inference requires assumptions beyond the existence of moments:
The classic impossibility result of \cite{Bahadur56} shows that one cannot
learn about the population mean from i.i.d.~samples of any size, even if all
moments are assumed to exist. The substantial assumption pursued here is
that the population tails are such that extreme value theory provides
reasonable approximations. This effectively amounts to an assumption that
the tails of the underlying distribution are approximately (generalized)
Pareto. Given the theoretical prevalence and empirical success of extreme
value theory for learning about the tail of distributions (for overviews and
references, see, for instance, \cite{Embrechts97} or \cite{DeHaan07}), this
seems a reasonably general starting point, especially given that some
assumption must be made. What is more, the approximate Pareto tail is only
imposed in the extreme tail with approximate mass of $k/n$ for $k$ fixed,
which is enough to ensure that the largest (and smallest) $k$ observations
are governed by extreme value theory.

\cite{Muller17} pursue this \textquotedblleft fixed-$k$\textquotedblright\
approach for the purpose of inference about tail properties, such as extreme
quantiles. In contrast, the remaining literature on the modelling of tails
considers asymptotics where $k=k_{n}$ diverges with the sample size. In
large samples, $k_{n}$ diverging asymptotics allow for consistent estimation
of tail properties, at least pointwise for a fixed population. In practice,
though, the approximations generated from $k_{n}$ diverging asymptotics are
not very useful for, say, samples of size $n=50$ or $n=100$, as there are
only a handful of observations that can usefully be thought of as stemming
from the tail, so that any approximation that invokes \textquotedblleft
consistency\textquotedblright\ of tail property estimators becomes
misleading.

The separate analysis of the largest and remaining terms of a sum of
independent random variables goes back to at least \cite{Csorgo88}; also see 
\cite{Zaliapin05}, \cite{Kratz2014} and \cite{Mueller19}. The relatively
closest precursors to this work are Peng (2001, 2004)\nocite{Peng01}\nocite%
{Peng04}\ and \cite{Johansson03}. These authors are concerned with inference
about the mean from an i.i.d.~sample under very heavy tails, that is, the
underlying population has less than two moments. For such populations, the
usual t-statistic does not converge to a normal distribution. Peng (2001,
2004) and \cite{Johansson03} suggest estimating the contribution of the two
tails to the overall mean by consistently estimating the tail Pareto
parameters using the smallest and largest $k_{n}$ observations, with $k_{n}$
diverging, and combining those estimates with the estimate of the mean of
the remaining middle observations.

Another approach to overcome the \cite{Bahadur56} impossibility result is to
assume bounded support, with known bounds. \cite{Romano00}, \cite{Schlag07}
and \cite{Gossner13} derive corresponding methods.

\subsection{\label{sec:evt}Extreme Value Theory}

Let $W_{1}^{R}\geq W_{2}^{R}\geq ...\geq W_{k}^{R}$ denote the largest $k$
order statistics from an i.i.d.~sample from a population with distribution $%
F $. Suppose the right tail of $F$ is approximately Pareto in the sense that
for some scale parameter $\sigma >0$ and tail index $\xi >0$%
\begin{equation}
\lim_{w\rightarrow \infty }\frac{1-F(w)}{(w/\sigma )^{-1/\xi }}=1
\label{evt1}
\end{equation}%
so that the second moment of $W$ exists if and only if $\xi <1/2$. Then $W$
is in the maximum domain of attraction of the Fr\'{e}chet limit law 
\begin{equation}
n^{-\xi }W_{1}^{R}\Rightarrow \sigma X_{1}  \label{eq:evtmax}
\end{equation}%
where $X_{1}^{-1/\xi }\sim E_{1}$ with $E_{1}$ an exponentially distributed
random variable.

As is well known (see, for instance, Theorem 2.8.2 of \cite{Galambos78}), (%
\ref{eq:evtmax}) implies that extreme value theory also holds jointly for
the first $k\,$order statistics%
\begin{equation}
n^{-\xi }\mathbf{W}^{R}=n^{-\xi }\left( 
\begin{array}{c}
W_{1}^{R} \\ 
\vdots \\ 
W_{k}^{R}%
\end{array}%
\right) \Rightarrow \sigma \mathbf{X}=\sigma \left( 
\begin{array}{c}
X_{1} \\ 
\vdots \\ 
X_{k}%
\end{array}%
\right) .  \label{evt_con}
\end{equation}%
The distribution of $\mathbf{X}$ satisfies $\{X_{j}^{-1/\xi
}\}_{j=1}^{k}\sim \{\sum_{l=1}^{j}E_{l}\}_{j=1}^{k}$, where $E_{l}$ are
i.i.d.~exponential random variables.

Since the new theoretical results of this paper concern rates of
convergence, a suitable strengthening of the approximate Pareto tail
assumption (\ref{evt1})\ is needed. \cite{Falk04} define the $\delta $%
-neighborhood of the Pareto distribution with index $\xi $ as follows.

\begin{condition}
\label{cnd:pareto_dens}For some $\delta ,w_{0}>0$, $F$ admits a density for $%
w>w_{0}$ of the form%
\begin{equation}
f(w)=(\xi \sigma )^{-1}(\frac{w}{\sigma })^{-1/\xi -1}(1+h(w))
\label{eq:cond1}
\end{equation}%
with $|h(w)|$ uniformly bounded by $Cw^{-\delta /\xi }$ for some finite $C.$
\end{condition}

Theorem 5.5.5 of \cite{Reiss89} shows that under Condition 1, (\ref{evt_con}%
) provides accurate approximations in the sense that%
\begin{equation}
\sup_{B}|\mathbb{P}(n^{-\xi }\mathbf{W}^{R}\in B)-\mathbb{P}(\sigma \mathbf{X%
}\in B)|=O(n^{-\delta })  \label{tvd_conv}
\end{equation}%
for $\delta \leq 1$, where the supremum is taken over all Borel sets $%
B\subseteq 
\mathbb{R}
^{k}$.

Many heavy-tailed distributions satisfy Condition 1: for the right tail of a
student-t distribution with $\nu $ degrees of freedom, $\xi =1/\nu $ and $%
\delta =2\xi $, for the tail of a Fr\'{e}chet or generalized extreme value
distribution with parameter $\alpha $, $\xi =1/\alpha $ and $\delta =1,$ and
for an exact Pareto tail, $\delta $ may be chosen arbitrarily large. But
there also exist heavy-tailed distributions in the domain of attraction of a
Fr\'{e}chet limit law that do not satisfy Condition 1, such as density of
the form (\ref{eq:cond1}) with $h(x)=1/\log (1+x)$, for example. Under some
additional regularity conditions, Theorem 3.2 of \cite{Falk93} shows
Condition 1 to be necessary to obtain an error rate of extreme value
approximations of order $n^{-\delta }$ for $\delta >0$. Roughly speaking,
Condition 1 thus formalizes the assumption that extreme value theory
provides accurate approximations.

\subsection{\label{sec:Treview}Approximations to the t-Statistic}

Let $T_{n}$ be the usual t-statistic computed from an i.i.d.~sample $%
W_{1},\ldots ,W_{n}$, where $W_{i}\sim W$. If $\mathbb{E}[W]=0$ and $\mathbb{%
E}[W^{2}]<\infty $, then $T_{n}\Rightarrow \mathcal{N}(0,1).$ A seminal
paper by \cite{Bentkus93} establishes a bound on the rate of this
convergence which does not require the third moment of $W$ to exist. In
particular, \cite{Bentkus93} show that for some $C>0$ that does not depend
on $F$, and $\mathbb{E}[W^{2}]=1$,%
\begin{equation}
\sup_{t}|\mathbb{P}(T_{n}<t)-\Phi (t)|\leq C\mathbb{E}[W^{2}\mathbf{1}%
[|W|>n^{1/2}]]+Cn^{-1/2}\mathbb{E}[|W|^{3}\mathbf{1}[|W|\leq n^{1/2}]]
\label{BG_bound}
\end{equation}%
where $\Phi (t)=\mathbb{P}(Z<t)$, $Z\sim \mathcal{N}(0,1)$. (The explicit
claim of uniformity with respect to $F$ is made in \cite{Bentkus96}.) This
result is key for the new theoretical results of this paper.

To put (\ref{BG_bound}) into perspective, recall that the classic
Berry-Esseen bound shows a $n^{-1/2}$ rate of the approximation $%
Z_{n}=n^{-1/2}\sum_{i=1}^{n}W_{i}/\sqrt{\mathbb{E}[W^{2}]}\Rightarrow 
\mathcal{N}(0,1)$, with a constant that involves the third moment of $W$. If 
$\mathbb{E}[|W|^{3}]$ does not exist, then the classic Berry-Esseen bound is
inapplicable. Theorem 5 of \cite{Petrov75} provides the analogue of (\ref%
{BG_bound}) with $T_{n}$ replaced by $Z_{n}$. The contribution of \cite%
{Bentkus93} is thus to show that the random norming inherent in the
t-statistic does not affect Petrov's (1975) result.

Subsequent research by \cite{Hall04} provides a sharp bound on the rate of
convergence: If $\mathbb{E}[W^{2}]<\infty $, their results imply that%
\begin{equation}
\frac{\sup_{t}|\mathbb{P}(T_{n}<t)-\Phi (t)|}{n\mathbb{P}%
(|W|>n^{-1/2})+n^{1/2}|\mathbb{E}[WA_{n}]|+n^{-1/2}\mathbb{E}%
[|W|^{3}A_{n}]+n^{-1}\mathbb{E}[|W|^{4}A_{n}]}  \label{HW_bound}
\end{equation}%
with $A_{n}=\mathbf{1}[|W|\leq n^{1/2}]$ is bounded away from zero and
infinity uniformly in $n$.

A final relevant result from the literature concerns the bootstrap
approximation to the distribution of the t-statistic. Let $\mathbf{W}%
=(W_{1},\ldots ,W_{n})$, and let $T_{n}^{\ast }$ be a bootstrap draw of $%
T_{n}$ from the demeaned empirical distribution of $W_{i}$, conditional on $%
\mathbf{W}$. \cite{Bloznelis03} show that if $F$ is non-lattice and $\mathbb{%
E}[|W|^{3}]<\infty $, then%
\begin{equation}
\sup_{t}|\mathbb{P}(T_{n}^{\ast }<t|\mathbf{W})-\mathbb{P}%
(T_{n}<t)|=o(n^{-1/2})\text{\ a.s.}  \label{boot3}
\end{equation}%
while, for $\mathbb{E}[W^{3}]\neq 0$, $\lim \inf_{n\rightarrow \infty
}n^{1/2}\sup_{t}|\mathbb{P}(T_{n}<t)-\Phi (t)|>0$. In other words, as long
as $W$ has finite non-zero third moment, the error in the bootstrap
approximation to the distribution of the t-statistic is of smaller order
than the normal approximation, and the bootstrap provides a refinement over
the usual t-test.

\section{New Theoretical Results}

To ease exposition, we focus in this section on the case where the left tail
of $W$ is light, as in the introduction. The analogous results also hold
when both tails are moderately heavy with tail index smaller than $1/2$; we
provide an analogue of Theorem \ref{thm_main} in Appendix \ref{sec:gen2tails}%
.

\subsection{\label{sec:bootThm}Properties of Bootstrapped t-Statistic under $%
1/3<\protect\xi <1/2$}

\begin{theorem}
\label{thm:boot}Suppose (\ref{evt1})\ holds for $1/3<\xi <1/2$, and $%
\int_{-\infty }^{0}|w|^{3}dF(w)<\infty $. Then under $\mathbb{E}[W]=0$

(a) $\lim \inf_{n\rightarrow \infty }n^{1/(2\xi )-1}\sup_{t}|\mathbb{P}%
(T_{n}<t)-\Phi (t)|>0$ and

(b) $n^{3(1/2-\xi )}\sup_{t}|\mathbb{P}(T_{n}^{\ast }<t|\mathbf{W})-\Phi
(t)|=O_{p}(1).$
\end{theorem}

Since $3(1/2-\xi )>1/(2\xi )-1$ for $1/3<\xi <1/2$, the triangle inequality
implies that $\sup_{t}|\mathbb{P}(T_{n}^{\ast }<t|\mathbf{W})-\mathbb{P}%
(T_{n}<t)|=O_{p}(n^{1-1/(2\xi )})$, so Theorem \ref{thm:boot} shows that the
bootstrap does not provide a refinement if the underlying population has
between two and three moments, at least as long as the population has an
approximate Pareto tail. This result is apparently new, but it is not
difficult to prove. From Markov's inequality, $\int_{-\infty
}^{0}|w|^{3}dF(w)<\infty $ implies that also $|W|$ has a Pareto tail with
index $1/3<\xi <1/2$ in the sense of (\ref{evt1}). Part (a) now simply
follows from evaluating the sharp bound on the rate of convergence in (\ref%
{HW_bound}). Part (b) follows from applying the \cite{Bentkus93} bound (\ref%
{BG_bound}) to the empirical distribution of $\bar{W}_{i}=W_{i}-n^{-1}%
\sum_{j=1}^{n}W_{j}$: By (\ref{evt_con}), $n^{-\xi }\max_{i}|W_{i}|$
converges in distribution, so $\max_{i}|\bar{W}_{i}|=O_{p}(n^{\xi })$. Since 
$\xi <1/2$, this implies $n^{-1}\sum_{i=1}^{n}\bar{W}_{i}^{2}\mathbf{1}[|%
\bar{W}_{i}|>\sqrt{n}]\overset{p}{\rightarrow }0$. Furthermore, $|W_{i}|^{3}$
has a Pareto tail of index $3\xi >1$. Thus $n^{-3\xi }\sum_{i=1}^{n}|\bar{W}%
_{i}|^{3}$ converges in distribution to a stable distribution (see, for
instance, \cite{LePage81}, who elucidate the connection between extreme
value theory and stable limit laws), so that $n^{-3/2}\sum_{i=1}^{n}|\bar{W}%
_{i}|^{3}=O_{p}(n^{3\xi -3/2})$, and the result follows.

The existence of three moments, corresponding to a tail index of $\xi <1/3$,
is necessary to obtain the first term of an Edgeworth expansion that
underlies the proof of \cite{Bloznelis03}. More intuitively, recall that
under $\xi <1/3$, the Berry-Esseen bound shows that the central limit
theorem has an approximation quality of order $n^{-1/2}$. Now under (\ref%
{evt1}), $\mathbb{P}(W_{1}^{R}>\sigma \sqrt{n})$ is of order $(1-n^{-1/(2\xi
)})^{n}\approx n^{1-1/(2\xi )}$. Thus, for $\xi >1/3$, the largest
observation is of order $\sqrt{n}$ with a probability that is an order of
magnitude larger than $n^{-1/2}$. Non-normal observations of order $\sqrt{n}$
are not negligible in the central limit theorem, so the rare large values of 
$W_{1}^{R}$ under $\xi >1/3$ are responsible for a deterioration of the
central limit theorem approximation compared to the $\xi <1/3$ case (cf. 
\cite{Hall04}). But from (\ref{eq:evtmax}) $W_{1}^{R}$ is of order $n^{\xi }$
in nearly all samples, so the bootstrap approximation misses this effect,
and systematically underestimates the heaviness of the tail.

\subsection{\label{sec:mainThm}New Asymptotic Approximation}

We first discuss the approximate parametric problem in more detail. Under
the Pareto tail assumption (\ref{evt1}), we find from a straightforward
calculation that for large $w$, $m(w)=-\mathbb{E}[W|W\leq w]\approx \sigma
^{1/\xi }w^{1-1/\xi }/(1-\xi ).$ Let $s_{n}^{2}=(n-k)^{-1}%
\sum_{i=1}^{n-k}(W_{i}^{s}-\bar{W}^{s})^{2}$ be the usual variance estimator
from the $n-k$ smallest observations. With $k$ fixed, $s_{n}^{2}$ still
converges in probability to the unconditional variance of $W$, $s_{n}^{2}%
\overset{p}{\rightarrow }\limfunc{Var}[W]$. Since the ultimate test we
derive is scale invariant, it is without loss of generality to normalize $%
\limfunc{Var}[W]=1$. From the convergence to the joint extreme value
distribution in (\ref{evt_con}), $n^{-\xi }\mathbf{W}^{R}\overset{a}{\sim }%
\sigma \mathbf{X}$, where we write $\overset{a}{\sim }$ for
\textquotedblleft is approximately distributed as.\textquotedblright\
Furthermore, under local alternatives $\mathbb{E}[W]=n^{-1/2}\mu ,$ the
t-statistic 
\begin{equation*}
T_{n}^{s}=\frac{\sum_{i=1}^{n-k}W_{i}^{s}}{\sqrt{(n-k)s_{n}^{2}}}
\end{equation*}%
computed from $\{W_{i}^{s}\}_{i=1}^{n-k}$ is approximately normal with mean $%
\mu -n^{-1/2}m(W_{k}^{R})\approx \mu -n^{-1/2}\sigma ^{1/\xi
}(W_{k}^{R})^{1-1/\xi }/(1-\xi )$. Combining these two approximations yields%
\begin{equation}
\mathbf{Y}_{n}=\left( 
\begin{array}{c}
\mathbf{W}^{R}/\sqrt{(n-k)s_{n}^{2}} \\ 
T_{n}^{s}%
\end{array}%
\right) \overset{a}{\sim }\left( 
\begin{array}{c}
\eta _{n}\mathbf{X} \\ 
Z+\mu -\eta _{n}\frac{1}{1-\xi }X_{k}^{1-1/\xi }%
\end{array}%
\right) =\mathbf{Y}_{n}^{\ast }  \label{Y0_Yn}
\end{equation}%
with $\eta _{n}=\sigma n^{-(1/2-\xi )}$ and $Z\sim \mathcal{N}(0,1)$
independent of $\mathbf{X}$.\ The first $k$ elements of $\mathbf{Y}_{n}$ are
the largest $k$ order statistics divided by the denominator of the $k+1$
element $T_{n}^{s}$, so that $\mathbf{Y}_{n}$ is invariant to changes in
scale $\{W_{i}\}_{i=1}^{n}\rightarrow \{cW_{i}\}_{i=1}^{n}$ for $c>0$. The
approximate parametric model on the right-hand side of (\ref{Y0_Yn}) treats
these as jointly extreme value with scale $\eta _{n}$ and tail index $\xi $,
and conditionally normally distributed with some (negative) mean that is a
function of $X_{k}$ and the parameters $\eta _{n}$ and tail index $\xi $
under $\mu =0$.

As discussed in the introduction, the central idea of this paper is to use
the parametric model $\mathbf{Y}_{n}^{\ast }$ to determine a level $\alpha $
test $\varphi :%
\mathbb{R}
^{k+1}\mapsto \{0,1\}$ of $H_{0}:\mu =0$ that satisfies $\mathbb{E}[\varphi (%
\mathbf{Y}_{n}^{\ast })]\leq \alpha $ by construction for all $\xi <1/2$, at
least for all $n\geq n_{0}$ and some appropriate upper bounds on $\sigma $.
We discuss the construction of such tests in the next section. Any such test 
$\varphi $ may then be applied to the left-hand side of (\ref{Y0_Yn}), $%
\varphi (\mathbf{Y}_{n})$, to test $H_{0}:\mathbb{E}[W]=0$ from the
observations $W_{1},\ldots ,W_{n}$.

Our main theoretical result is the following.

\begin{theorem}
\label{thm_main}For $k>1$, let $r_{k}(\xi )=\frac{3(1+k)(1-2\xi )}{%
2(1+k+2\xi )}$. Suppose Condition 1 holds with $\delta \geq r_{k}(\xi )$, $%
\int_{-\infty }^{0}|w|^{p}dF(w)<\infty $ for all $p>0$ and that $\varphi :%
\mathbb{R}
^{k+1}\mapsto \{0,1\}$ is such that for some finite $m_{\varphi }$, $\varphi
:%
\mathbb{R}
^{k+1}\mapsto \{0,1\}$ can be written as an affine function of $\{\varphi
_{j}\}_{j=1}^{m_{\varphi }}$, where each $\varphi _{j}$ is of the form 
\begin{equation*}
\varphi _{j}(\mathbf{y},y_{0})=\mathbf{1}[\mathbf{y}\in \mathcal{H}_{j}]%
\mathbf{1}[y_{0}\leq b_{j}(\mathbf{y})]
\end{equation*}%
with $b_{j}:%
\mathbb{R}
^{k}\mapsto 
\mathbb{R}
$ a Lipschitz continuous function and $\mathcal{H}_{j}$ a Borel measurable
subset of $%
\mathbb{R}
^{k}$ with boundary $\partial \mathcal{H}_{j}$. For $\mathbf{u}%
=(1,u_{2},\ldots ,u_{k})^{\prime }\in 
\mathbb{R}
^{k}$ with $1\geq u_{2}\geq u_{3}\geq \ldots \geq u_{k}$, let $\mathcal{I}%
_{j}(\mathbf{u})=\{s>0:s\mathbf{u}\in \partial \mathcal{H}_{j}\}$. Assume
further that for some $L>0$, and Lebesgue almost all $\mathbf{u},$ $\mathcal{%
I}_{j}(\mathbf{u})$ contains at most $L$ elements in the interval $%
[L^{-1},\infty )$.

Then under $H_{0}:\mu =0$, for $\frac{1+k}{1+3k}<\xi <1/2$ and any $\epsilon
>0$ 
\begin{equation*}
|\mathbb{E}[\varphi (\mathbf{Y}_{n})]-\mathbb{E}[\varphi (\mathbf{Y}%
_{n}^{\ast })]|\leq Cn^{-r_{k}(\xi )+\epsilon }.
\end{equation*}
\end{theorem}

Recall from Theorem \ref{thm:boot} (a) above that the rate of convergence of
the normal approximation to the distribution of the t-statistic is $%
n^{1/(2\xi )-1}$. Since for $\frac{1+k}{1+3k}<\xi <1/2$, $r_{k}(\xi
)>1/(2\xi )-1$, the theorem shows that the difference in the rejection rates
of $\varphi $ in the parametric model $\mathbb{E}[\varphi (\mathbf{Y}%
_{n}^{\ast })]$ and in the original inference for the mean problem $\mathbb{E%
}[\varphi (\mathbf{Y}_{n})]$ is of smaller order. In this sense, the new
approximation provides a refinement for underlying populations that have
between two and three moments.

The \cite{Bentkus93} bound (\ref{BG_bound}) implies that conditional on $%
\mathbf{W}^{R}$, $T_{n}^{s}$ is well approximated by a standard normal
distribution, since the $W_{i}^{s}$ form an i.i.d.~sample from a truncated
distribution with less heavy tails compared to the original population.
Furthermore, under Condition 1, it follows from (\ref{tvd_conv}) that the
distribution of $\mathbf{W}^{R}/\sqrt{(n-k)}$ is well approximated by the
distribution of $\eta _{n}\mathbf{X}$. The difficulty in the proof of
Theorem \ref{thm_main} arises from the presence of $s_{n}^{2}$ in the scale
normalization of $\mathbf{W}^{R}$ in $\mathbf{Y}_{n}$. While it is easy to
show that $s_{n}^{2}\overset{p}{\rightarrow }\limfunc{Var}[W]=1$, the proof
of Theorem \ref{thm_main} requires this convergence to be sufficiently fast,
and this complication leads to the presence of $k$ in the rate $r_{k}$
(intuitively, larger $k$ lead to more truncation, so $s_{n}^{2}$ is
estimated from a distribution with a lighter tail), and the technical
requirements on the form of $\varphi $.

Note that the usual full sample t-statistic of $H_{0}:\mathbb{E}[W]=0$, $%
T_{n}$, is approximated in terms of $\mathbf{Y}_{n}=(Y_{1n},\ldots
,Y_{kn},T_{n}^{s})^{\prime }$ by 
\begin{equation}
\tilde{T}_{n}=\frac{T_{n}^{s}+\sum_{i=1}^{k}Y_{in}}{\sqrt{%
1+\sum_{i=1}^{k}Y_{in}^{2}}}  \label{Ttilda}
\end{equation}%
up to a $O_{p}(n^{-1/2})$ term. Under (\ref{Y0_Yn}), the distribution of $%
\tilde{T}_{n}$ is approximated by 
\begin{equation}
T(\mathbf{Y}_{n}^{\ast })=\frac{Z+\mu -\eta _{n}\frac{1}{1-\xi }%
X_{k}^{1-1/\xi }+\eta _{n}\sum_{i=1}^{k}X_{i}}{\sqrt{1+\eta
_{n}^{2}\sum_{i=1}^{k}X_{i}^{2}}}.  \label{Tsingle}
\end{equation}%
Application of Theorem \ref{thm_main} with $\varphi (\mathbf{Y}_{n}^{\ast })=%
\mathbf{1}[T(\mathbf{Y}_{n}^{\ast })<t]$ shows that this approximation has a
faster rate of convergence compared to the usual standard normal
approximation. \cite{Mueller19} shows that one can combine extreme value
theory to improve the rates of approximation to sums of i.i.d.~random
variables compared to the central limit theorem under $\xi >1/3$. On
implication of Theorem \ref{thm_main} is thus a corresponding result for the
case of self-normalized sums (\ref{Ttilda})\ and (\ref{Tsingle}).

In principle, one could use this implication also to construct an
alternative test $\varphi $ that simply amounts to a t-test with
appropriately increased critical value to ensure size control in the
approximate model, $\mathbb{E}[\varphi (\mathbf{Y}_{n}^{\ast })]\leq \alpha $%
. This is woefully inefficient, however, since the much larger critical
value is only needed for samples where $\xi $ and $\eta _{n}$ are large,
which would defeat the objective of obtaining a test that remains close to
efficient for populations with thin tails.

\section{Construction of a New Test}

\subsection{Generalized Parametric Model}

To obtain accurate approximations in small samples also for potentially
thin-tailed distributions, it makes sense to extend the parametric
approximation to populations with an approximate generalized Pareto tail.
The c.d.f. $F$ of such populations satisfies%
\begin{equation}
F(w)\approx 1-(1+\xi (w/\sigma -\nu ))^{-1/\xi }\text{, \ \ \ }\xi \in
(-\infty ,1/2]  \label{gp_ass}
\end{equation}%
for all $w$ close to the upper bound of the support of $F$, and here and in
the following, expressions of the form $(1+\xi x)^{-1/\xi }$ are understood
to equal $e^{-x}$ for $\xi =0$. The Pareto tail assumption (\ref{evt1}) of
Section \ref{sec:evt} is recovered as a special case for $\xi >0$ with $\nu
=1/\xi $ and $\sigma $ rescaled by $\xi $.\footnote{%
To avoid notational clutter, this section redefines some of the notation
previously introduced in Sections 2 and 3 as appropriate for the more
general model.}

Assumption (\ref{gp_ass})\ accommodates infinite support thin-tailed
distributions, such as the exponential distribution, with $\xi =0$, as well
as distributions with finite upper bound on their support, such as the
uniform distribution with $\xi =-1$. From the seminal work of \cite%
{Balkema74} and \cite{Pickands75} (also see Theorem 5.1.1 of \cite{Reiss89}%
), it follows that under an appropriate formalization of (\ref{gp_ass}),
there exist real sequences $a_{n}$ and $\kappa _{n}$ such that 
\begin{equation}
\frac{\mathbf{W}^{R}}{a_{n}}-\kappa _{n}\Rightarrow \mathbf{X}=(X_{1},\ldots
,X_{k})^{\prime }  \label{gp_conv}
\end{equation}%
is (jointly) generalized extreme value distributed, so that $\{(\xi
X_{j}+1)^{-1/\xi }\}_{j=1}^{k}\sim \{\sum_{l=1}^{j}E_{l}\}_{j=1}^{k}$ with $%
E_{l}$ i.i.d.~exponential random variables. If $F$ is exactly generalized
Pareto in the sense of (\ref{gp_ass}), then Corollary 1.6.9 of \cite{Reiss89}%
\ implies%
\begin{equation}
\left\{ \left( \xi \left( \frac{W_{j}^{R}}{a_{n}}-\kappa _{n}\right)
+1\right) ^{-1/\xi }\right\} _{j=1}^{k}\sim \left\{ \left( \frac{n}{%
\sum_{l=1}^{n+1}E_{l}}\right) \sum_{l=1}^{j}E_{l}\right\} _{j=1}^{k}
\label{exact_exponential}
\end{equation}%
with $a_{n}=\sigma n^{\xi }$ and $\xi \kappa _{n}=1+n^{-\xi }(\xi \nu -1)$,
so that $\sum_{l=1}^{n+1}E_{l}/n\approx 1$ is the only approximation
involved in (\ref{gp_conv}). If only the right tail of $F$ of mass $p_{R}>0$
is exactly generalized Pareto, then (\ref{exact_exponential}) holds
conditionally on the event $\sum_{l=1}^{j}E_{l}/\sum_{l=1}^{n+1}E_{l}\leq
np_{R}$, whose probability is larger than 99\% for $k=8$ and all $p_{R}\geq
16/n$, $n\geq 50$.

Under (\ref{gp_ass}) and (\ref{gp_conv}), from the same logic that led to (%
\ref{Y0_Yn}), we obtain the approximate model%
\begin{equation}
\mathbf{Y}_{n}=\left( 
\begin{array}{c}
\mathbf{W}^{R}/\sqrt{(n-k)s_{n}^{2}} \\ 
T_{n}^{s}%
\end{array}%
\right) \overset{a}{\sim }\left( 
\begin{array}{c}
\eta _{n}(\mathbf{X}+\kappa _{n}\mathbf{e}) \\ 
Z+\mu -\eta _{n}m^{\ast }(\mathbf{X},\kappa _{n},\xi )%
\end{array}%
\right) =\mathbf{Y}_{n}^{\ast }  \label{Y0_Yn_GEV}
\end{equation}%
where $\mathbf{e}$ is a $k\times 1$ vector of ones, $\eta _{n}=n^{-1/2}a_{n}$
and%
\begin{equation*}
m^{\ast }(\mathbf{X},\kappa _{n},\xi )=(1+\xi X_{k})^{-1/\xi }\left( \kappa
_{n}+\frac{1+\xi X_{k}}{\xi (1-\xi )}-\frac{1}{\xi }\right) .
\end{equation*}

With $\kappa _{n}\rightarrow 1/\xi $ for $\xi >0$, it is tempting to employ
the additional approximation $\kappa _{n}=1/\xi $ to eliminate the location
parameter in (\ref{Y0_Yn_GEV}), and this is implicitly applied in standard
extreme value theory as reviewed in Section \ref{sec:evt}. However, unless $%
n $ is very large, this leads to a considerably deterioration of the
approximation in (\ref{gp_conv}), and hence (\ref{Y0_Yn_GEV}), so we do not
do so in the following.

For practical implementations it is important to allow for the possibility
that both tails are potentially moderately heavy. This is straightforward
under an assumption that also the left-tail of $F$ is approximately
generalized Pareto in the sense of (\ref{gp_ass}): Let $\mathbf{W}^{L}$ be
the set of smallest $k$ order statistics. Further let $W_{i}^{m}$ be the $%
n-2k$ \textquotedblleft middle\textquotedblright\ order statistics $%
k+1,\ldots ,n-k-1$, and let $s_{n}^{2}$ be the sample variance of $W_{i}^{m}$%
. Then in analogy to (\ref{Y0_Yn_GEV}),%
\begin{equation}
((n-2k)s_{n}^{2})^{-1/2}\left( 
\begin{array}{c}
\mathbf{W}^{R} \\ 
-\mathbf{W}^{L} \\ 
\sum_{i=1}^{n-2k}W_{i}^{m}%
\end{array}%
\right) \overset{a}{\sim }\left( 
\begin{array}{c}
\eta _{n}^{R}(\mathbf{X}^{R}+\kappa _{n}^{R}\mathbf{e}) \\ 
\eta _{n}^{L}(\mathbf{X}^{L}+\kappa _{n}^{L}\mathbf{e}) \\ 
Z-\eta _{n}^{R}m^{\ast }(\mathbf{X}^{R},\kappa _{n}^{R},\xi ^{R})+\eta
_{n}^{L}m^{\ast }(\mathbf{X}^{L},\kappa _{n}^{L},\xi ^{L})%
\end{array}%
\right) =\mathbf{Y}_{n}^{\ast }.  \label{Y_twosided}
\end{equation}%
where $\mathbf{X}^{L}$ and $\mathbf{X}^{R}$ are independent and generalized
extreme value distributed with tail index $\xi ^{L}$ and $\xi ^{R}$,
respectively, and independent of $Z\sim \mathcal{N}(0,1)$.

The scale and location parameters $\eta _{n}$ and $\kappa _{n}$ in this
generalized model depend on the known sample size $n$. But they also depend
on the tail parameters of the underlying population: Recall that $\eta
_{n}=n^{-1/2}a_{n}=\sigma n^{\xi -1/2}$ in (\ref{exact_exponential}). With $%
\sigma $ unknown, this product can in principle take on any positive value,
even with $(n,\xi )$ known, and the same holds for the parameter $\kappa
_{n} $. For this reason, we will now drop the index $n$ in the nuisance
parameter $\theta =(\kappa ^{L},\eta ^{L},\xi ^{L},\kappa ^{R},\eta ^{R},\xi
^{R})\in \Theta _{0}$ and in the $2k+1$ dimensional observation $\mathbf{Y}%
^{\ast }\mathbf{=Y}_{n}^{\ast }$ from the approximate parametric model in (%
\ref{Y_twosided}). In this notation, the problem becomes the construction of
a powerful test $\varphi (\mathbf{Y}^{\ast })$ of $H_{0}:\mu =0$ against $%
H_{a}:\mu \neq 0$ that satisfies 
\begin{equation}
\sup_{\theta \in \Theta _{0}}\mathbb{E}_{\theta }[\varphi (\mathbf{Y}^{\ast
})]\leq \alpha ,  \label{gen_tst_prob}
\end{equation}%
where 
\begin{equation}
\mathbf{Y}^{\ast }=\left( 
\begin{array}{c}
\eta ^{R}(\mathbf{X}^{R}+\kappa ^{R}\mathbf{e}) \\ 
\eta ^{L}(\mathbf{X}^{L}+\kappa ^{L}\mathbf{e}) \\ 
Z+\mu -\eta ^{R}m^{\ast }(\mathbf{X}^{R},\kappa ^{R},\xi ^{R})+\eta
^{L}m^{\ast }(\mathbf{X}^{L},\kappa ^{L},\xi ^{L})%
\end{array}%
\right) =\left( 
\begin{array}{c}
\mathbf{Y}^{R} \\ 
\mathbf{Y}^{L} \\ 
Y_{0}%
\end{array}%
\right)  \label{Z0}
\end{equation}%
and $\mathbf{Y}^{J}=(Y_{1}^{J},\ldots ,Y_{k}^{J})^{\prime }$ for $J\in
\{L,R\}$.

From the representation of the joint generalized extreme value distribution
in terms of i.i.d.~exponentially distributed random variables, it follows
that the density of $\mathbf{Y}^{\ast }$ is given by%
\begin{equation}
f(\mathbf{y}^{\ast }|\theta ,\mu )=f_{T}(\mathbf{y}^{R}|\theta ^{R})f_{T}(%
\mathbf{y}^{L}|\theta ^{L})\phi (y_{0}-\mu +M^{\ast }(\mathbf{y}^{R},\theta
^{R})-M^{\ast }(\mathbf{y}^{L},\theta ^{L}))  \label{dens_joint}
\end{equation}%
where $\theta ^{J}=(\kappa ^{J},\eta ^{J},\xi ^{J})$, $\phi $ is the density
of a standard normal, $M^{\ast }(\mathbf{y},\theta ^{S})=\eta m^{\ast }(%
\mathbf{y/}\eta -\mathbf{e}\kappa ,\kappa ,\xi )$, and $f_{T}$ is the
\textquotedblleft tail\textquotedblright\ density 
\begin{equation*}
f_{T}(\mathbf{y}|\theta ^{S})=\mathbf{1}[1+\xi x_{k}>0]\mathbf{1}[1+\xi
x_{1}>0]\eta ^{-k}\exp \left[ -(1+\xi x_{k})^{-1/\xi }-(1+\xi
^{-1})\sum_{i=1}^{k}\log (1+\xi x_{i}))\right]
\end{equation*}%
with $\theta ^{S}=(\kappa ,\eta ,\xi )$ the parameter of a \textquotedblleft
single tail\textquotedblright\ and $x_{i}=y_{i}/\eta -\kappa $ in obvious
notation.

\subsection{\label{sec:paraspace}Nuisance Parameter Space}

Allowing for arbitrary values of the location and scale parameters in the
testing problem (\ref{gen_tst_prob}) is not fruitful: An unreasonably large
nuisance parameter space $\Theta _{0}$ leads to excessively conservative
inference, and it renders the computational determination of powerful tests
prohibitively difficult. With that in mind, in the default construction, we
consider a nuisance parameter space $\Theta _{0}$ that is partially
motivated by a desire to obtain good size control in samples from a demeaned
Pareto population when $n\geq n_{0}=50$. In the description of $\Theta _{0}$%
, we refer to the extreme value approximation extended to the most extreme $%
n_{0}$ observations, $\{W_{i}^{J}\}_{i=1}^{n_{0}}\overset{a}{\sim }%
\{Y_{i}^{J}\}_{i=1}^{n_{0}}$ for $J\in \{L,R\}$. As noted in (\ref%
{exact_exponential}), this remains an good approximation for an exact
generalized Pareto population even if $n=n_{0}$.

\renewcommand{\theenumi}{(\alph{enumi})}In particular, for $J\in \{L,R\}$,
we impose

\begin{enumerate}
\item $\xi ^{J}<1/2$

\item $\kappa ^{J}\leq 1/\xi ^{J}$ for $\xi ^{J}>0$

\item $\sum_{i=1}^{n_{0}}\mathbb{E}[Y_{i}^{J}]\geq 0$

\item $\sum_{i=k+1}^{n_{0}-k}\mathbb{E}[(Y_{i}^{J})^{2}]\leq 2.$
\end{enumerate}

Restriction (a) imposes that the tails are such that at least two moments
exists. Restriction (b) says that any potential tail shift is only inward
relative to the non-demeaned Pareto default. Note that arbitrarily large
inward shifts are incompatible with the population having mean zero.
Restriction (c) puts a corresponding lower bound on the inward shift: For
the right tail, it requires that the sum of the largest $n_{0}$ observations
still has positive mean. To motivate restriction (d), note that the
normalization by $s_{n}$ implies that the sum of squared demeaned middle
observations cannot be larger than unity. Ignoring the demeaning, taking
expectations and approximating the distribution of these observations by
again extending the extreme value distribution yields restriction (d) with a
right-hand side of unity. We relax the upper bound to equal 2 to accommodate
approximating errors in this argument.

We further impose cross restrictions between the two tails:

\begin{enumerate} \setcounter{enumi}{4}%

\item%
$\mathbb{E}[Y_{k}^{R}]\geq -\mathbb{E}[Y_{k}^{L}]$

\item%
$\sum_{i=1}^{n_{0}/2}\mathbb{E}[Y_{i}^{L}]>\sum_{i=1}^{n_{0}/2}\mathbb{E}%
[Y_{i}^{R}]$ implies $\mathbb{E}[Y_{n_{0}/2}^{R}]>0$

\item%
$\sum_{i=1}^{n_{0}/2}\mathbb{E}[Y_{i}^{R}]>\sum_{i=1}^{n_{0}/2}\mathbb{E}%
[Y_{i}^{L}]$ implies $\mathbb{E}[Y_{n_{0}/2}^{L}]>0$

\item%
$\sum_{i=k+1}^{n_{0}/2}\mathbb{E}[(Y_{i}^{L})^{2}]+\sum_{i=k+1}^{n_{0}/2}%
\mathbb{E}[(Y_{i}^{R})^{2}]\leq 2$%
\end{enumerate}%
\renewcommand{\theenumi}{\arabic{enumi}}

Restriction (e) amounts to an assumption that the two tails don't overlap.
Under an extended tail assumption up to the most extreme $n_{0}/2$
observations, the middle observations take on values between $%
-Y_{n_{0}/2}^{L}$ and $Y_{n_{0}/2}^{R},$ leading to restrictions (f)-(g)
under the null hypothesis of the overall mean being zero. Finally,
restriction (h) is the analogous version of restriction (d) for each tail.

While restriction (c) involves the extreme value approximation for the most
extreme $n_{0}$ observations, note that this approximation is only used to
motivate a lower bound on $\kappa ^{J}$, and for no other purpose. Consider,
for instance, a sample of size $n=n_{0}=50$ from a mean-zero population with
a Pareto right tail and a uniform left tail, with overall continuous
density. Since the uniform distribution is relatively more spread out
compared to the left-tail of a demeaned Pareto distribution, the right tail
is shifted outward compared to a demeaned Pareto distribution. Thus,
restriction (c) is satisfied for this population, and as long as $k$ is
smaller than $n_{0}/2=25$, the approximate parametric model (\ref{Y_twosided}%
) can still be a good approximation

At the same time, one might argue that if the sample size $n$ is much larger
than $n_{0}$, this default parameter space $\Theta _{0}$ is artificially
large, and more powerful inference could be obtained by suitably reducing
it. Note, however, that for any sample size $n$, the tails could be as large
as they are in a sample of size $n_{0}=50$. For instance, consider a sample
of size $n=5000$ from a population that is a mixture between a point mass at
zero and a demeaned Pareto distribution, with 99\% mass on the point mass at
zero. Then only approximately 50 observations in the sample will be
non-zero, and those follow the demeaned Pareto distribution, so $\Theta _{0}$
is again appropriate, and mechanical reduction of $\Theta _{0}$ as a
function of $n$ leads to a poorly performing test in this problem.

Ultimately, inference about the mean requires a substantial assumption about
the tails, and the stronger the assumptions, the more powerful the potential
inference. The restriction $\Theta _{0}$ as described here is one such
choice, and as will be shown below, it yields informative inference while
maintaining a high degree of robustness under moderately heavy tails.

\subsection{Numerical Determination of Powerful Tests}

Our approach is a variant of the algorithms in \cite{Elliott15}, denoted by
EMW in the following, and \cite{MuellerWatsonLRCov}; see \cite{Mueller20}
for a detailed survey. This approach yields a likelihood ratio-type test of
the form%
\begin{equation}
\varphi ^{\ast }(\mathbf{y}^{\ast })=\mathbf{1}\left[ \frac{f_{a}(\mathbf{y}%
^{\ast })}{\sum_{i=1}^{M}\lambda _{i}f(\mathbf{y}^{\ast }|\theta _{i},0)}>1%
\right]  \label{opt_test}
\end{equation}%
where the $M$ values of $\theta _{i}\in \Theta _{0}$ and associated positive
weights $\lambda _{i}$ are iteratively determined so that the discrete
mixture of $\theta $ taking on the values $\theta _{i}$ with probability $%
\lambda _{i}/\sum_{j=1}^{M}\lambda _{j}$ forms an approximate least
favorable distribution for testing $H_{0}$ against the alternative $H_{a}:$%
\textquotedblleft the density of $\mathbf{Y}^{\ast }$ is $f_{a}$%
\textquotedblright . Here $f_{a}$ is chosen to equal $f_{a}(\mathbf{y}^{\ast
})=\int f(\mathbf{y}^{\ast }|\theta ,\mu )dF_{a}(\theta ,\mu )$ for some
weighting function $F_{a}$ that determines against what kind of alternatives
the resulting test is designed to be particularly powerful.

\subsubsection{Specification of Weighting Function}

We choose $F_{a}(\theta ,\mu )$ to be an improper weighting function with
density that is proportional to 
\begin{equation}
\mathbf{1}[-1/2\leq \xi ^{L}\leq 1/2]\mathbf{1}[-1/2\leq \xi ^{R}\leq
1/2]/(\eta ^{R}\eta ^{L})  \label{F1}
\end{equation}%
so that the implied density on $\mu $, $\kappa ^{L}$ and $\kappa ^{R}$ is
flat. This choice is numerically convenient, as it leads to the product form 
$f_{a}(\mathbf{y}^{\ast })=f_{a}^{S}(\mathbf{y}^{L})f_{a}^{S}(\mathbf{y}%
^{R}) $ with $f_{a}^{S}(\mathbf{y)=}\int_{-1/2}^{1/2}f_{a|\xi }^{S}(\mathbf{y%
}|\xi )d\xi $ and $f_{a|\xi }^{S}$ proportional to the density of the scale
and location maximal invariant considered in \cite{Muller17}. By the same
arguments as employed there, $f_{a|\xi }^{S}$ can be obtained by one
dimensional Gaussian quadrature, and we approximate $f_{a}^{S}$ by an
average of those over a grid of values for $\xi \in \lbrack -1/2,1/2]$.

The lower bound of $-1/2$ on $(\xi ^{L},\xi ^{R})$ in (\ref{F1}) plays no
important role, since for values of $\xi ^{J}$ that imply a thin tail, with
very high probability the test is constrained to be of a form that does not
involve $f_{a}$, as discussed next.

\subsubsection{\label{sec:switching}Switching}

A key ingredient in the algorithm of EMW is an importance sampling estimate
of the null rejection probability $\limfunc{RP}(\theta )=\mathbb{E}_{\theta
}[\varphi ^{c}(\mathbf{Y}^{\ast })]=\int \varphi ^{c}(\mathbf{y}^{\ast })f(%
\mathbf{y}^{\ast }|\theta ,0)d\mathbf{y}^{\ast }$ of a candidate test $%
\varphi ^{c}:%
\mathbb{R}
^{2k+1}\mapsto \{0,1\}$ under $\theta ,$ 
\begin{equation}
\widehat{\limfunc{RP}}(\theta )=N^{-1}\sum_{l=1}^{N}\varphi ^{c}(\mathbf{Y}%
_{(l)}^{\ast })\frac{f(\mathbf{Y}_{(l)}^{\ast }|\theta ,0)}{\bar{f}(\mathbf{Y%
}_{(l)}^{\ast })}  \label{RP_hat}
\end{equation}%
where $\mathbf{Y}_{(l)}^{\ast }$, $l=1,\ldots ,N$ are i.i.d.~draws from the
proposal density $\bar{f}$ (so that by the LLN, $\widehat{\limfunc{RP}}%
(\theta )\rightarrow \mathbb{E}_{\bar{f}}[\varphi ^{c}(\mathbf{Y}^{\ast })f(%
\mathbf{Y}^{\ast }|\theta ,0)/\bar{f}(\mathbf{Y}^{\ast })]=\limfunc{RP}%
(\theta )$ in obvious notation), where an appropriate $\bar{f}$ may be
obtained by the algorithm in \cite{MuellerWatsonLRCov}. Clearly, the larger $%
\Theta _{0}$, the larger the number of importance sampling draws $N$ needs
to be for $\widehat{\limfunc{RP}}$ to be of satisfactory accuracy uniformly
in $\theta \in \Theta _{0}$.

The nuisance parameter space $\Theta _{0}$ of the last section is unbounded:
the restrictions there did not put any lower bound on the scale parameters $%
\eta ^{J}$ or the shape parameters $\xi ^{J}$, $J\in \{L,R\}$. Since the
distribution of $\mathbf{Y}^{\ast }$ is highly informative about the scale
of the tails, it is not possible to obtain uniformly accurate approximations
via $\widehat{\limfunc{RP}}$ over $\Theta _{0}$, even with arbitrary
computational resources. It is therefore necessary to choose the test $%
\varphi $ in a way that does not require a computational check of $\mathbb{E}%
_{\theta }[\varphi (\mathbf{Y}^{\ast })]\leq \alpha $ over the entirety of $%
\Theta _{0}$.

The solution to this challenge suggested in \cite{Elliott15} is to \emph{%
switch} to a default test with known size control under $\Theta _{00}\subset
\Theta _{0}$, where the switching rule is such that the default test is
employed with probability very close to one whenever $\mathbf{Y}^{\ast }$ is
generated from $\Theta _{00}$. If for $J\in \{L,R\}$, $\xi ^{J}$ is very
small or $\kappa ^{J}$ is very large, then the resulting observation \textbf{%
$Y$}$^{J}$ is highly compressed in the sense that $Y_{1}^{J}/Y_{k}^{J}$ is
positive and not much larger than one. Also, if $Y_{1}^{J}$ (and thus the
entire vector $\mathbf{Y}^{J}$) is small, then even if the tail is heavy in
the sense of $Y_{1}^{J}/Y_{k}^{J}$ being large, the tail still only makes a
minor contribution to the overall variation of the data. We operationalize
this by introducing the switching index 
\begin{equation}
\chi (\mathbf{Y}^{J})=\max (0,\min (Y_{1}^{J}-\rho _{1},\mathbf{1}%
[Y_{k}^{J}>0](Y_{1}^{J}/Y_{k}^{J}-1-\rho _{r}))  \label{switch}
\end{equation}%
for positive values of $\rho _{r}$, $\rho _{1}$ close to zero, so that $\chi
(\mathbf{Y}^{J})=0$ implies that either $Y_{1}^{J}$ is small or $%
Y_{1}^{J}/Y_{k}^{J}$ is close to unity. If $\chi (\mathbf{Y}^{J})=0$ for one
tail, but not the other, then the problem is heuristically close to knowing
that only one of the tails is moderately heavy. For example, suppose $\chi (%
\mathbf{Y}^{L})=0$, so the left-tail seems thin. Under approximation (\ref%
{Y0_Yn_GEV}), the sum of all observations that are not in the right tail
equals $\tilde{Y}_{0}^{L}=Y_{0}-\sum_{i=1}^{k}Y_{i}^{L}$, with corresponding
approximate variance equal to $\tilde{V}^{L}=1+\sum_{i=1}^{k}(Y_{i}^{L})^{2}$%
. It hence makes sense to switch to a \textquotedblleft single
tail\textquotedblright\ test $\varphi ^{S}:%
\mathbb{R}
^{k+2}\mapsto \{0,1\}$ that treats $\mathbf{Y}^{R}$ as the extreme
observations from the potentially heavy tail, and $\tilde{Y}_{0}^{L}$ to be
approximately normal with mean $-M^{\ast }(\mathbf{Y}^{R},\theta ^{R})$ and
variance $\tilde{V}^{L}$. In analogy to (\ref{opt_test}), such a test is of
the form%
\begin{equation}
\varphi ^{S}(\mathbf{y}^{R}\mathbf{,y}^{L},y_{0})=\mathbf{1}\left[ \frac{%
f_{a}^{S}(\mathbf{y}^{R})}{\sum_{i=1}^{M^{S}}\lambda _{i}^{S}f^{S}(\mathbf{y}%
^{R}\mathbf{,y}^{L},y_{0}|\theta _{i}^{S})}>1\right]  \label{testS}
\end{equation}%
where $\theta _{i}^{S}\in 
\mathbb{R}
^{3}$ and $\lambda _{i}^{S}>0$ form again a numerically determined
approximate least favorable distribution, and 
\begin{equation}
f^{S}(\mathbf{y}^{R}\mathbf{,y}^{L},y_{0}|\theta _{i}^{S})=f_{T}(\mathbf{y}%
|\theta ^{S})\phi ((\tilde{y}_{0}^{L}+M^{\ast }(\mathbf{y}^{R},\theta ^{S})/%
\sqrt{\tilde{v}^{L}})/\sqrt{\tilde{v}^{L}}.  \label{densSapprox}
\end{equation}%
The corresponding test for a thin right tail is given by $\varphi ^{S}(%
\mathbf{y}^{L}\mathbf{,y}^{R},-y_{0})$.

If both tails seem thin, then one would expect that the distribution of the
analogue to the full-sample t-statistic (cf. (\ref{Tsingle}) from Section %
\ref{sec:mainThm}) 
\begin{equation}
T(\mathbf{Y}^{\ast })=\frac{Y_{0}+\sum_{i=1}^{k}Y_{i}^{R}-%
\sum_{i=1}^{k}Y_{i}^{L}}{\sqrt{1+\sum_{i=1}^{k}(Y_{i}^{R})^{2}+%
\sum_{i=1}^{k}(Y_{i}^{L})^{2}}}  \label{Tboth}
\end{equation}%
to be reasonably well approximated by a standard normal distribution,
especially if $\sum_{i=1}^{k}(Y_{i}^{R})^{2}+\sum_{i=1}^{k}(Y_{i}^{L})^{2}$
is small.

These considerations, a numerical analysis, and the sequential structure of
the eventual algorithm presented in Section \ref{sec:algorithm} below
motivate the restriction of tests $\varphi $ to reject, $\varphi (\mathbf{Y}%
^{\ast })=1$, only if all of the following four conditions hold:

\begin{enumerate}
\item $|T(\mathbf{Y}^{\ast })|>\func{cv}_{T}(\mathbf{Y}^{\ast })$, where $%
\func{cv}_{T}(\mathbf{Y}^{\ast })=w_{\func{cv}}(\mathbf{Y}^{\ast })\func{cv}%
_{\alpha }^{Z}+(1-w_{\func{cv}}(\mathbf{Y}^{\ast }))\func{cv}_{\alpha }^{T}$
with $w_{\func{cv}}($\textbf{$Y$}$^{\ast
})=1/(1+\sum_{i=1}^{k}(Y_{i}^{R})^{2}+\sum_{i=1}^{k}(Y_{i}^{L})^{2})$ and $(%
\func{cv}_{\alpha }^{Z},\func{cv}_{\alpha }^{T})$ the $1-\alpha /2$
quantiles of a standard normal and student-t distribution with degrees of
freedom equal to $80+10\log (\alpha )$, respectively;

\item $\varphi _{\chi }^{S}(\mathbf{Y}^{R},\mathbf{Y}^{L},Y_{0})=1$;

\item $\varphi _{\chi }^{S}(\mathbf{Y}^{L},\mathbf{Y}^{R},-Y_{0})=1$;

\item the \textquotedblleft two tailed\textquotedblright\ test of the form (%
\ref{opt_test}) rejects, $\varphi ^{\ast }(\mathbf{Y}^{\ast })=1$;
\end{enumerate}

where%
\begin{equation}
\varphi _{\chi }^{S}(\mathbf{y}^{R}\mathbf{,y}^{L},y_{0})=\mathbf{1}\left[ 
\frac{\exp [5\chi (\mathbf{y}^{L})]\cdot f_{a}^{S}(\mathbf{y}^{R})}{%
\sum_{i=1}^{M^{S}}\lambda _{i}^{S}f^{S}(\mathbf{y}^{R}\mathbf{,y}%
^{L},y_{0}|\theta _{i}^{S})}>1\right] .  \label{testSchi}
\end{equation}%
The additional term $\exp [5\chi (\mathbf{y}^{L})]$ in (\ref{testSchi})
compared to (\ref{testS}) ensures that conditions 2 and 3 are not binding
whenever the corresponding switching index $\chi (\mathbf{y}^{L})$ is large.
Its continuity in $\mathbf{y}^{L}$ avoids the sharp change of the form of
the rejection region as a function of $\mathbf{y}^{\ast }$ that would be
induced by a simpler hard threshold rule $\varphi _{\chi }^{S}(\mathbf{y}^{R}%
\mathbf{,y}^{L},y_{0})=\mathbf{1}[\chi (\mathbf{y}^{L})=0]\varphi ^{S}(%
\mathbf{y}^{R}\mathbf{,y}^{L},y_{0})$.

Condition 1 implies that $\varphi $ never rejects if the analogue $T(\mathbf{%
Y}^{\ast })$ of the usual t-statistic does not reject; in that sense, we
seek to \textquotedblleft robustify\textquotedblright\ the usual t-statistic
to obtain better size control. Condition 1 has the additional appeal that
sums of the form (\ref{RP_hat}) then effectively only involve $\mathbf{Y}%
_{(l)}^{\ast }$ for which $|T(\mathbf{Y}_{(l)}^{\ast })|\geq \func{cv}_{T}(%
\mathbf{Y}_{(l)}^{\ast })$, with an associated gain in computing speed.

We stress that the definition of a \textquotedblleft thin
tail\textquotedblright\ in (\ref{switch}), the approximate normality of (\ref%
{Tboth})\ and so forth are purely heuristic and do not enter the evaluation
of $\mathbb{E}_{\theta }[\varphi (\mathbf{Y}^{\ast })]$ by the algorithm;
this probability is always computed from the distribution (\ref{Z0}) of $%
\mathbf{Y}^{\ast }$. The heuristics merely motivate the particular form of $%
\varphi $ just described. As discussed, it is not possible to numerically
check that $\mathbb{E}_{\theta }[\varphi (\mathbf{Y}^{\ast })]\leq \alpha $
for \emph{all} $\theta \in \Theta _{0}$. But we employ extensive numerical
analysis to ensure that $\mathbb{E}_{\theta }[\varphi (\mathbf{Y}^{\ast
})]\leq \alpha $ over a very large set, including values of $\theta $ that
lead to the events $\chi (\mathbf{Y}^{J})=0$ for $J\in \{L,R\}$ with
probability close to zero, close to one or in between. The simple form that $%
\varphi $ takes on with very high probability in the remainder of the
parameter space makes it plausible that $\mathbb{E}_{\theta }[\varphi (%
\mathbf{Y}^{\ast })]\leq \alpha $ for all $\theta \in \Theta _{0}$, or at
the least, very nearly so.

\subsubsection{\label{sec:ImpSamp}Recombining Tails in Importance Sampling}

Even though the switching rule of the last subsection reduces the
numerically relevant parameter space to a bounded set, this set still turns
out to be so large that a very large number $N$ of importance sampling draws
are necessary to obtain adequate approximations. The computationally
expensive part in the evaluation of $\widehat{\limfunc{RP}}(\theta )$ in (%
\ref{RP_hat}) for different $\theta $ is the evaluation of $f(\mathbf{Y}%
_{(l)}^{\ast }|\theta ,0)$ (since all $\bar{f}(\mathbf{Y}_{(l)}^{\ast })$
can be computed once and stored).

These evaluations can be dramatically sped up by recombining two
\textquotedblleft single tails\textquotedblright\ in different combinations:
For a given $\theta ^{S}=(\kappa ,\eta ,\xi )$, let $\mathbf{Y}^{e}\in 
\mathbb{R}
^{k+1}$ be an \textquotedblleft extended\textquotedblright\ single tail with
distribution 
\begin{equation*}
\mathbf{Y}^{e}=\left( 
\begin{array}{c}
\eta (\mathbf{X}+\kappa \mathbf{e}) \\ 
Z/\sqrt{2}-\eta m^{\ast }(\mathbf{X},\kappa ,\xi )%
\end{array}%
\right) =\left( 
\begin{array}{c}
\mathbf{Y}^{S} \\ 
Y_{0}^{e}%
\end{array}%
\right)
\end{equation*}%
where $\mathbf{X}$ is distributed as as in (\ref{gp_conv}), independent of $%
Z\sim \mathcal{N}(0,1)$. Denote the density of $\mathbf{Y}^{e}$ by $f^{e}(%
\mathbf{y}^{e}|\theta ^{S})$. Given two independent vectors $\mathbf{Y}%
_{(1)}^{e}$ and $\mathbf{Y}_{(2)}^{e}$ distributed according to $\theta
_{1}^{S}=\theta ^{L}$ and $\theta _{2}^{S}=\theta ^{R}$, respectively, note
that their combination into the \textquotedblleft both
tails\textquotedblright\ observation $(\mathbf{Y}_{(1)}^{S\prime },\mathbf{Y}%
_{(2)}^{S\prime },Y_{0,(1)}^{e}-Y_{0,(2)}^{e})^{\prime }\in 
\mathbb{R}
^{2k+1}$ has the same distribution as $\mathbf{Y}^{\ast }$ in (\ref{Z0}),
since the difference of two independent normals of variance 1/2 is again
standard normal. Thus, with $\mathbf{Y}_{(l)}^{e}$ i.i.d.~draws from a
suitable proposal density $\bar{f}^{e}$, one obtains the alternative
estimator%
\begin{equation}
\widetilde{\limfunc{RP}}(\theta
)=(KN)^{-1}\sum_{k=1}^{K}\sum_{l=1}^{N}\varphi ^{c}((\mathbf{Y}%
_{(l)}^{S\prime },\mathbf{Y}_{(l+k)}^{S\prime
},Y_{0,(l)}^{e}-Y_{0,(l+k)}^{e})^{\prime })\frac{f^{e}(\mathbf{Y}%
_{(l)}^{e}|\theta ^{L})f^{e}(\mathbf{Y}_{(l+k)}^{e}|\theta ^{R})}{\bar{f}%
^{e}(\mathbf{Y}_{(l)}^{e})\bar{f}^{e}(\mathbf{Y}_{(l+k)}^{e})}
\label{smartRP_hat}
\end{equation}%
that recombines each extended single tail with $K$ different other extended
single tails, for a total of $KN$ importance draws. Yet evaluation of (\ref%
{smartRP_hat}) only requires a simple product of the $(K+N)$ values $f^{e}(%
\mathbf{Y}_{(l)}^{e}|\theta ^{S})$ for $\theta ^{S}\in \{\theta ^{L},\theta
^{R}\}$. We let $K=128$ and $N=640,000$ for a total of nearly 82 million
importance sampling draws.

\subsubsection{\label{sec:algorithm}Implementation}

The overall algorithm proceeds in four stages. To describe these stages, let 
$\Theta _{0}^{S}\subset 
\mathbb{R}
^{3}$ be the set of parameters satisfying the constraints (a)-(d) of Section %
\ref{sec:paraspace} on one tail. Let $\Theta _{s}^{S}\subset \Theta _{0}^{S}$
be such that for $\theta ^{S}\in \Theta _{s}^{S}$, the event that the
switching index is zero, $\chi (\mathbf{Y}^{J})=0$, happens with at least
90\% probability, and $\Theta _{ss}^{S}\subset \Theta _{s}^{S}$ be such that 
$\chi (\mathbf{Y}^{J})=0$ with probability of exactly 90\%. Note that $%
\Theta _{s}^{S}$ and $\Theta _{ss}^{S}$ depend on $(\rho _{r},\rho _{1})$ in
(\ref{switch}).

\begin{enumerate}
\item Choose $(\rho _{r},\rho _{1})$ such that $\mathbb{E}_{\theta }[\mathbf{%
1}[|T(\mathbf{Y}^{\ast })|>\func{cv}_{T}(\mathbf{Y}^{\ast })]]\leq \alpha $
for all $\theta =(\theta ^{L\prime },\theta ^{R\prime })^{\prime }$ with $%
\theta ^{L},\theta ^{R}\in \Theta _{ss}^{S}$.

\item Use the algorithm of EMW to numerically determine $\varphi _{\chi
}^{S} $ via $\{\lambda _{i}^{S}\}_{i=1}^{M_{S}}$ and $\{\theta
_{i}^{S}\}_{i=1}^{M_{S}}$ in (\ref{testSchi}) so that 
\begin{equation*}
\mathbb{E}_{\theta }[\mathbf{1}[|T(\mathbf{Y}^{\ast })|>\func{cv}_{T}(%
\mathbf{Y}^{\ast })]\varphi _{\chi }^{S}(\mathbf{Y}^{R},\mathbf{Y}%
^{L},Y_{0})]\leq \alpha
\end{equation*}%
for all $\theta =(\theta ^{L\prime },\theta ^{R\prime })^{\prime }\in \Theta
_{0}$ with $\theta ^{L}\in \Theta _{ss}^{S}$ and $\theta ^{R}\in \Theta
_{0}^{S}\backslash \Theta _{s}^{S}$.

\item Use the algorithm of EMW to determine $\varphi ^{\ast }$ in (\ref%
{opt_test}) via $\{\lambda _{i}\}_{i=1}^{M}$ and $\{\theta _{i}\}_{i=1}^{M}$
so that the overall test $\varphi $ of the form described in Section \ref%
{sec:switching} satisfies $\mathbb{E}_{\theta }[\varphi (\mathbf{Y}^{\ast
})]\leq \alpha $ under all $\theta =(\theta ^{L\prime },\theta ^{R\prime
})^{\prime }\in \Theta _{0}$ for $\theta ^{L},\theta ^{R}\in \Theta
_{0}^{S}\backslash \Theta _{s}^{S}$.

\item Spot-check that $\varphi $ indeed satisfies $\mathbb{E}_{\theta
}[\varphi (\mathbf{Y}^{\ast })]\leq \alpha $ for all $\theta \in \Theta _{0}$%
, including $\theta =(\theta ^{L\prime },\theta ^{R\prime })^{\prime }$ with 
$\theta ^{L},\theta ^{R}\in \Theta _{s}^{S}$.
\end{enumerate}

Note that the parameter set under consideration becomes consecutively larger
in Steps 1-3, and the form imposed on tests described in Section \ref%
{sec:switching} ensures that any potential remaining overrejections of a
stage can be corrected by the subsequent stage, which increases the
numerical stability of the algorithm. Null rejection probabilities are
estimated throughout with the importance sampling estimator of Section \ref%
{sec:ImpSamp}. This estimator has an importance sampling standard error
(appropriately adjusted for the dependence in (\ref{smartRP_hat})) of no
more than $0.05\%$, $0.15\%$ and $0.2\%$ for $\alpha =1\%$, $5\%$, 10\%,
respectively.

We apply this algorithm to the default nuisance parameter space with $%
n_{0}=50$ of Section \ref{sec:paraspace} for $k=8$ and $\alpha \in
\{0.002,0.004,\ldots $, $0.008,0.01,0.02,\ldots $, $0.10,0.12,\ldots $, $%
0.20,0.25,0.30,0.40,0.50\}$, and also for $k=4$ to the larger nuisance
parameter space where $n_{0}=25$ in the constraints of Section \ref%
{sec:paraspace}. We ensure that the 95\% and 99\% level confidence intervals
obtained via test inversion\footnote{%
In the rare samples where test inversion yields disconnected sets, we set
the confidence interval equal to the smallest interval that contains all
non-rejections.} always contain the 90\% and 95\% level intervals,
respectively, and that the p-value is always coherent with the level of the
reported confidence interval by adding the obvious additional constraints to
the form of the tests for $\alpha \neq 0.05$. After trivial modifications
that decrease their rejection probability by an arbitrarily small amount,
these tests satisfy the condition of the two-tailed analogue of Theorem \ref%
{thm_main}; see Appendix \ref{sec:gen2tails} for details. For comparison
purposes, we also generate tests with $k\in \{4,12\}$ for $\alpha \in
\{0.01,0.05\}$ in the default parameter space. For $k=8$ and a given level $%
\alpha $, the computations take about one hour on a modern 24 core
workstation in a Fortran implementation, and about 3 hours for $k=12$. Once
the values for $\{\lambda _{i}^{S},\theta _{i}^{S}\}_{i=1}^{M_{S}}$ and $%
\{\lambda _{i},\theta _{i}\}_{i=1}^{M}$ are determined in this fashion, the
evaluation of the resulting test $\varphi $, as required in applications, is
computationally trivial.

\subsection{\label{sec:GMM}Application to GMM}

Suppose we estimate the parameter vector $\vartheta =(\beta ,\gamma ^{\prime
})^{\prime }\in 
\mathbb{R}
^{q}$ by Hansen's (1982)\nocite{Hansen82} Generalized Method of Moments
using the $r\times 1$ moment condition $\mathbb{E}[g(\vartheta ,z)]=0$ from
data $z_{j}$, $j=1,\ldots ,n_{z}$ and $r\times r$ positive definite
weighting function $\hat{\Psi}$. Suppose further that the data $z_{j}$ is
i.i.d.~across clusters defined by the partition $\{\mathcal{C}%
_{i}\}_{i=1}^{n}$ of $\{j:1\leq j\leq n\}$ (so that $\mathcal{C}_{i}=\{i\}$
and $n=n_{z}$ under i.i.d.~sampling of $z_{j}$). Then as $n\rightarrow
\infty $, under standard regularity conditions, $\hat{\vartheta}=(\hat{\beta}%
,\hat{\gamma}^{\prime })^{\prime }$ satisfies%
\begin{equation}
\sqrt{n}(\hat{\vartheta}-\vartheta )=\left( \Gamma ^{\prime }\Psi \Gamma
\right) ^{-1}\Gamma ^{\prime }\Psi \cdot n^{-1/2}\sum_{i=1}^{n}G_{i}+o_{p}(1)
\label{GMM_asy}
\end{equation}%
where $G_{i}=\sum_{j\in \mathcal{C}_{i}}g(\vartheta ,z_{j})$ are i.i.d., $%
\hat{\Gamma}=n^{-1}\sum_{j=1}^{n_{z}}\partial g(\vartheta ,z_{j})/\partial
\vartheta ^{\prime }|_{\vartheta =\hat{\theta}}\overset{p}{\rightarrow }%
\Gamma $ and $\hat{\Psi}\overset{p}{\rightarrow }\Psi $ with $\Gamma $ and $%
\Psi $ non-stochastic, so that the large sample variability of $\hat{%
\vartheta}$ is entirely driven by the average of i.i.d.~observations $G_{i}$%
. Correspondingly, the standard GMM hypothesis test of $H_{0}:\beta =\beta
_{0}$ is numerically equivalent to the usual t-test of $H_{0}:\beta =\beta
_{0}$ computed from the $n$ observations%
\begin{equation}
\hat{W}_{i}=\hat{\beta}+\iota _{1}^{\prime }(\hat{\Gamma}^{\prime }\hat{\Psi}%
\hat{\Gamma})^{-1}\hat{\Gamma}^{\prime }\hat{G}_{i},\text{ }i=1,\ldots ,n
\label{What_GMM}
\end{equation}%
where $\hat{G}_{i}=\sum_{j\in \mathcal{C}_{i}}g(\hat{\vartheta},z_{j})$ and $%
\iota _{1}$ is the $q\times 1$ vector $(1,0,\ldots ,0)^{\prime }$.

Thus, to the extent that $G_{i}$ follows a distribution with moderately
heavy tails, one would expect that small sample inference is improved by
applying the new test to the observations $\{\hat{W}_{i}\}_{i=1}^{n}$. A
corresponding analytical refinement result analogous to Theorem \ref%
{thm_main} is beyond the scope of this paper.

\section{Small Sample Results}

This section presents six sets of small sample results: two for inference
about the mean from an i.i.d.~sample, two for the difference of population
means from two independent samples, and two for a regression coefficient
with clustered standard errors. In all three cases, the data is either
generated from analytical distributions, or from draws with replacement from
a large data set. We focus on tests of nominal 5\% level in the main text;
results for 1\% level tests are reported in the appendix and exhibit broadly
similar patterns.

\subsection{Inference for the Mean}

\begin{figure}[tbp] 
\begin{center}
\small%
\caption{Population Densities in Monte Carlo Experiments}\bigskip

\includegraphics[width=5.985in,keepaspectratio]{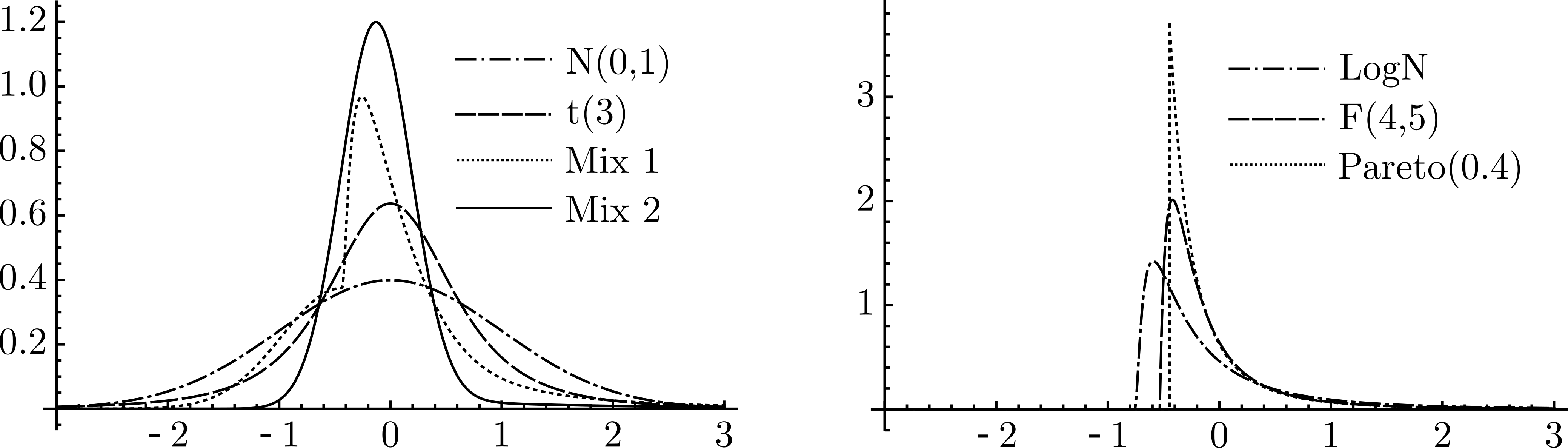}

\end{center}
\begin{small}%
\end{small}%
\end{figure}%
We initially compare our default test with $k=8$ (\textquotedblleft \textsc{%
new default}\textquotedblright ) with three standard tests for the
population mean: standard t-statistic based inference with critical value
from a student-t distribution with $n-1$ degrees of freedom
\textquotedblleft \textsc{t-stat}\textquotedblright ; the percentile-t
bootstrap based on the absolute value of the t-statistic \textquotedblleft 
\textsc{sym-boot}\textquotedblright ; and the percentile-t bootstrap based
on the signed t-statistic \textquotedblleft \textsc{asym-boot}%
\textquotedblright . The data is generated from one of seven populations:
the standard normal distribution N(0,1), the log-normal distribution LogN,
the F-distribution with 4 degrees of freedom in the numerator and 5 in the
denominator F(4,5), the student-t distribution with 3 degrees of freedom
t(3), an equal probability mixture between a N(0,1) and LogN distribution
Mix1, and a 95 / 5 mixture between a N(0,1/25) and a LogN distribution Mix2.
All population distributions are normalized to have mean zero and unit
variance; the corresponding densities are plotted in Figure 1. Technically,
only the Pareto distribution, the t-distribution and the F-distribution
exhibit heavy Pareto like tails in the sense of (\ref{evt_con}) with tail
indices $\xi =0.4$, $\xi =1/3$ and $\xi =0.4$, respectively, but as a
practical matter, also the log-normal and the two mixture distributions are
right-skewed enough to make small sample inference challenging.

\begin{table}[tbp] 
\begin{center}
\small%
\caption{Small Sample Results in Inference for the Mean}

\bigskip

$%
\begin{tabular}{lrrrrrrr}
\hline\hline
& N(0,1) & LogN & F(4,5) & \ \ t(3) & P(0.4) & Mix 1 & Mix 2 \\ \hline
\multicolumn{8}{c}{$n=50$} \\ \hline
\textsc{t-stat} & 5.0\TEXTsymbol{\vert}\textbf{0.99} & 10.0\TEXTsymbol{\vert}%
0.74 & 13.5\TEXTsymbol{\vert}0.65 & 4.7\TEXTsymbol{\vert}\textbf{1.01} & 13.6%
\TEXTsymbol{\vert}0.62 & 7.4\TEXTsymbol{\vert}0.88 & 18.8\TEXTsymbol{\vert}%
0.60 \\ 
\textsc{sym-boot} & 5.0\TEXTsymbol{\vert}\textbf{1.00} & 7.8\TEXTsymbol{\vert%
}1.07 & 10.8\TEXTsymbol{\vert}1.27 & 4.1\TEXTsymbol{\vert}\textbf{1.11} & 
10.6\TEXTsymbol{\vert}1.35 & 6.9\TEXTsymbol{\vert}1.12 & 18.1\TEXTsymbol{%
\vert}1.44 \\ 
\textsc{asym-boot} & 5.2\TEXTsymbol{\vert}\textbf{1.00} & 6.9\TEXTsymbol{%
\vert}0.96 & 8.9\TEXTsymbol{\vert}1.03 & 7.4\TEXTsymbol{\vert}1.06 & 8.6%
\TEXTsymbol{\vert}1.07 & 8.1\TEXTsymbol{\vert}1.02 & 17.6\TEXTsymbol{\vert}%
1.08 \\ 
\textsc{new default} & 3.8\TEXTsymbol{\vert}\textbf{1.10} & 3.2\TEXTsymbol{%
\vert}\textbf{0.93} & 4.5\TEXTsymbol{\vert}\textbf{0.77} & 3.3\TEXTsymbol{%
\vert}\textbf{1.44} & 5.2\TEXTsymbol{\vert}\textbf{0.72} & 3.3\TEXTsymbol{%
\vert}\textbf{1.11} & 12.2\TEXTsymbol{\vert}0.66 \\ \hline
\multicolumn{8}{c}{$n=100$} \\ \hline
\textsc{t-stat} & 4.9\TEXTsymbol{\vert}\textbf{1.00} & 8.2\TEXTsymbol{\vert}%
0.83 & 10.9\TEXTsymbol{\vert}0.73 & 4.6\TEXTsymbol{\vert}\textbf{1.01} & 11.5%
\TEXTsymbol{\vert}0.71 & 6.9\TEXTsymbol{\vert}0.91 & 15.4\TEXTsymbol{\vert}%
0.60 \\ 
\textsc{sym-boot} & 5.0\TEXTsymbol{\vert}\textbf{1.00} & 6.7\TEXTsymbol{\vert%
}1.04 & 9.1\TEXTsymbol{\vert}1.21 & 4.2\TEXTsymbol{\vert}\textbf{1.08} & 9.2%
\TEXTsymbol{\vert}1.19 & 6.4\TEXTsymbol{\vert}1.06 & 14.1\TEXTsymbol{\vert}%
1.17 \\ 
\textsc{asym-boot} & 5.1\TEXTsymbol{\vert}\textbf{1.00} & 6.5\TEXTsymbol{%
\vert}0.97 & 7.5\TEXTsymbol{\vert}1.02 & 6.6\TEXTsymbol{\vert}1.05 & 7.7%
\TEXTsymbol{\vert}1.00 & 7.4\TEXTsymbol{\vert}1.00 & 13.4\TEXTsymbol{\vert}%
0.95 \\ 
\textsc{new default} & 4.8\TEXTsymbol{\vert}\textbf{1.01} & 3.1\TEXTsymbol{%
\vert}\textbf{1.26} & 3.6\TEXTsymbol{\vert}\textbf{1.04} & 3.8\TEXTsymbol{%
\vert}\textbf{1.37} & 3.6\TEXTsymbol{\vert}\textbf{1.00} & 3.3\TEXTsymbol{%
\vert}\textbf{1.31} & 7.9\TEXTsymbol{\vert}0.75 \\ \hline
\multicolumn{8}{c}{$n=500$} \\ \hline
\textsc{t-stat} & 5.0\TEXTsymbol{\vert}\textbf{1.00} & 5.9\TEXTsymbol{\vert}%
\textbf{0.95} & 7.8\TEXTsymbol{\vert}0.87 & 4.8\TEXTsymbol{\vert}\textbf{1.01%
} & 7.9\TEXTsymbol{\vert}0.86 & 5.8\TEXTsymbol{\vert}\textbf{0.97} & 9.6%
\TEXTsymbol{\vert}0.77 \\ 
\textsc{sym-boot} & 5.0\TEXTsymbol{\vert}\textbf{1.00} & 5.4\TEXTsymbol{\vert%
}\textbf{1.01} & 6.9\TEXTsymbol{\vert}1.10 & 4.7\TEXTsymbol{\vert}\textbf{%
1.03} & 6.8\TEXTsymbol{\vert}1.19 & 5.4\TEXTsymbol{\vert}\textbf{1.01} & 8.1%
\TEXTsymbol{\vert}1.04 \\ 
\textsc{asym-boot} & 5.0\TEXTsymbol{\vert}\textbf{1.00} & 5.5\TEXTsymbol{%
\vert}\textbf{1.00} & 7.0\TEXTsymbol{\vert}1.01 & 6.1\TEXTsymbol{\vert}1.02
& 6.4\TEXTsymbol{\vert}1.05 & 6.0\TEXTsymbol{\vert}1.00 & 7.7\TEXTsymbol{%
\vert}0.95 \\ 
\textsc{new default} & 4.9\TEXTsymbol{\vert}\textbf{1.00} & 4.1\TEXTsymbol{%
\vert}\textbf{1.18} & 4.3\TEXTsymbol{\vert}\textbf{1.21} & 4.5\TEXTsymbol{%
\vert}\textbf{1.13} & 4.1\TEXTsymbol{\vert}\textbf{1.22} & 4.4\TEXTsymbol{%
\vert}\textbf{1.18} & 3.2\TEXTsymbol{\vert}\textbf{1.21}%
\end{tabular}%
$

\end{center}
\linespread{1.00}\selectfont
\begin{small}
\linespread{1.00}\selectfont%

Notes: Entries are the null rejection probability in percent, and the
average length of confidence intervals relative to average length of
confidence intervals based on size corrected t-statistic (bold if null
rejection probability is smaller than 6\%) of nominal 5\% level tests. Based
on 20,000 replications.%
\linespread{1.00}\selectfont
\end{small}
\linespread{1.00}\selectfont%
\end{table}%

Table 1 reports null rejection probabilities, along with the average length
of the resulting confidence interval, expressed as a multiple of the average
length of the infeasible confidence interval that is based on the
t-statistic, but applies the size adjusted critical value. As can be seen
from Table 1, the new method comes much closer to controlling size under
moderately heavy-tailed distributions. For the thin-tailed normal
population, the new method only leads to 10\% longer intervals for $n=50$,
and essentially no excessive length for $n\in \{100,500\}$. For other
populations, the intervals of the new method are often much longer than
those from other methods; but since the other methods do not come close to
controlling size, that comparison is not meaningful (entries in bold
indicate where tests are close to valid with a null rejection probability
below 6\%). Remarkably, for $n=50$, the new method yields shorter confidence
intervals than the size corrected t-statistic for some populations while
still controlling size. The explicit modelling of the tails can also yield
efficiency gains, since under a Pareto-like tail, the sample mean is not the
efficient estimator of the population mean.

An exception to the good performance of the new method is the student-t
population with three degrees of freedom. Even though it has fairly heavy
tails, with the third moment not existing, its symmetry enables \textsc{%
t-stat} and \textsc{sym-boot} to control size at much less cost to average
length compared to the new method.\footnote{%
The analytical result by \cite{Bakirov05} shows that the usual 5\% level
t-test remains small sample valid under arbitrary scale mixtures of normals,
which includes all t-distributions.}

\begin{table}[tbp] 
\begin{center}
\small%
\caption{Small Sample Results of New Methods for Inference for the  Mean}

$%
\begin{tabular}{lrrrrrrr}
\hline\hline
& N(0,1) & LogN & F(4,5) & \ \ t(3) & P(0.4) & Mix 1 & Mix 2 \\ \hline
\multicolumn{8}{c}{$n=25$} \\ \hline
\textsc{def}: $k=8,n_{0}=50$ & 4.7\TEXTsymbol{\vert}\textbf{1.00} & 13.1%
\TEXTsymbol{\vert}0.64 & 16.5\TEXTsymbol{\vert}0.56 & 4.0\TEXTsymbol{\vert}%
\textbf{1.02} & 17.8\TEXTsymbol{\vert}0.53 & 8.0\TEXTsymbol{\vert}0.82 & 18.7%
\TEXTsymbol{\vert}0.62 \\ 
$k=4,n_{0}=50$ & 2.7\TEXTsymbol{\vert}\textbf{1.25} & 7.4\TEXTsymbol{\vert}%
0.70 & 10.1\TEXTsymbol{\vert}0.61 & 3.1\TEXTsymbol{\vert}\textbf{1.16} & 11.6%
\TEXTsymbol{\vert}0.57 & 5.2\TEXTsymbol{\vert}\textbf{0.93} & 11.7%
\TEXTsymbol{\vert}0.67 \\ 
$k=12,n_{0}=50$ & \multicolumn{1}{c}{NA} & \multicolumn{1}{c}{NA} & 
\multicolumn{1}{c}{NA} & \multicolumn{1}{c}{NA} & \multicolumn{1}{c}{NA} & 
\multicolumn{1}{c}{NA} & \multicolumn{1}{c}{NA} \\ 
$k=4,n_{0}=25$ & 2.3\TEXTsymbol{\vert}\textbf{1.42} & 4.3\TEXTsymbol{\vert}%
\textbf{0.80} & 5.7\TEXTsymbol{\vert}\textbf{0.67} & 2.1\TEXTsymbol{\vert}%
\textbf{1.49} & 7.3\TEXTsymbol{\vert}0.61 & 3.1\TEXTsymbol{\vert}\textbf{1.09%
} & 9.3\TEXTsymbol{\vert}0.72 \\ \hline
\multicolumn{8}{c}{$n=50$} \\ \hline
\textsc{def}: $k=8,n_{0}=50$ & 3.8\TEXTsymbol{\vert}\textbf{1.10} & 3.2%
\TEXTsymbol{\vert}\textbf{0.93} & 4.5\TEXTsymbol{\vert}\textbf{0.77} & 3.3%
\TEXTsymbol{\vert}\textbf{1.44} & 5.2\TEXTsymbol{\vert}\textbf{0.72} & 3.3%
\TEXTsymbol{\vert}\textbf{1.11} & 12.2\TEXTsymbol{\vert}0.66 \\ 
$k=4,n_{0}=50$ & 3.9\TEXTsymbol{\vert}\textbf{1.15} & 2.6\TEXTsymbol{\vert}%
\textbf{0.99} & 3.6\TEXTsymbol{\vert}\textbf{0.81} & 3.6\TEXTsymbol{\vert}%
\textbf{1.46} & 4.3\TEXTsymbol{\vert}\textbf{0.76} & 3.2\TEXTsymbol{\vert}%
\textbf{1.15} & 11.1\TEXTsymbol{\vert}0.66 \\ 
$k=12,n_{0}=50$ & 3.5\TEXTsymbol{\vert}\textbf{1.15} & 4.3\TEXTsymbol{\vert}%
\textbf{0.86} & 6.1\TEXTsymbol{\vert}0.73 & 2.8\TEXTsymbol{\vert}\textbf{1.43%
} & 8.0\TEXTsymbol{\vert}0.68 & 2.8\TEXTsymbol{\vert}\textbf{1.08} & 10.9%
\TEXTsymbol{\vert}0.67 \\ 
$k=4,n_{0}=25$ & 3.7\TEXTsymbol{\vert}\textbf{1.15} & 2.1\TEXTsymbol{\vert}%
\textbf{1.15} & 2.8\TEXTsymbol{\vert}\textbf{0.94} & 3.2\TEXTsymbol{\vert}%
\textbf{1.70} & 3.2\TEXTsymbol{\vert}\textbf{0.87} & 3.2\TEXTsymbol{\vert}%
\textbf{1.32} & 11.6\TEXTsymbol{\vert}0.73 \\ \hline
\multicolumn{8}{c}{$n=100$} \\ \hline
\textsc{def}: $k=8,n_{0}=50$ & 4.8\TEXTsymbol{\vert}\textbf{1.01} & 3.1%
\TEXTsymbol{\vert}\textbf{1.26} & 3.6\TEXTsymbol{\vert}\textbf{1.04} & 3.8%
\TEXTsymbol{\vert}\textbf{1.37} & 3.6\TEXTsymbol{\vert}\textbf{1.00} & 3.3%
\TEXTsymbol{\vert}\textbf{1.31} & 7.9\TEXTsymbol{\vert}0.75 \\ 
$k=4,n_{0}=50$ & 5.0\TEXTsymbol{\vert}\textbf{1.02} & 3.0\TEXTsymbol{\vert}%
\textbf{1.23} & 3.4\TEXTsymbol{\vert}\textbf{1.00} & 3.8\TEXTsymbol{\vert}%
\textbf{1.51} & 3.8\TEXTsymbol{\vert}\textbf{0.95} & 3.5\TEXTsymbol{\vert}%
\textbf{1.30} & 5.6\TEXTsymbol{\vert}\textbf{0.71} \\ 
$k=12,n_{0}=50$ & 4.0\TEXTsymbol{\vert}\textbf{1.06} & 2.8\TEXTsymbol{\vert}%
\textbf{1.27} & 3.5\TEXTsymbol{\vert}\textbf{1.07} & 3.6\TEXTsymbol{\vert}%
\textbf{1.33} & 3.4\TEXTsymbol{\vert}\textbf{1.04} & 2.9\TEXTsymbol{\vert}%
\textbf{1.33} & 8.7\TEXTsymbol{\vert}0.75 \\ 
$k=4,n_{0}=25$ & 5.1\TEXTsymbol{\vert}\textbf{1.01} & 2.6\TEXTsymbol{\vert}%
\textbf{1.37} & 3.3\TEXTsymbol{\vert}\textbf{1.13} & 4.2\TEXTsymbol{\vert}%
\textbf{1.56} & 3.7\TEXTsymbol{\vert}\textbf{1.07} & 3.6\TEXTsymbol{\vert}%
\textbf{1.43} & 6.2\TEXTsymbol{\vert}0.82 \\ \hline
\multicolumn{8}{c}{$n=500$} \\ \hline
\textsc{def}: $k=8,n_{0}=50$ & 4.9\TEXTsymbol{\vert}\textbf{1.00} & 4.1%
\TEXTsymbol{\vert}\textbf{1.18} & 4.3\TEXTsymbol{\vert}\textbf{1.21} & 4.5%
\TEXTsymbol{\vert}\textbf{1.13} & 4.1\TEXTsymbol{\vert}\textbf{1.22} & 4.4%
\TEXTsymbol{\vert}\textbf{1.18} & 3.2\TEXTsymbol{\vert}\textbf{1.21} \\ 
$k=4,n_{0}=50$ & 5.1\TEXTsymbol{\vert}\textbf{1.00} & 4.4\TEXTsymbol{\vert}%
\textbf{1.31} & 4.6\TEXTsymbol{\vert}\textbf{1.26} & 4.4\TEXTsymbol{\vert}%
\textbf{1.31} & 4.4\TEXTsymbol{\vert}\textbf{1.24} & 4.2\TEXTsymbol{\vert}%
\textbf{1.32} & 3.2\TEXTsymbol{\vert}\textbf{1.10} \\ 
$k=12,n_{0}=50$ & 5.1\TEXTsymbol{\vert}\textbf{1.00} & 3.7\TEXTsymbol{\vert}%
\textbf{1.18} & 4.2\TEXTsymbol{\vert}\textbf{1.17} & 3.1\TEXTsymbol{\vert}%
\textbf{1.16} & 4.2\TEXTsymbol{\vert}\textbf{1.16} & 3.7\TEXTsymbol{\vert}%
\textbf{1.19} & 2.7\TEXTsymbol{\vert}\textbf{1.28} \\ 
$k=4,n_{0}=25$ & 5.0\TEXTsymbol{\vert}\textbf{1.00} & 4.2\TEXTsymbol{\vert}%
\textbf{1.34} & 4.2\TEXTsymbol{\vert}\textbf{1.38} & 4.5\TEXTsymbol{\vert}%
\textbf{1.28} & 4.3\TEXTsymbol{\vert}\textbf{1.32} & 4.1\TEXTsymbol{\vert}%
\textbf{1.36} & 3.1\TEXTsymbol{\vert}\textbf{1.27}%
\end{tabular}%
$

\end{center}
\linespread{1.00}\selectfont
\begin{small}
\linespread{1.00}\selectfont%

Notes: See Table 1.%
\linespread{1.00}\selectfont
\end{small}
\linespread{1.00}\selectfont%
\end{table}%

Table 2 compares different versions of the new method across the same set of
seven populations. We consider $k\in \{4,8,12\}$ for the default parameter
space with $n_{0}=50$, and also include the even more robust test with $k=4$
constructed from the parameter space in Section \ref{sec:paraspace} with $%
n_{0}=25$. For $n=25,$ only the test with $n_{0}=25$ comes close to
controlling size for non-thin tailed populations (and the test with $k=12$
cannot be applied at all, since there is only a single \textquotedblleft
middle\textquotedblright\ observation). For larger $n$, the tests with $k=4$
are even more successful in controlling size compared to the default method,
but at a non-negligible cost in terms of longer confidence intervals. In
contrast, the test for $k=12$ does not yield an additional substantial
reduction in average length, and has worse size control for $n=50$. These
results underlie our choice of the test with $k=8$ and $n_{0}=50$ as the
default, and in the following, we exclusively focus on this variant.

A potential objection to this first set of Monte Carlo results is that the
underlying populations have smooth tails, which might overstate the
effectiveness of the new method \textquotedblleft in
practice\textquotedblright . To address this concern, consider a population
that is equal to the (discrete) distribution from a large economic data set.
We use the income data of 2016 mortgage applicants as reported by banks
under the Home Mortgage Disclosure Act (HMDA). From this database of more
than 16 million applications, we create subpopulations that condition on
U.S.~state and the gender of the applicant, as well as the purpose of the
mortgage (home purchase, home improvement or refinancing) and whether or not
the unit is owner-occupied. We eliminate all records with missing data, and
only retain subpopulations with at least 5000 observations. For each of the
resulting 300 subpopulations, we compare the performance of alternative
methods for inference about the mean, based on i.i.d.~samples of size $n$
(that is, sampling is with replacement).

\begin{figure}[tbp] 
\begin{center}
\small%
\caption{Small Sample Results for HMDA Populations}\bigskip
\includegraphics[width=5.8356in,keepaspectratio]{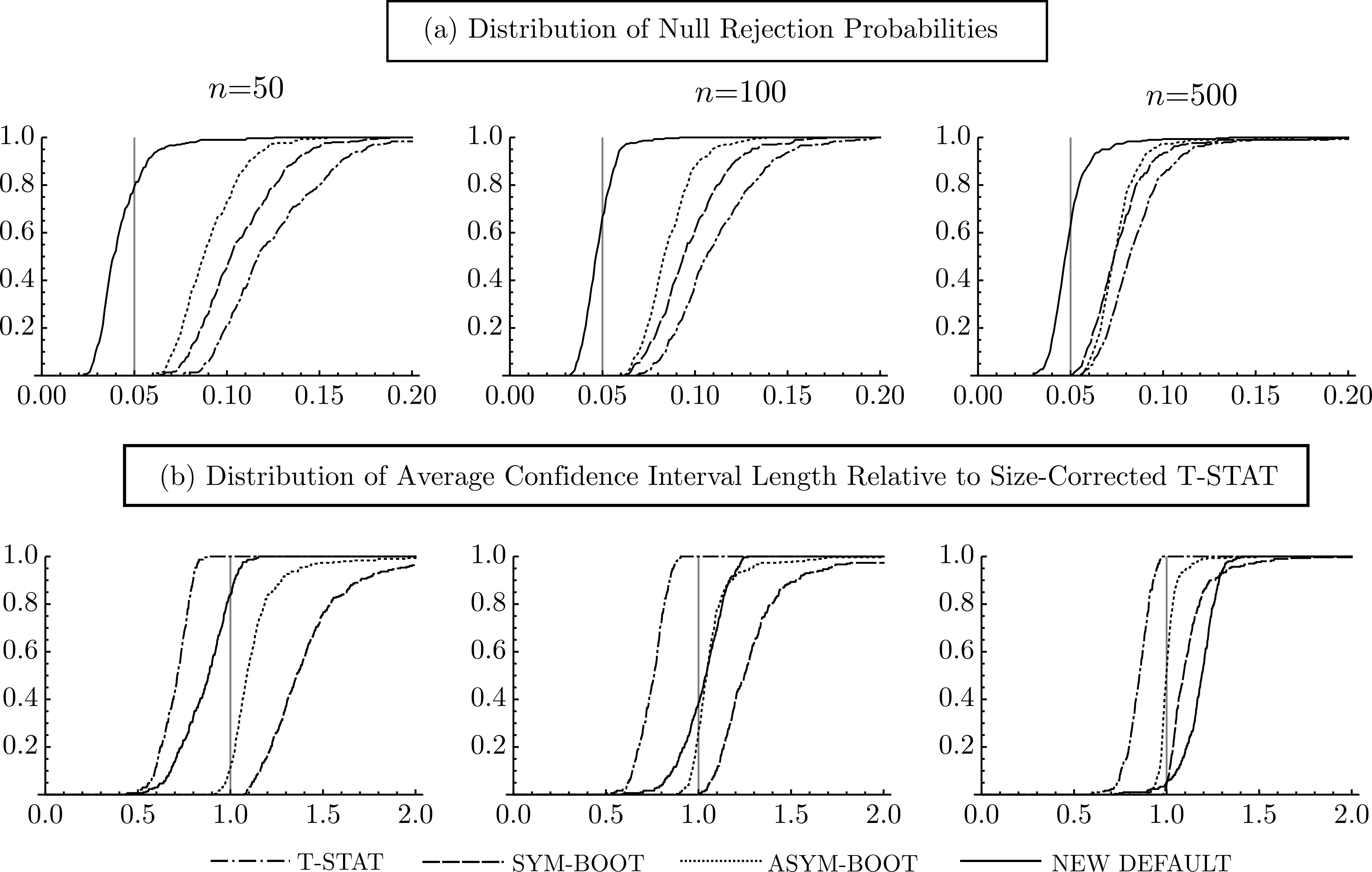}

\end{center}
\begin{small}%
\end{small}%
\end{figure}%

Panel (a) of Figure 2 plots the cumulative distribution function of the null
rejection probabilities over the 300 subpopulations for each tests
considered in Table 1, estimated from 20,000 draws from each subpopulation.
Nominally, all mass should be to the left of the 5\% line, but the
traditional tests don't comes close. For instance, for $n=100,$ the usual
t-statistic has null rejection probability of less than 10\% for only
approximately 40\% of the 300 subpopulation. In comparison, the new test
controls size much more successfully.

Panel (b) of Figure 2 plots the cumulative distribution function of the
average length of the confidence intervals, relative to the average length
of the size corrected t-statistic based interval. For $n=50$, the new method
not only controls size better than the bootstrap tests, but it also leads to
confidence intervals that are typically shorter on average. In fact, they
are substantially shorter than what is obtained from the infeasible size
corrected interval. For $n=\{100,500\}$, this is no longer the case and the
better size control of the new method comes at the cost of somewhat longer
confidence intervals.

One might argue that in the HMDA example, one could avoid the complications
of the heavy right tail of the income distribution by considering the
logarithm of the applicants' income. But, of course, there is no robust way
to transform a confidence interval for the population mean of log-income
into a valid confidence interval for the population mean income. What is
more, in many contexts, the policy relevant parameter is the population mean
(and not, say, the median) of some potentially heavy-tailed distribution:
think of health care costs, or flood damage, or asset returns.

\subsection{Difference between Two Population Means}

Our second set of Monte Carlo experiment concerns inference about the
difference of two population means $\mathbb{E}[W^{\func{I}}]-\mathbb{E}[W^{%
\func{II}}]$ based on two independent equal-sized i.i.d.~samples $%
W_{i}^{j}\sim W^{j}$, $i=1,\ldots ,n/2$, $j\in \{$I$,$II$\}$. Casting this
in terms of a linear regression and applying the general mapping (\ref%
{What_GMM})\ yields%
\begin{equation*}
\hat{W}_{i}=\left\{ 
\begin{array}{l}
\bar{W}^{\func{I}}-\bar{W}^{\func{II}}+2(W_{i}^{\func{I}}-\bar{W}^{\func{I}})%
\text{ \ for }i\leq n/2 \\ 
\bar{W}^{\func{I}}-\bar{W}^{\func{II}}-2(W_{i-n/2}^{\func{II}}-\bar{W}^{%
\func{II}})\ \text{ for }i>n/2%
\end{array}%
\right.
\end{equation*}%
where $\bar{W}^{j}=(n/2)^{-1}\sum_{i=1}^{n/2}W_{i}^{j}$ are the sample means
for $j\in \{$I$,$II$\}$.

We initially generate data according to%
\begin{equation}
W_{i}^{\func{I}}=\nu _{i}+\varepsilon _{i}^{\func{I}}\text{, \ \ }W_{i}^{%
\func{II}}=\varepsilon _{i}^{\func{II}}  \label{W2S}
\end{equation}%
for $i=1,\ldots ,n/2$, where $\varepsilon _{i}^{j}\sim iid\mathcal{N}%
(0,1/10) $ across $i$ and $j\in \{$I$,$II$\}$, and $\nu _{i}$ is distributed
according to one of the distributions of Table 1. Inference about $\mathbb{E}%
[W^{\func{I}}]-\mathbb{E}[W^{\func{II}}]$ can then be thought of as
inference about the average treatment effect $\mathbb{E}[\nu _{i}]$, with
the design amounting to a large but highly heterogeneous additive treatment
effect.

\begin{table}[tbp] 
\begin{center}
\small%
\caption{Small Sample Results for Difference of Population Means}

\bigskip

$%
\begin{tabular}{lrrrrrrr}
\hline\hline
& N(0,1) & LogN & F(4,5) & \ \ t(3) & P(0.4) & Mix 1 & Mix 2 \\ \hline
\multicolumn{8}{c}{$n=50$} \\ \hline
\textsc{t-stat} & 5.7\TEXTsymbol{\vert}\textbf{0.96} & 8.9\TEXTsymbol{\vert}%
0.81 & 8.9\TEXTsymbol{\vert}0.83 & 5.1\TEXTsymbol{\vert}\textbf{0.99} & 8.9%
\TEXTsymbol{\vert}0.83 & 7.2\TEXTsymbol{\vert}0.90 & 9.0\TEXTsymbol{\vert}%
0.83 \\ 
\textsc{sym-boot} & 5.7\TEXTsymbol{\vert}\textbf{0.97} & 8.3\TEXTsymbol{\vert%
}1.02 & 8.6\TEXTsymbol{\vert}1.12 & 4.7\TEXTsymbol{\vert}\textbf{1.07} & 8.6%
\TEXTsymbol{\vert}1.12 & 6.8\TEXTsymbol{\vert}1.06 & 8.9\TEXTsymbol{\vert}%
1.22 \\ 
\textsc{asym-boot} & 5.9\TEXTsymbol{\vert}\textbf{0.97} & 8.8\TEXTsymbol{%
\vert}0.93 & 9.1\TEXTsymbol{\vert}0.98 & 7.4\TEXTsymbol{\vert}1.03 & 9.1%
\TEXTsymbol{\vert}0.98 & 8.6\TEXTsymbol{\vert}0.98 & 10.0\TEXTsymbol{\vert}%
1.02 \\ 
\textsc{new default} & 2.0\TEXTsymbol{\vert}\textbf{1.40} & 3.8\TEXTsymbol{%
\vert}\textbf{1.06} & 4.7\TEXTsymbol{\vert}\textbf{1.06} & 2.3\TEXTsymbol{%
\vert}\textbf{1.47} & 4.8\TEXTsymbol{\vert}\textbf{1.05} & 3.5\TEXTsymbol{%
\vert}\textbf{1.25} & 6.2\TEXTsymbol{\vert}0.99 \\ \hline
\multicolumn{8}{c}{$n=100$} \\ \hline
\textsc{t-stat} & 5.5\TEXTsymbol{\vert}\textbf{0.98} & 7.8\TEXTsymbol{\vert}%
0.87 & 7.9\TEXTsymbol{\vert}0.87 & 4.9\TEXTsymbol{\vert}\textbf{1.00} & 8.2%
\TEXTsymbol{\vert}0.86 & 6.9\TEXTsymbol{\vert}0.92 & 9.9\TEXTsymbol{\vert}%
0.82 \\ 
\textsc{sym-boot} & 5.4\TEXTsymbol{\vert}\textbf{0.98} & 7.0\TEXTsymbol{\vert%
}1.04 & 7.4\TEXTsymbol{\vert}1.13 & 4.4\TEXTsymbol{\vert}\textbf{1.07} & 7.7%
\TEXTsymbol{\vert}1.14 & 6.4\TEXTsymbol{\vert}1.05 & 9.6\TEXTsymbol{\vert}%
1.21 \\ 
\textsc{asym-boot} & 5.4\TEXTsymbol{\vert}\textbf{0.98} & 7.6\TEXTsymbol{%
\vert}0.98 & 7.9\TEXTsymbol{\vert}1.01 & 6.9\TEXTsymbol{\vert}1.04 & 8.3%
\TEXTsymbol{\vert}1.01 & 7.6\TEXTsymbol{\vert}0.99 & 10.6\TEXTsymbol{\vert}%
1.02 \\ 
\textsc{new default} & 4.4\TEXTsymbol{\vert}\textbf{1.08} & 3.4\TEXTsymbol{%
\vert}\textbf{1.28} & 4.2\TEXTsymbol{\vert}\textbf{1.20} & 3.7\TEXTsymbol{%
\vert}\textbf{1.43} & 4.4\TEXTsymbol{\vert}\textbf{1.18} & 4.0\TEXTsymbol{%
\vert}\textbf{1.30} & 7.8\TEXTsymbol{\vert}1.01 \\ \hline
\multicolumn{8}{c}{$n=500$} \\ \hline
\textsc{t-stat} & 5.4\TEXTsymbol{\vert}\textbf{0.98} & 6.4\TEXTsymbol{\vert}%
0.94 & 6.4\TEXTsymbol{\vert}0.93 & 4.6\TEXTsymbol{\vert}\textbf{1.01} & 6.8%
\TEXTsymbol{\vert}0.91 & 5.7\TEXTsymbol{\vert}\textbf{0.97} & 8.3\TEXTsymbol{%
\vert}0.85 \\ 
\textsc{sym-boot} & 5.5\TEXTsymbol{\vert}\textbf{0.98} & 5.8\TEXTsymbol{\vert%
}\textbf{1.01} & 6.0\TEXTsymbol{\vert}1.12 & 4.4\TEXTsymbol{\vert}\textbf{%
1.03} & 6.3\TEXTsymbol{\vert}1.10 & 5.3\TEXTsymbol{\vert}\textbf{1.01} & 7.7%
\TEXTsymbol{\vert}1.08 \\ 
\textsc{asym-boot} & 5.4\TEXTsymbol{\vert}\textbf{0.98} & 6.3\TEXTsymbol{%
\vert}0.99 & 6.6\TEXTsymbol{\vert}1.04 & 5.8\TEXTsymbol{\vert}\textbf{1.02}
& 6.6\TEXTsymbol{\vert}1.02 & 6.2\TEXTsymbol{\vert}1.00 & 8.3\TEXTsymbol{%
\vert}0.98 \\ 
\textsc{new default} & 5.3\TEXTsymbol{\vert}\textbf{0.99} & 4.2\TEXTsymbol{%
\vert}\textbf{1.21} & 4.1\TEXTsymbol{\vert}\textbf{1.25} & 4.2\TEXTsymbol{%
\vert}\textbf{1.15} & 4.2\TEXTsymbol{\vert}\textbf{1.24} & 4.2\TEXTsymbol{%
\vert}\textbf{1.22} & 4.2\TEXTsymbol{\vert}\textbf{1.31}%
\end{tabular}%
$

\end{center}
\linespread{1.00}\selectfont
\begin{small}
\linespread{1.00}\selectfont%

Notes: See Table 1.%
\linespread{1.00}\selectfont
\end{small}
\linespread{1.00}\selectfont%
\end{table}%
Table 3 compares the new method to standard t-statistic based inference and
a symmetric and asymmetric percentile-t bootstrap, where now the bootstrap
samples combine $n/2$ randomly selected observations with replacement from
each of the two samples. In this exercise the design with $\nu _{i}\sim 
\mathcal{N}(0,1)$ leads to a much longer confidence interval from the new
method with $n=50$. The reason is that with $\varepsilon _{i}^{j}\sim 
\mathcal{N}(0,1/10)$ in (\ref{W2S}), $W_{i}^{\func{I}}$ has much larger
variance than $W_{i}^{\func{II}}$. The distribution of $\hat{W}_{i}$ is thus
approximately equal to a 50-50 mixture of two normal distributions with very
different variances, which is heavier tailed than a normal distribution. At
the same time, for asymmetric $\nu _{i}$, standard methods do not control
size well, while the new method does so much more successfully.

\begin{figure}[tbp] 
\begin{center}
\small%
\caption{Small Sample Results for Two Samples from HDMA Populations}\bigskip
\includegraphics[width=5.8356in,keepaspectratio]{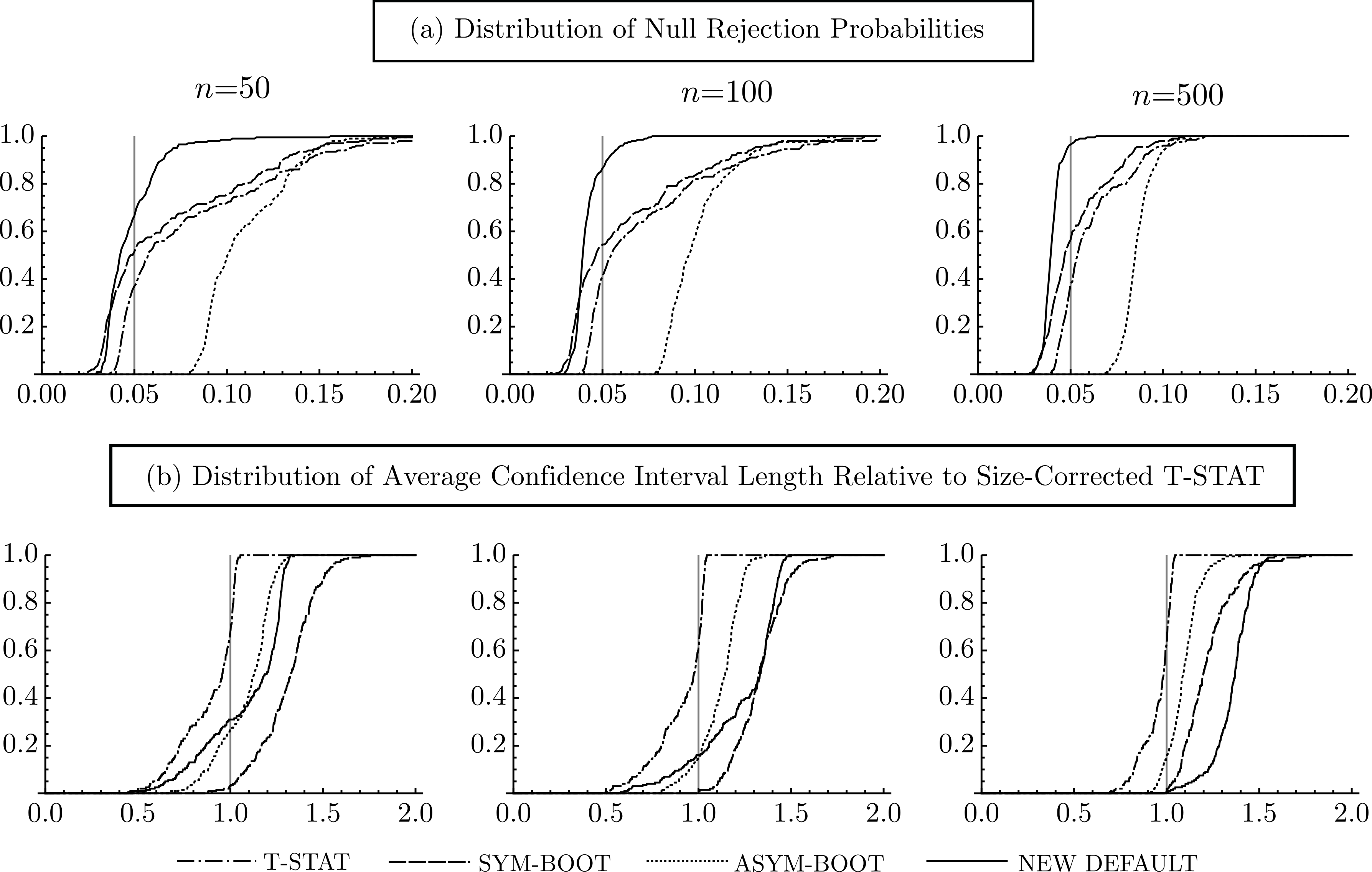}

\end{center}
\begin{small}%
\end{small}%
\end{figure}%
As a second exercise, we generate $W_{i}^{j}$ as $n/2$ i.i.d.~draws of two
randomly selected subpopulations of the HMDA data set considered in the last
section. Inference about $\mathbb{E}[W^{\func{I}}]-\mathbb{E}[W^{\func{II}}]$
then corresponds to inference about the average treatment effect if the
treatment induces a change from the distribution of income in one
subpopulation to the distribution in another---maybe a plausible calibration
for an intervention that affects individuals' incomes. Figure 3 reports the
performance of the inference methods of Table 3 for 200 randomly selected
pairs of subpopulations, in analogy to Figure 2 above. We find that also in
this exercise, standard methods fail to produce reliable inference, while
the new method is substantially more successful at controlling size.

\subsection{\label{sec:clust}Clustered Linear Regression}

A third set of Monte Carlo experiments explores the performance of the new
method for inference in a clustered linear regression%
\begin{equation}
Y_{it}=\beta X_{it}+Z_{it}^{\prime }\gamma +u_{it}\text{, }t=1,\ldots ,T_{i}%
\text{, }i=1,\ldots ,n  \label{cluster}
\end{equation}%
with conditionally mean zero $u_{it}$, so that there are $T_{i}$
observations in cluster $i$. Viewing linear regression as a special case of
GMM inference, we obtain from the development of Section \ref{sec:GMM} and
the Frisch-Waugh Theorem that 
\begin{equation*}
\hat{W}_{i}=\hat{\beta}+\left( n^{-1}\sum_{j=1}^{n}\sum_{t=1}^{T_{j}}\hat{X}%
_{jt}\right) ^{-1}\sum_{t=1}^{T_{i}}\hat{X}_{it}\hat{u}_{it}
\end{equation*}%
where $\hat{u}_{it}$ and $\hat{\beta}$ are the OLS\ estimates of $u_{it}$
and $\beta $, and $\hat{X}_{it}$ are the residuals of a OLS regression of $%
X_{it}$ on $Z_{it}$. We consider four tests of $H_{0}:\beta =\beta _{0}$:
The t-statistic implemented by STATA, which is nearly identical to a
standard t-test applied to $\hat{W}_{i}$, except for degree of freedom
corrections; the suggestion of \cite{Imbens16} to account for a potentially
small number of heterogeneous clusters \textquotedblleft \textsc{Im-Ko}%
\textquotedblright\ (we consider the variant that involves the data
dependent degree of freedom adjustment $K_{IK}$ in their notation); the wild
cluster bootstrap that imposes the null hypothesis suggested by \cite%
{Cameron08} \textquotedblleft \textsc{CGM}\textquotedblright ; and the new
default test applied to $\hat{W}_{i}$ \textquotedblleft \textsc{new default}%
\textquotedblright .

We initially consider data generated from model (\ref{cluster}) where 
\begin{equation}
u_{it}=\nu _{i}X_{it}+\varepsilon _{it}\text{,}  \label{uit}
\end{equation}%
$\nu _{i}$ is i.i.d.~mean-zero with a distribution that is one of the seven
populations considered in Table 1, one element of $Z_{it}$ is a constant,
and $X_{it},$ the 5 non-constant elements of $Z_{it}$, and $\varepsilon
_{it} $ are independent standard normal. We set $T_{i}=T=10$ for all
clusters. The presence of $\nu _{i}$ induces heteroskedastic correlations
within each cluster of observations $\{Y_{it}\}_{t=1}^{T}$.

\begin{table}[tbp] 
\begin{center}
\small%
\caption{Small Sample Results in Clustered Regression Design}

\bigskip

$%
\begin{tabular}{lrrrrrrr}
\hline\hline
& N(0,1) & LogN & F(4,5) & \ \ t(3) & P(0.4) & Mix 1 & Mix 2 \\ \hline
\multicolumn{8}{c}{$n=50$} \\ \hline
\textsc{STATA} & 5.1\TEXTsymbol{\vert}\textbf{1.00} & 9.3\TEXTsymbol{\vert}%
0.80 & 10.7\TEXTsymbol{\vert}0.76 & 4.7\TEXTsymbol{\vert}\textbf{1.01} & 10.9%
\TEXTsymbol{\vert}0.75 & 6.9\TEXTsymbol{\vert}0.92 & 12.3\TEXTsymbol{\vert}%
0.75 \\ 
\textsc{Im-Ko} & 4.9\TEXTsymbol{\vert}\textbf{1.00} & 9.1\TEXTsymbol{\vert}%
0.81 & 10.5\TEXTsymbol{\vert}0.77 & 4.5\TEXTsymbol{\vert}\textbf{1.02} & 10.7%
\TEXTsymbol{\vert}0.75 & 6.7\TEXTsymbol{\vert}0.92 & 12.0\TEXTsymbol{\vert}%
0.75 \\ 
\textsc{CGM} & 5.0\TEXTsymbol{\vert}\textbf{1.01} & 9.4\TEXTsymbol{\vert}0.77
& 10.8\TEXTsymbol{\vert}0.72 & 5.0\TEXTsymbol{\vert}\textbf{1.00} & 11.0%
\TEXTsymbol{\vert}0.70 & 7.0\TEXTsymbol{\vert}0.89 & 12.3\TEXTsymbol{\vert}%
0.68 \\ 
\textsc{new default} & 3.3\TEXTsymbol{\vert}\textbf{1.34} & 3.5\TEXTsymbol{%
\vert}\textbf{0.97} & 4.4\TEXTsymbol{\vert}\textbf{0.92} & 2.8\TEXTsymbol{%
\vert}\textbf{1.44} & 4.5\TEXTsymbol{\vert}\textbf{0.89} & 3.3\TEXTsymbol{%
\vert}\textbf{1.19} & 7.3\TEXTsymbol{\vert}0.88 \\ \hline
\multicolumn{8}{c}{$n=100$} \\ \hline
\textsc{STATA} & 5.2\TEXTsymbol{\vert}\textbf{0.99} & 7.6\TEXTsymbol{\vert}%
0.87 & 9.5\TEXTsymbol{\vert}0.81 & 4.7\TEXTsymbol{\vert}\textbf{1.01} & 9.8%
\TEXTsymbol{\vert}0.79 & 6.7\TEXTsymbol{\vert}0.93 & 11.8\TEXTsymbol{\vert}%
0.75 \\ 
\textsc{Im-Ko} & 5.1\TEXTsymbol{\vert}\textbf{1.00} & 7.5\TEXTsymbol{\vert}%
0.87 & 9.4\TEXTsymbol{\vert}0.81 & 4.7\TEXTsymbol{\vert}\textbf{1.01} & 9.8%
\TEXTsymbol{\vert}0.80 & 6.6\TEXTsymbol{\vert}0.93 & 11.7\TEXTsymbol{\vert}%
0.75 \\ 
\textsc{CGM} & 5.0\TEXTsymbol{\vert}\textbf{1.00} & 7.7\TEXTsymbol{\vert}0.85
& 9.6\TEXTsymbol{\vert}0.77 & 4.9\TEXTsymbol{\vert}\textbf{1.01} & 9.9%
\TEXTsymbol{\vert}0.75 & 6.6\TEXTsymbol{\vert}0.91 & 11.9\TEXTsymbol{\vert}%
0.69 \\ 
\textsc{new default} & 4.5\TEXTsymbol{\vert}\textbf{1.11} & 3.2\TEXTsymbol{%
\vert}\textbf{1.26} & 4.2\TEXTsymbol{\vert}\textbf{1.12} & 4.0\TEXTsymbol{%
\vert}\textbf{1.42} & 4.4\TEXTsymbol{\vert}\textbf{1.10} & 3.8\TEXTsymbol{%
\vert}\textbf{1.32} & 7.0\TEXTsymbol{\vert}0.96 \\ \hline
\multicolumn{8}{c}{$n=500$} \\ \hline
\textsc{STATA} & 5.1\TEXTsymbol{\vert}\textbf{1.00} & 6.1\TEXTsymbol{\vert}%
0.95 & 7.1\TEXTsymbol{\vert}0.91 & 5.0\TEXTsymbol{\vert}\textbf{1.00} & 7.5%
\TEXTsymbol{\vert}0.89 & 5.5\TEXTsymbol{\vert}\textbf{0.97} & 8.8\TEXTsymbol{%
\vert}0.83 \\ 
\textsc{Im-Ko} & 5.1\TEXTsymbol{\vert}\textbf{1.00} & 6.1\TEXTsymbol{\vert}%
0.95 & 7.1\TEXTsymbol{\vert}0.91 & 5.0\TEXTsymbol{\vert}\textbf{1.00} & 7.4%
\TEXTsymbol{\vert}0.89 & 5.5\TEXTsymbol{\vert}\textbf{0.98} & 8.8\TEXTsymbol{%
\vert}0.83 \\ 
\textsc{CGM} & 5.0\TEXTsymbol{\vert}\textbf{1.00} & 6.1\TEXTsymbol{\vert}0.94
& 7.3\TEXTsymbol{\vert}0.87 & 5.1\TEXTsymbol{\vert}\textbf{0.99} & 7.7%
\TEXTsymbol{\vert}0.85 & 5.6\TEXTsymbol{\vert}\textbf{0.97} & 9.0\TEXTsymbol{%
\vert}0.80 \\ 
\textsc{new default} & 5.0\TEXTsymbol{\vert}\textbf{1.00} & 4.1\TEXTsymbol{%
\vert}\textbf{1.20} & 4.3\TEXTsymbol{\vert}\textbf{1.24} & 4.7\TEXTsymbol{%
\vert}\textbf{1.14} & 4.4\TEXTsymbol{\vert}\textbf{1.23} & 4.0\TEXTsymbol{%
\vert}\textbf{1.22} & 3.5\TEXTsymbol{\vert}\textbf{1.29}%
\end{tabular}%
$

\end{center}
\linespread{1.00}\selectfont
\begin{small}
\linespread{1.00}\selectfont%

Notes: Entries are the null rejection probability in percent, and the
average length of confidence intervals relative to average length of
confidence intervals based on size corrected \textsc{STATA} (bold if null
rejection probability is smaller than 6\%) of nominal 5\% level tests.%
\linespread{1.00}\selectfont
\end{small}
\linespread{1.00}\selectfont%
\end{table}%

Table 4 reports the results. As in the inference about the mean problem, the
new method is seen to control size much more successfully compared to the
other methods, although at a cost in average confidence interval length that
is more pronounced than in Table 1 for the thin-tailed $\nu _{i}\sim 
\mathcal{N}(0,1)$. Intuitively, the product of two independent normals $\nu
_{i}X_{it}$ has considerably heavier tails than a normal distribution, but
it is still symmetric.

In the final Monte Carlo exercise we again consider a discrete population
from a large economic data set. In particular, we consider a sample of all
employed workers aged 18-65 from the 2018 merged outgoing rotation group
sample of the Current Population Survey (CPS). We let the dependent variable 
$Y_{it}$ be the logarithm of wages, and pick the regressor of interest $%
X_{it}$ and the 5 non-constant controls $Z_{it}$ as a random subset of
potential regressor including gender, race, age and dummies for Hispanic,
non-white, married, public sector employer, union membership and whether
hours or the wage was imputed. The resulting coefficient $\beta $ on $X_{it}$
in the regression using the entire 145,838 individuals in the database is
the population coefficient. We cluster at the level of 308 Metropolitan
Statistical Areas (MSAs).\footnote{%
For the purposes of this exercise, we treat as additional MSAs the part of
each U.S. state outside of any CBSA area.} The four different methods of
Table 4 are then employed to conduct inference about $\beta $ based on a
sample consisting of all individuals that reside in $n$ randomly selected
MSAs, where the MSAs are drawn with replacement. By construction the
clusters are thus i.i.d.~and the population regression coefficient is equal
to $\beta $.

\begin{figure}[tbp] 
\begin{center}
\small%
\caption{Small Sample Results for CPS Clustered Regressions}

\bigskip
\includegraphics[width=5.8356in,keepaspectratio]{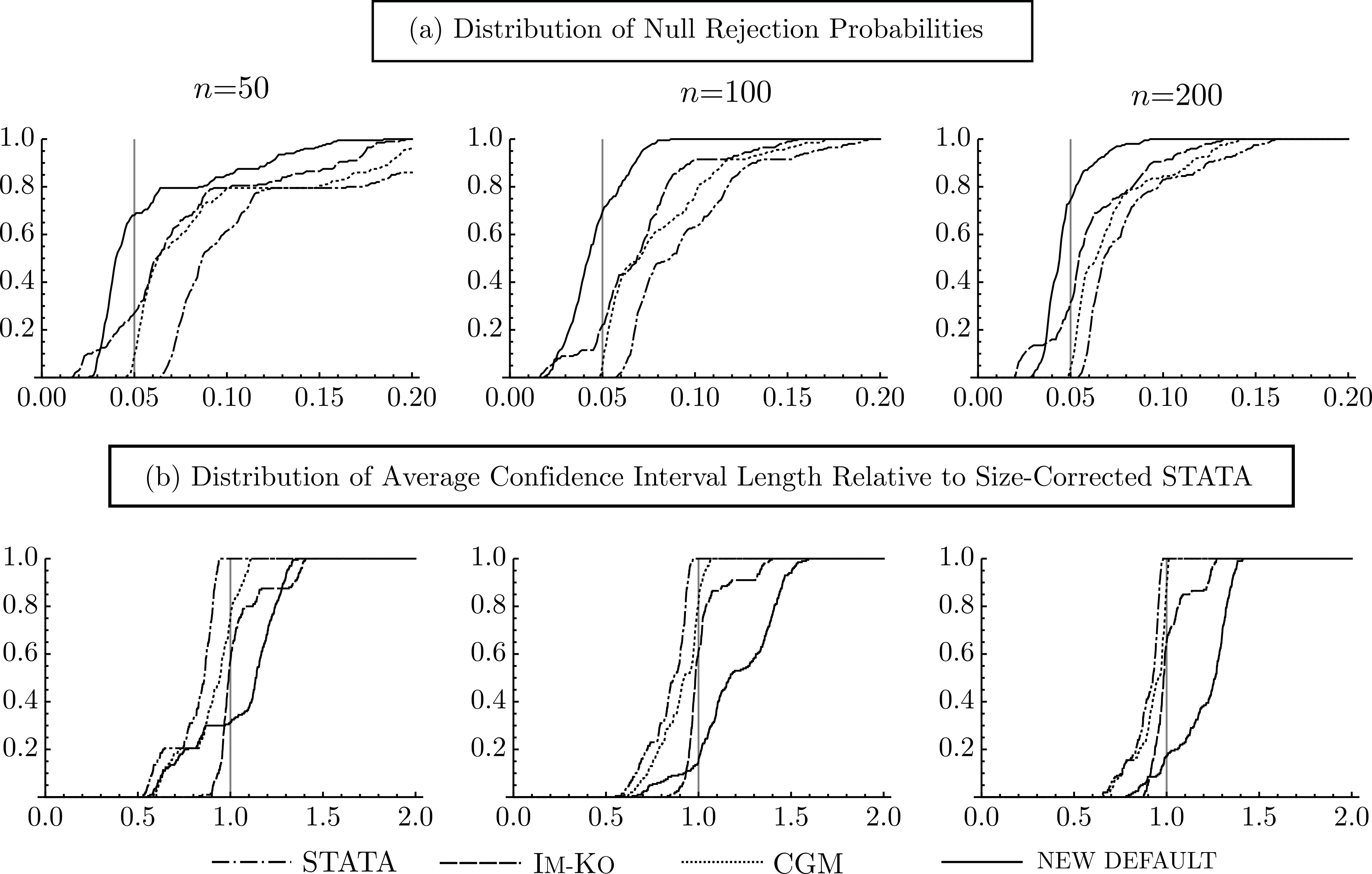}

\end{center}
\begin{small}%
\end{small}%
\end{figure}%
Figure 4 depicts the results over 200 populations generated in this manner,
analogous to Figure 3, for $n\in \{50,100,200\}$. (We consider $n=200$
rather than $n=500$ for the largest sample size to avoid that with high
probability, samples contain many identical clusters.) In this design none
of the methods come close to perfectly controlling size. Still, the new
method is substantially more successful, albeit at the cost of considerably
longer average confidence intervals for $n\in \{100,200\}$.

The poor performance of the standard methods might come as a surprise given
that none of the variables in the CPS exercise are heavy-tailed, and the
number of clusters is not particularly small. Approximately (cf. equation (%
\ref{GMM_asy})), the variability of the OLS estimator $\hat{\beta}$ is
driven by the average of the $n$ i.i.d.~random variables 
\begin{equation*}
G_{i}=\sum_{t=1}^{T_{i}}\tilde{X}_{it}u_{it}
\end{equation*}%
where $u_{it}$ is the population regression error and $\tilde{X}_{it}$ is
the residual of a population regression of $X_{it}$ on $Z_{it}$. The
distribution of $G_{i}$ may be heavy-tailed because (i) $u_{it}$ has a
heavy-tailed component, as in (\ref{uit}) above; (ii) the joint distribution
of $(\tilde{X}_{it},u_{it})$ is such that $\tilde{X}_{it}u_{it}$ is
heavy-tailed; (iii) $\tilde{X}_{it}$ is heavy-tailed; (iv) $T_{i}$ is
heterogeneous across $i$, so that large clusters with big $T_{i}$ lead to $%
G_{i}$ with high variance; or a combination of these effects. MSAs are
highly heterogeneous in their size: the largest contains 6,163 individuals,
and the smallest only 42. Effect (iv) is thus clearly present, and the
suggestion by \cite{Imbens16} is designed to accommodate effects (iii) and
(iv). But as reported in Table 4, if $G_{i}$ is heavy-tailed due to effect
(i), then the adjustment of \cite{Imbens16} does not help much. The CPS\
design seems to exhibit all four effects to some degree, making correct
inference quite challenging, and the new method relatively most successful
at controlling size.

\section{Conclusion}

Whenever researchers compare a t-statistic to the usual standard normal
critical value they effectively assume that the central limit theorem
provides a reasonable approximation. This is true when conducting inference
for the mean from an i.i.d.~sample, but it holds more generally for linear
regression, GMM inference, and so forth. As is well understood, the central
limit theorem requires that the contribution of each term to the overall
variation is small. To some extent, this is empirically testable: one can
simply compare the absolute values of each (demeaned) term with the sample
standard deviation. The normal approximation then surely becomes suspect if
the largest absolute term is, say, equal to half of a standard deviation.

One may view the new test suggested here as a formalization of this notion:
the extreme terms are set apart, and if they are large, then the test
automatically becomes more conservative. What is more, even if the sample
realization from an underlying population with a heavy tail fails to
generate a very large term, it still leaves a tell-tale sign in the large
spacings between the largest terms. Correspondingly, the test also becomes
more conservative if the largest observations are far apart from each other,
even if the largest one isn't all that large---the new method seeks to infer
the likelihood of a potential large outlier based on the spacings of the
extreme terms. These adjustments are disciplined by an assumption of
Pareto-like tails. But as was found in the small sample simulations, they
help generate more reliable inference also when the underlying population is
more loosely characterized by a moderately heavy tail.

It would be desirable to extend the new method also to $F$-type tests of
null hypotheses that restrict more than one parameter. Such an extension is
far from straightforward, though: extreme value theory for the joint
extremes of vector-valued observations does not yield a tightly parametrized
approximate model, and the approach pursued here to determine the
appropriate adjustments is already very computationally challenging.%
\clearpage%
\pagebreak

\appendix%
\renewcommand{\baselinestretch}{1.2} \small \normalsize \sloppy
\allowdisplaybreaks
\begin{small}%

\section{Appendix}

\subsection{Proof of Theorem \protect\ref{thm_main}}

We write $C$ for a generic large enough positive constant, not necessarily
the same in each instance. Without loss of generality, under Condition 1 we
can choose $w_{0}$ large enough so that uniformly in $w\geq w_{0}$, 
\begin{eqnarray}
C^{-1}w^{-1/\xi } &\leq &1-F(w)\leq Cw^{-1/\xi }  \label{ineq_Fbar} \\
f(w) &\leq &Cw^{-1/\xi -1}.  \label{ineq_fbar}
\end{eqnarray}

Also we normalize $\limfunc{Var}[W]=1$. Define%
\begin{eqnarray*}
m(w) &=&-\mathbb{E}[W|W\leq m]\text{, }m^{\ast }(w)=\sigma ^{1/\xi }\frac{%
w^{1-1/\xi }}{1-\xi } \\
A_{n} &=&\mathbf{1}[W_{k}^{R}>w_{0}]\text{, }A_{n}^{\ast }=\mathbf{1}[n^{\xi
}\sigma X_{k}>w_{0}]\text{,}
\end{eqnarray*}%
$V(w)=\limfunc{Var}[W|W\leq w]$ and $\Delta _{n}=|(1-k/n)^{-1/2}/s_{n}-1/%
\sqrt{V(W_{k}^{R})}|$.

The proof of Theorem 2 is based on a number of preliminary Lemmas. We assume
throughout that the assumptions of Theorem \ref{thm_main} hold, and that $%
n>k+1$. All limits are taken as $n\rightarrow \infty $.

\begin{lemma}
\label{lm:A}For any $p>0$,

(a) $n^{p}\mathbb{P}(W_{k}^{R}\leq w_{0})\rightarrow 0$;

(b) $n^{p}\mathbb{P}(n^{\xi }\sigma X_{k}\leq w_{0})\rightarrow 0.$
\end{lemma}

\begin{proof}
(a) Follows from $\mathbb{P}(W_{k}^{R}\leq w_{0})=\sum_{i=0}^{k-1}\dbinom{n}{%
n-i}F(w_{0})^{n-i}(1-F(w_{0}))^{i}\leq n^{k}F(w_{0})^{n-k}$ and $F(w_{0})<1$.

(b) $\mathbb{\,}$Follows from a direct calculation from the density of $%
(\sum_{l=1}^{k}E_{l})^{-\xi }=X_{k}/\sigma $.
\end{proof}

\begin{lemma}
\label{lm:expecT}(a) For any $p<k/\xi $, $\mathbb{E}[A_{n}|n^{-\xi
}W_{k}^{R}|^{p}]=O(1).$

(b) $\mathbb{E}[A_{n}||n^{-\xi }\mathbf{W}^{R}||]=O(1).$

(c) For any $p<0$, $\mathbb{E}[X_{k}^{p}]=O(1).$

(d) $\mathbb{E}[||\mathbf{X}||^{2}]<\infty .$
\end{lemma}

\begin{proof}
(a) The density of $S=n^{-\xi }W_{k}^{R}$, for $s>n^{-\xi }w_{0}$, is given
by%
\begin{equation*}
\frac{n!}{(n-k)!(k-1)!}(1-F(n^{\xi }s))^{k-1}F(n^{\xi }s)^{n-k}f(n^{\xi
}s)n^{\xi }\leq F(w_{0})^{-k}n^{k+\xi }(1-F(n^{\xi }s))^{k-1}F(n^{\xi
}s)^{n}f(n^{\xi }s).
\end{equation*}%
Using (\ref{ineq_Fbar}) and (\ref{ineq_fbar}), we have $(1-F(n^{\xi
}s))^{k-1}\leq Cn^{1-k}s^{(1-k)/\xi }$ and $f(n^{\xi }s)\leq Cn^{-1-\xi
}s^{-1/\xi -1}$. Furthermore, using $(1-a/n)^{n}\leq e^{-a}$ for all $0\leq
a\leq n$ and (\ref{ineq_Fbar}), we have uniformly in $s\geq n^{-\xi }w_{0}$%
\begin{equation*}
F(n^{\xi }s)^{n}\leq (1-C^{-1}s^{-1/\xi }/n)^{n}\leq \exp (-C^{-1}s^{-1/\xi
}).
\end{equation*}%
Thus, the density of $S_{n}$ is bounded above by $Cs^{-k/\xi -1}\exp
(-C^{-1}s^{-1/\xi })$ on $s\in \lbrack n^{-\xi }w_{0},\infty )$, and the
result follows.

(b) $A_{n}||\mathbf{W}^{R}||\leq kW_{1}^{R}$, and, proceeding as in the
proof of part (a), the density of $n^{-\xi }W_{1}^{R}$ for $s>n^{-\xi }w_{0}$
is bounded above by $Cs^{-1/\xi -1}\exp (-C^{-1}s^{-1/\xi })$ on $s\in
\lbrack n^{-\xi }w_{0},\infty )$, so the result follows.

(c) $\mathbb{E}[X_{k}^{p}]=\sigma ^{p}\mathbb{E}[(\sum_{l=1}^{k}E_{l})^{-p%
\xi }]<\infty $, where the last inequality follows a direct calculation.

(d) $||\mathbf{X}||\leq k\sigma E_{1}^{-\xi }$, and the result follows by a
direct calculation.
\end{proof}

\begin{lemma}
\label{lm:w}For $w>w_{0}$, let $\tilde{W}^{0}$ be a random variable with
c.d.f. equal to $F(\tilde{w})/F(w)$ for $\tilde{w}<w$, and equal to one
otherwise, and let $\tilde{W}=\tilde{W}^{0}+m(w)$. Then, uniformly in $%
w>w_{0}$

(a) $m(w)\leq Cw^{-1/\xi +1}$;

(b) $|V(w)-1|\leq Cw^{2-1/\xi }$;

(c) $|m(w)-m^{\ast }(w)|\leq Cw^{1-(1+\delta )/\xi }+Cw^{1-2/\xi }$;

(d) for any $\beta _{0}>1$ and $1<\beta <\beta _{0}$, $\mathbb{E}[|\tilde{W}%
|^{\beta }]\leq Cw^{\beta _{0}-1/\xi }$;

(e) $\mathbb{E}[\tilde{W}^{2}\mathbf{1}[\tilde{W}^{2}>V(w)n]]\leq
Cw^{2}n^{-1/(2\xi )}$;

(f) $\mathbb{E}[|\tilde{W}|^{3}\mathbf{1}[\tilde{W}^{2}\leq V(w)n]]\leq
Cw^{(1/2-r_{k}(\xi ))/\xi }.$
\end{lemma}

\begin{proof}
(a) Follows from $m(w)=\mathbb{E}[W\mathbf{1}[W>w]]/F(w)$, (\ref{ineq_fbar})
and $F(w)\geq F(w_{0})>0$.

(b) $V(w)=\mathbb{E}[W^{2}\mathbf{1}[W<w]]/F(w)-m(w)^{2}$, so that%
\begin{equation*}
1-V(w)=\frac{F(w)-1}{F(w)}+\frac{1-\mathbb{E}[W^{2}\mathbf{1}[W<w]]}{F(w)}%
+m(w)^{2}.
\end{equation*}%
Now for $w>w_{0}$, $F(w)^{-1}\leq F(w_{0})^{-1}$, $1-F(w)\leq Cw^{-1/\xi }$
by (\ref{ineq_Fbar}), and $m(w)\leq Cw^{-1/\xi +1}$ from part (a).
Furthermore, using (\ref{ineq_fbar}) 
\begin{eqnarray*}
1-\mathbb{E}[W^{2}\mathbf{1}[W\left. <w\right. ]] &=&\int_{w}^{\infty
}f(s)s^{2}ds\leq C\int_{w}^{\infty }s^{-1/\xi +1}ds \\
&\leq &Cw^{2-1/\xi }
\end{eqnarray*}%
so the result follows.

(c) For $w>w_{0}$%
\begin{eqnarray*}
|m(w)-m^{\ast }(w)| &=&\left\vert \frac{\int_{w}^{\infty }sf(s)ds}{F(w)}-%
\frac{\sigma ^{1/\xi }\frac{w^{1-1/\xi }}{1-\xi }}{F(w)}+\frac{\sigma
^{1/\xi }\frac{w^{1-1/\xi }}{1-\xi }}{F(w)}-\sigma ^{1/\xi }\frac{w^{1-1/\xi
}}{1-\xi }\right\vert \\
&\leq &F(w)^{-1}\left\vert \int_{w}^{\infty }sf(s)ds-\int_{w}^{\infty }(\xi
\sigma )^{-1}(\frac{s}{\sigma })^{-1/\xi -1}ds\right\vert +\sigma ^{1/\xi }%
\frac{w^{1-1/\xi }}{1-\xi }F(w)^{-1}(1-F(w))
\end{eqnarray*}%
and%
\begin{eqnarray*}
\left\vert \int_{w}^{\infty }sf(s)ds-\int_{w}^{\infty }(\xi \sigma )^{-1}(%
\frac{s}{\sigma })^{-1/\xi -1}ds\right\vert &\leq &C\int_{w}^{\infty
}s^{-1/\xi }|h(s)|ds \\
&\leq &C\int_{w}^{\infty }s^{-(\delta +1)/\xi }ds \\
&\leq &Cw^{1-(1+\delta )/\xi }
\end{eqnarray*}%
and $F(w)^{-1}\leq F(w_{0})^{-1}$, $1-F(w)\leq Cw^{-1/\xi }$, so that%
\begin{equation*}
|m(w)-m^{\ast }(w)|\leq Cw^{1-(1+\delta )/\xi }+Cw^{1-2/\xi }.
\end{equation*}

(d) By the $c_{r}$ inequality and the result of part (a)%
\begin{eqnarray*}
\mathbb{E}[|\tilde{W}|^{\beta }] &=&\mathbb{E}[|W+m(w)|^{\beta }\left\vert
W<w\right. ] \\
&\leq &C\mathbb{E}[|W|^{\beta _{0}-1/\xi }|W|^{1/\xi -\beta _{0}+\beta
}\left\vert W<w\right. ]+C|m(w)|^{\beta } \\
&\leq &Cw^{\beta _{0}-1/\xi }\mathbb{E}[|W|^{1/\xi -\beta _{0}+\beta
}\left\vert W<w\right. ]+Cw^{-\beta /\xi +\beta } \\
&\leq &Cw^{\beta _{0}-1/\xi }+Cw^{-\beta /\xi +\beta }
\end{eqnarray*}%
where the last inequality follows from $\mathbb{E}[|W|^{1/\xi -\beta
_{0}+\beta }]<\infty $.

(e) Note that $V(w)\geq V(w_{0})>0$, and by the result in part (a), $|\tilde{%
W}|\leq Cw$ uniformly in $w\geq w_{0}$ almost surely. Thus%
\begin{eqnarray*}
\mathbb{E}[\tilde{W}^{2}\mathbf{1}[\tilde{W}^{2}\left. >\right. V(w)n]]
&\leq &\mathbb{E}[\tilde{W}^{2}\mathbf{1}[\tilde{W}^{2}>V(w_{0})n]] \\
&\leq &Cw^{2}\mathbb{P}(\tilde{W}^{2}>V(w_{0})n) \\
&\leq &Cw^{2}\mathbb{P}(W^{2}>\tfrac{1}{2}V(w_{0})n)+Cw^{2}\mathbf{1}%
[m(w)^{2}>\tfrac{1}{2}V(w_{0})n]
\end{eqnarray*}%
where the third inequality uses 
\begin{eqnarray*}
\mathbb{P}(|\tilde{W}|\left. >\right. s) &=&\mathbb{P}(|W+m(w)|>s\left\vert
W<w\right. ) \\
&\leq &\mathbb{P}(|W|>\tfrac{1}{2}s\left\vert W<w\right. )+\mathbf{1}[m(w)>%
\tfrac{1}{2}s] \\
&\leq &F(w_{0})^{-1}\mathbb{P}(|W|>\tfrac{1}{2}s)+\mathbf{1}[m(w)>\tfrac{1}{2%
}s]
\end{eqnarray*}%
for all $s>0$. Now $\mathbf{1}[m(w)^{2}>\tfrac{1}{2}V(w_{0})n]=0$ for all
large enough $n$, since $m(w)\leq C$ uniformly in $w$ from part (a).
Finally, $\mathbb{P}(W^{2}>\tfrac{1}{2}V(w_{0})n)\leq Cn^{-1/(2\xi )}$ from (%
\ref{ineq_Fbar}).

(f) Apply part (d) with $\beta _{0}=(1/2-r_{k}(\xi ))/\xi +1/\xi >3$ for $%
\xi \in \lbrack 1/3,1/2)$ to obtain 
\begin{equation*}
\mathbb{E}[|\tilde{W}|^{3}\mathbf{1}[\tilde{W}^{2}\left. \leq V(w)n\right.
]]\leq \mathbb{E}[|\tilde{W}|^{3}]\leq Cw^{\beta _{0}}.
\end{equation*}
\end{proof}

\begin{lemma}
\label{lm:sig}For any $\epsilon >0$, $\mathbb{E}[A_{n}\mathbf{1}[\Delta
_{n}>Cn^{-r_{k}(\xi )+1/2-\xi }|\mathbf{W}^{R}]\leq Cn^{-r_{k}(\xi
)+\epsilon }(1+(n^{-\xi }W_{k}^{R})^{k/\xi -\epsilon }).$
\end{lemma}

\begin{proof}
We initially prove%
\begin{equation}
\mathbb{E}[A_{n}\mathbf{1}[|s_{n}^{2}-V(W_{k}^{R})|>n^{-r_{k}(\xi )+1/2-\xi
}]|\mathbf{W}^{R}]\leq Cn^{-r_{k}(\xi )+\epsilon }(1+(W_{k}^{R}/n^{\xi
})^{k/\xi -\epsilon }).  \label{eq:lmsig_ini1}
\end{equation}%
With $\tilde{W}$ as defined in Lemma \ref{lm:w}, note that by the $c_{r}$
inequality, for any $\beta >0$ 
\begin{eqnarray}
\mathbb{E}[|\tilde{W}^{2}-V(w)|^{\beta }] &\leq &C\mathbb{E}[|\tilde{W}%
|^{2\beta }]+CV(w)^{\beta }  \notag \\
&\leq &C\mathbb{E}[|\tilde{W}|^{2\beta }]  \label{ineq_wq}
\end{eqnarray}%
since $V(w)\leq 1$ and $\mathbb{E}[|\tilde{W}|^{2\beta }]>0$ uniformly in $%
w\geq w_{0}$. Let $\tilde{W}_{i}$, $i=1,\ldots ,n-k$ be i.i.d.~and
distributed like $\tilde{W}$, and define $\tilde{Q}_{i}=(n-k)^{-1}(\tilde{W}%
_{i}^{2}-V(w)).$ Note that $\mathbb{E}[\tilde{W}_{i}]=\mathbb{E}[\tilde{Q}%
_{i}]=0$. By Rosenthal's (1970)\nocite{Rosenthal70} inequality, for any $p>2$
\begin{equation*}
\mathbb{E}\left[ \left\vert \sum_{i=1}^{n-k}\tilde{Q}_{i}\right\vert ^{p}%
\right] \leq C(n-k)\mathbb{E}[|\tilde{Q}_{1}|^{p}]+C((n-k)\mathbb{E}[\tilde{Q%
}_{1}^{2}])^{p/2}.
\end{equation*}%
Application of (\ref{ineq_wq}) and Lemma \ref{lm:w} (d) yields, for $w\geq
w_{0}$, $n>k+1$ and any $p_{0}>p>2$%
\begin{eqnarray*}
(n-k)\mathbb{E}[|\tilde{Q}_{1}|^{p}] &\leq &Cn^{1-p}w^{2p_{0}-1/\xi } \\
&=&Cn^{2\xi p_{0}-p}(w/n^{\xi })^{2p_{0}-1/\xi } \\
((n-k)\mathbb{E}[\tilde{Q}_{1}^{2}])^{p/2} &\leq &Cn^{-p/2}\mathbb{E}[|%
\tilde{W}|^{4}]^{p/2} \\
&\leq &Cn^{-p/2}w^{2p_{0}-p_{0}/(2\xi )} \\
&\leq &Cn^{2\xi p_{0}-(p+p_{0})/2}(w/n^{\xi })^{2p_{0}-p_{0}/(2\xi )}
\end{eqnarray*}%
so that uniformly in $w\geq w_{0}$%
\begin{equation*}
\mathbb{E}\left[ \left\vert \sum_{i=1}^{n-k}\tilde{Q}_{i}\right\vert ^{p}%
\right] \leq Cn^{2\xi p_{0}-p}(1+(w/n^{\xi })^{2p_{0}-1/\xi }).
\end{equation*}%
By Markov's inequality, for any $\alpha \in 
\mathbb{R}
$%
\begin{equation*}
\mathbb{P}\left( \left\vert \sum_{i=1}^{n-k}\tilde{Q}_{i}\right\vert >\tfrac{%
1}{2}n^{\alpha }\right) \leq 2^{p}\frac{\mathbb{E}\left[ \left\vert
\sum_{i=1}^{n-k}\tilde{Q}_{i}\right\vert ^{p}\right] }{n^{p\alpha }}.
\end{equation*}%
Thus, with $\alpha =-r_{k}(\xi )+1/2-\xi $, $p_{0}=(k+1)/(2\xi )-\epsilon /2$
and $p=p_{0}-\epsilon /2$, we obtain from some algebra that%
\begin{eqnarray}
\mathbb{P}\left( \left\vert \sum_{i=1}^{n-k}Q_{i}\right\vert \left. >\right. 
\tfrac{1}{2}n^{-r_{k}(\xi )+1/2-\xi }\right) &\leq &Cn^{-r_{k}(\xi
)+\epsilon (3/2-2\xi -r_{k}(\xi ))}(1+(w/n^{\xi })^{k/\xi -\epsilon })
\label{ineq_sq} \\
&\leq &Cn^{-r_{k}(\xi )+\epsilon }(1+(w/n^{\xi })^{k/\xi -\epsilon })  \notag
\end{eqnarray}%
since $3/2-2\xi -r_{k}(\xi )\leq 1$ uniformly in $\xi \in \lbrack 1/3;1/2]$.
Furthermore, by Markov's inequality%
\begin{eqnarray}
\mathbb{P}\left( \left\vert (n-k)^{-1}\sum_{i=1}^{n-k}\tilde{W}%
_{i}\right\vert ^{2}\left. >\right. \tfrac{1}{2}n^{\alpha }\right) &\leq &2%
\frac{(n-k)^{-1}V(w)}{n^{\alpha }}  \label{ineq_sw} \\
&\leq &Cn^{\alpha -1}\leq Cn^{-r_{k}(\xi )+\epsilon }.  \notag
\end{eqnarray}

Now note that conditional on $\mathbf{W}^{R}$, $\{W_{i}^{s}+m(W_{k}^{R})%
\}_{i=1}^{n-k}$ has the same distribution as $\{\tilde{W}_{i}\}_{i=1}^{n-k}$
with $w=W_{k}^{R}$. Thus, conditional on $\mathbf{W}^{R}$, the distribution
of%
\begin{equation*}
s_{n}^{2}-V(W_{k}^{R})=(n-k)^{-1}%
\sum_{i=1}^{n-k}((W_{i}^{s}+m(W_{k}^{R}))^{2}-V(W_{k}^{R}))-\left(
(n-k)^{-1}\sum_{i=1}^{n-k}(W_{i}^{s}+m(W_{k}^{R}))\right) ^{2}
\end{equation*}%
is equal to the distribution of $\sum_{i=1}^{n-k}Q_{i}-\left(
(n-k)^{-1}\sum_{i=1}^{n-k}\tilde{W}_{i}\right) ^{2}$ for $w=W_{k}^{R}$, so (%
\ref{eq:lmsig_ini1}) follows from (\ref{ineq_sq}) and (\ref{ineq_sw}).

To conclude the proof of the lemma, note that $0<V(w)<\infty $ uniformly in $%
w\geq w_{0}$, so for a large enough finite $C$, $|1/s_{n}-1/\sqrt{%
V(W_{k}^{R})}|>Cn^{-r_{k}(\xi )+1/2-\xi }$ implies $%
|s_{n}^{2}-V(W_{k}^{R})|<n^{-r_{k}(\xi )+1/2-\xi }$. The result thus follows
from (\ref{eq:lmsig_ini1}) and $%
\sup_{w>w_{0}}|(1-k/n)^{1/2}-1|V(w)=O(n^{-1/2})$.
\end{proof}

\begin{lemma}
\label{lm:Hi}(a)\ $\mathbb{E[}A_{n}^{\ast }|\mathbf{1}[n^{-1/2+\xi }\sigma 
\mathbf{X}/\sqrt{V(n^{\xi }\sigma X_{k})}\left. \in \mathcal{H}_{j}]-\mathbf{%
1}[n^{-1/2+\xi }\sigma \mathbf{X}\in \mathcal{H}_{j}]|\right. ]\leq
Cn^{-r_{k}(\xi )}$;

(b) For all $\epsilon >0$, $|\mathbb{E}[A_{n}(\mathbf{1}[\mathbf{W}^{R}/%
\sqrt{(n-k)s_{n}^{2}}\in \mathcal{H}_{j}]-\mathbf{1}[n^{-1/2}\mathbf{W}^{R}/%
\sqrt{V(W_{k}^{R})}\in \mathcal{H}_{j}])]|\leq Cn^{-r_{k}(\xi )+\epsilon }$.
\end{lemma}

\begin{proof}
(a) By a first order Taylor expansion $|V(w)^{-1/2}-1|\leq C|1-V(w)|$
uniformly in $w\geq w_{0}$. For $\mathbf{s}\in 
\mathbb{R}
^{k}$ and $\mathcal{H}\subset 
\mathbb{R}
^{k}$, let $d(\mathbf{s},\mathcal{H})$ be the Euclidian distance of the
point $\mathbf{s}$ from the set $\mathcal{H}$. We have%
\begin{eqnarray*}
&&\mathbb{E[}A_{n}^{\ast }|\mathbf{1}[n^{-1/2+\xi }\sigma \mathbf{X}/\sqrt{%
V(n^{\xi }\sigma X_{k})}\left. \in \mathcal{H}_{j}]-\mathbf{1}[n^{-1/2+\xi
}\sigma \mathbf{X}\in \mathcal{H}_{j}]|]\right. \\
&\leq &\mathbb{E[}A_{n}^{\ast }\mathbf{1}[d(n^{-1/2+\xi }\sigma \mathbf{X}%
,\partial \mathcal{H}_{j})\leq Cn^{-1/2+\xi }||\mathbf{X}||\cdot |1-V(n^{\xi
}\sigma X_{k})|]] \\
&\leq &\mathbb{E[}A_{n}^{\ast }\mathbf{1}[d(n^{-1/2+\xi }\sigma \mathbf{X}%
,\partial \mathcal{H}_{j})\leq Cn^{-3/2+3\xi }||\mathbf{X}||\cdot
X_{k}^{2-1/\xi }]] \\
&\leq &\mathbb{E[}A_{n}^{\ast }\mathbf{1}[d(n^{-1/2+\xi }\sigma \mathbf{X}%
,\partial \mathcal{H}_{j})\leq Cn^{-3/2+3\xi }X_{1}^{3-1/\xi }]] \\
&\leq &\mathbb{E[}A_{n}^{\ast }\mathbf{1}[d(n^{-1/2+\xi }\sigma \mathbf{X}%
,\partial \mathcal{H}_{j})\leq Cn^{-r_{k}(\xi )}(1+X_{1})]]
\end{eqnarray*}%
where the second inequality follows from Lemma \ref{lm:w} (b), and the last
inequality holds because $-3/2+3\xi \leq -r_{k}(\xi )$ and $X_{1}^{3-1/\xi
}\leq X_{1}$ for all $X_{1}\geq 1$ for $\xi \in \lbrack 1/3,1/2]$.

Furthermore, with $\mathbf{U}=(U_{1},\ldots ,U_{k})^{\prime }\mathbf{=X/}%
X_{1}$, 
\begin{eqnarray*}
&&\mathbb{E}[\mathbf{1}[d(n^{-1/2+\xi }\sigma \mathbf{X},\partial \mathcal{H}%
_{j})\left. \leq \right. Cn^{-r_{k}(\xi )}X_{1}]] \\
&\leq &\mathbb{E}[\mathbb{E}[d(n^{-1/2+\xi }\sigma (X_{1}\mathbf{U}%
),\partial \mathcal{H}_{j})\leq Cn^{-r_{k}(\xi )}(1+X_{1})|\mathbf{U}]] \\
&=&\mathbb{E}[\sum_{s\in \mathcal{I}_{j}(\mathbf{U})}\mathbb{E}[\mathbf{1}%
[|n^{-1/2+\xi }\sigma X_{1}-s|\leq Cn^{-r_{k}(\xi )}(1+X_{1})|\mathbf{U}]]
\end{eqnarray*}%
A calculation shows that the conditional density of $X_{1}$ given $\mathbf{U}
$ is proportional to 
\begin{equation*}
u_{k}^{-k/\xi }s^{-k/\xi -1}\exp \left[ -(su_{k}/\sigma )^{-1/\xi }\right] 
\text{,}
\end{equation*}%
a decreasing function of $s$ for all large enough $s$, and the density of $%
U_{k}$ is bounded above by $Cu_{k}^{1/\xi -1}$. Thus, for all large enough $%
n $,%
\begin{eqnarray*}
&&\mathbb{E}[\sum_{s\in \mathcal{I}_{j}(\mathbf{U})}\mathbb{E}[\mathbf{1}%
[|n^{-1/2+\xi }\sigma X_{1}-s|\left. \leq Cn^{-r_{k}(\xi )}(1+X_{1})|\mathbf{%
U}]]\right. \\
&\leq &\mathbb{E}[L\mathbb{E}[\mathbf{1}[|n^{-1/2+\xi }\sigma
X_{1}-L^{-1}|\leq Cn^{-r_{k}(\xi )}(1+X_{1})|\mathbf{U}]] \\
&\leq &\mathbb{E}[L\mathbb{E}[\mathbf{1}[|n^{-1/2+\xi }\sigma
X_{1}-L^{-1}|\leq Cn^{-r_{k}(\xi )}X_{1}|\mathbf{U}]] \\
&\leq &C(n^{1/2-\xi })^{-k/\xi -1}n^{-r_{k}(\xi )}n^{1/2-\xi
}\int_{0}^{1}\exp \left[ -(L^{-1}n^{1/2-\xi }u_{k}/\sigma )^{-1/\xi }\right]
u_{k}^{-(k-1)/\xi -1}du_{k} \\
&\leq &Cn^{-r_{k}(\xi )-(1/2-\xi )k/\xi }\cdot n^{(1/2-\xi )(k-1)/\xi } \\
&=&Cn^{-r_{k}(\xi )-(1/2-\xi )/\xi }
\end{eqnarray*}%
and the result follows.

(b)\ Let $D_{n}=\mathbf{1}[\Delta _{n}\leq Cn^{-r_{k}(\xi )+1/2-\xi }]$.
Using Lemmas \ref{lm:sig} and \ref{lm:expecT} (a), we have for all $\epsilon
>0$%
\begin{eqnarray*}
\mathbb{E}[A_{n}(1-D_{n})] &=&\mathbb{E}[\mathbb{E}[A_{n}\mathbf{1}[\tilde{%
\Delta}_{n}>Cn^{-r_{k}(\xi )+1/2-\xi }]|\mathbf{W}^{R}]] \\
&\leq &Cn^{-r_{k}(\xi )+\epsilon }(1+\mathbb{E}[(n^{-\xi }W_{k}^{R})^{k/\xi
-\epsilon }]) \\
&\leq &Cn^{-r_{k}(\xi )+\epsilon }
\end{eqnarray*}%
so it suffices to show the claim with $A_{n}$ replaced by $A_{n}D_{n}$.

In the notation of the proof of part (a), we have 
\begin{eqnarray*}
&&\mathbb{E[}A_{n}D_{n}|\mathbf{1}[\mathbf{W}^{R}/\sqrt{(n-k)s_{n}^{2}}%
\left. \in \mathcal{H}_{j}]-\mathbf{1}[n^{-1/2}\mathbf{W}^{R}/\sqrt{%
V(W_{k}^{R})}\in \mathcal{H}_{j}]|]\right. \\
&\leq &\mathbb{E[}A_{n}D_{n}\mathbf{1}[d(n^{-1/2}\mathbf{W}^{R}/\sqrt{%
V(W_{k}^{R})},\partial \mathcal{H}_{j})\leq Cn^{-1/2}\Delta _{n}||\mathbf{W}%
^{R}||]] \\
&\leq &\mathbb{E[}A_{n}\mathbf{1}[d(n^{-1/2}\mathbf{W}^{R}/\sqrt{V(W_{k}^{R})%
},\partial \mathcal{H}_{j})\leq Cn^{-r_{k}(\xi )-\xi }W_{1}^{R}]] \\
&\leq &\mathbb{E[}A_{n}^{\ast }\mathbf{1}[d(n^{-1/2+\xi }\sigma \mathbf{X}/%
\sqrt{V(n^{\xi }\sigma X_{k})},\partial \mathcal{H}_{j})\leq Cn^{-r_{k}(\xi
)}\sigma X_{1}]]+n^{-\delta } \\
&\leq &\mathbb{E[}A_{n}^{\ast }\mathbf{1}[d(n^{-1/2+\xi }\sigma \mathbf{X}%
,\partial \mathcal{H}_{j})\leq Cn^{-r_{k}(\xi )}\sigma X_{1}]]+n^{-\delta
}+Cn^{-r_{k}(\xi )}
\end{eqnarray*}%
where the penultimate inequality follows from (\ref{tvd_conv}), and the last
inequality applies the result form part (a). The desired inequality now
follows from the same arguments as in the proof of part (a).
\end{proof}

\bigskip

\noindent \textbf{Proof of Theorem 2:}

Let $B_{n}=\mathbf{1}[W_{k}^{R}>w_{0}]\mathbf{1}[n^{-1/2}\mathbf{W}^{R}/%
\sqrt{V(W_{k}^{R})}\in \mathcal{H}_{j}]$ and $B_{n}^{\ast }=\mathbf{1}%
[n^{\xi }\sigma X_{k}>w_{0}]\mathbf{1}[n^{\xi -1/2}\sigma \mathbf{X}\in 
\mathcal{H}_{j}]$. Given the results in Lemmas \ref{lm:A} and \ref{lm:Hi}
(a), it suffices to show that%
\begin{multline*}
\left\vert \mathbb{E}B_{n}\mathbf{1}\left[ \frac{\sum_{i=1}^{n-k}W_{i}^{s}}{%
\sqrt{(n-k)s_{n}^{2}}}\leq b_{j}\left( \frac{\mathbf{W}^{R}}{\sqrt{%
(n-k)s_{n}^{2}}}\right) \right] \right. \\
\left. -\mathbb{E}B_{n}^{\ast }\mathbf{1}\left[ Z-n^{1/2}m^{\ast }(n^{\xi
}\sigma X_{k})\leq b_{j}\left( n^{\xi -1/2}\sigma \mathbf{X}\right) \right]
\right\vert \leq Cn^{-r_{k}(\xi )+\epsilon }
\end{multline*}%
for all $j=1,\ldots ,m_{\varphi }$.

Notice that conditional on $\mathbf{W}^{R}$,%
\begin{equation*}
\varsigma _{n}=\frac{\sum_{i=1}^{n-k}(W_{i}^{s}+m(W_{k}^{R}))}{\sqrt{%
(n-k)s_{n}^{2}}}
\end{equation*}%
has the same distribution as the t-statistic computed from the zero-mean
i.i.d.~sample $\tilde{W}_{1},\tilde{W}_{2},\ldots ,\tilde{W}_{n}$ with $%
\tilde{W}_{i}\sim \tilde{W}$ and $\tilde{W}$ defined in Lemma \ref{lm:w}
with $w=W_{k}^{R}$. Thus, by (\ref{BG_bound}), 
\begin{multline*}
A_{n}\sup_{s}|\mathbb{P}(\varsigma _{n}\leq s|\mathbf{W}^{R})-\Phi (x)|\leq
A_{n}CV(W_{k}^{R})^{-1}\mathbb{E}[\tilde{W}^{2}\mathbf{1}[\tilde{W}%
^{2}>V(W_{k}^{R})n] \\
+A_{n}Cn^{-1/2}V(W_{k}^{R})^{-3/2}\mathbb{E}[|\tilde{W}|^{3}\mathbf{1}[%
\tilde{W}^{2}\leq V(W_{k}^{R})n]].
\end{multline*}

Using $V(W_{k}^{R})\geq V(w_{0})>0$ if $A_{n}=1$ and applying Lemma \ref%
{lm:w} (e) and (f), the right-hand side is bounded above by 
\begin{eqnarray}
&&A_{n}C(n^{-\xi }W_{k}^{R})^{2}n^{2\xi -1/(2\xi )}+A_{n}C(n^{-\xi
}W_{k}^{R})^{(1/2-r_{k}(\xi ))/\xi }n^{-r_{k}(\xi )}\ \ \ \ \ \ \ \ \ \ \ \
\ \ \ \ \   \label{B1n} \\
&&\ \ \ \ \ \ \ \ \ \ \ \ \ \ \ \ \ \ \ \ \ \ \ \ \ \ \ \ \ \ \ \ \ \left.
\leq A_{n}Cn^{-r_{k}(\xi )}(1+(n^{-\xi
}W_{k}^{R})^{2}):=L_{1,n}(W_{k}^{R})\right.  \notag
\end{eqnarray}%
since $2\xi -1/(2\xi )\leq -r_{k}(\xi )$ and $0<(1/2-r_{k}(\xi ))/\xi \leq 1$
for all $\xi \in \lbrack 1/3,1/2]$ and $k>1.$

By the Lipschitz continuity of $b_{j}$ 
\begin{equation*}
\left\vert b_{j}\left( \frac{\mathbf{W}^{R}}{\sqrt{(n-k)s_{n}^{2}}}\right)
-b_{j}\left( \frac{n^{-1/2}\mathbf{W}^{R}}{\sqrt{V(W_{k}^{R})}}\right)
\right\vert \leq C\Delta _{n}n^{-1/2}||\mathbf{W}^{R}||\text{,}
\end{equation*}%
and defining 
\begin{eqnarray*}
\hat{M} &=&\frac{(n-k)m(W_{k}^{R})}{\sqrt{(n-k)s_{n}^{2}}}\text{, }\tilde{M}=%
\frac{(n-k)m(W_{k}^{R})}{n^{1/2}\sqrt{V(W_{k}^{R})}} \\
R_{n} &=&b_{j}(\mathbf{W}^{R}/\sqrt{(n-k)s_{n}^{2}})-b_{j}(n^{-1/2}\mathbf{W}%
^{R}/\sqrt{V(W_{k}^{R})})+\hat{M}-\tilde{M}
\end{eqnarray*}%
we have 
\begin{eqnarray*}
|A_{n}R_{n}| &\leq &A_{n}\Delta _{n}\left( n^{1/2}|m(W_{k}^{R})|+Cn^{-1/2}||%
\mathbf{W}^{R}||\right) \\
&\leq &\Delta _{n}A_{n}Cn^{\xi -1/2}((n^{-\xi }W_{k}^{R})^{1-1/\xi }+n^{-\xi
}||\mathbf{W}^{R}||):=\Delta _{n}L_{2,n}(\mathbf{W}^{R})
\end{eqnarray*}%
where the second inequality invoked Lemma \ref{lm:w} (a). In this notation%
\begin{equation*}
\mathbf{1}\left[ \frac{\sum_{i=1}^{n-k}W_{i}^{s}}{\sqrt{(n-k)s_{n}^{2}}}\leq
b_{j}\left( \frac{\mathbf{W}^{R}}{\sqrt{(n-k)s_{n}^{2}}}\right) \right] =%
\mathbf{\mathbf{1}}\left[ \varsigma _{n}\leq R_{n}+b_{j}\left( \frac{n^{-1/2}%
\mathbf{W}^{R}}{\sqrt{V(W_{k}^{R})}}\right) +\tilde{M}\right] .
\end{equation*}%
From Lemma \ref{lm:sig}, 
\begin{equation*}
\mathbb{E}[A_{n}\mathbf{1}[\Delta _{n}>Cn^{-r_{k}(\xi )+1/2-\xi }]|\mathbf{W}%
^{R}]\leq CA_{n}n^{-r_{k}(\xi )+\epsilon }(1+(n^{-\xi }W_{k}^{R})^{k/\xi
-\epsilon }):=L_{3,n}(W_{k}^{R}).
\end{equation*}%
Thus, uniformly in $s\in 
\mathbb{R}
$,%
\begin{eqnarray*}
&&\mathbb{E}[B_{n}\mathbf{1}[\left. \varsigma _{n}\leq \right. s+R_{n}|%
\mathbf{W}^{R}]] \\
&\leq &\mathbb{E}[B_{n}\mathbf{1}[\varsigma _{n}\leq s+R_{n}]\mathbf{1}%
[\Delta _{n}\leq Cn^{-r_{k}(\xi )+1/2-\xi }]|\mathbf{W}^{R}]+\mathbb{E}[A_{n}%
\mathbf{1}[\Delta _{n}>Cn^{-r_{k}(\xi )+1/2-\xi }]|\mathbf{W}^{R}] \\
&\leq &\mathbb{E}[B_{n}\mathbf{1}[\varsigma _{n}\leq s+Cn^{-r_{k}(\xi
)+1/2-\xi }L_{2,n}(\mathbf{W}^{R})|\mathbf{W}^{R}]+L_{3,n}(W_{k}^{R}) \\
&\leq &B_{n}\Phi (s+Cn^{-r_{k}(\xi )+1/2-\xi }L_{2,n}(\mathbf{W}%
^{R}))+L_{1,n}(W_{k}^{R})+L_{3,n}(W_{k}^{R}) \\
&\leq &B_{n}\Phi (s)+L_{1,n}(W_{k}^{R})+Cn^{-r_{k}(\xi )+1/2-\xi }L_{2,n}(%
\mathbf{W}^{R})+L_{3,n}(W_{k}^{R})
\end{eqnarray*}%
where the third inequality follows from (\ref{B1n}), and the fourth
inequality follows from an exact first order Taylor expansion and the fact
that the derivative of $\Phi $ is uniformly bounded. Thus, letting $s=b_{j}(%
\mathbf{W}^{R}/\sqrt{nV(W_{k}^{R})})+\tilde{M}$ and taking expectations, we
obtain%
\begin{eqnarray}
&&\mathbb{E}\left[ B_{n}\mathbf{1}\left[ \frac{\sum_{i=1}^{n-k}W_{i}^{s}}{%
\sqrt{(n-k)s_{n}^{2}}}\leq b_{j}\left( \frac{\mathbf{W}^{R}}{\sqrt{%
(n-k)s_{n}^{2}}}\right) \right] \right] -\mathbb{E}\left[ B_{n}\mathbf{1}%
\left[ Z\leq b_{j}\left( \frac{\mathbf{W}^{R}}{n^{1/2}\sqrt{V(W_{k}^{R})}}%
\right) +\tilde{M}\right] \right]  \notag  \label{ineq_expec} \\
&\leq &\mathbb{E}[L_{1,n}(W_{k}^{R})]+Cn^{-r_{k}(\xi )+1/2-\xi }\mathbb{E}%
[L_{2,n}(\mathbf{W}^{R})]+\mathbb{E}[L_{3,n}(W_{k}^{R})].
\label{ineq_expecx}
\end{eqnarray}%
Similarly, uniformly in $s\in 
\mathbb{R}
$,%
\begin{eqnarray*}
&&\mathbb{E}[B_{n}\mathbf{1}[\left. \varsigma _{n}\leq \right. s+R_{n}|%
\mathbf{W}^{R}]] \\
&\geq &\mathbb{E}[B_{n}\mathbf{1}[\varsigma _{n}\leq s+R_{n}]\mathbf{1}%
[\Delta _{n}\leq Cn^{-r_{k}(\xi )+1/2-\xi }]|\mathbf{W}^{R}]-\mathbb{E}[A_{n}%
\mathbf{1}[\Delta _{n}>Cn^{-r_{k}(\xi )+1/2-\xi }]|\mathbf{W}^{R}] \\
&\geq &B_{n}\Phi (s)-L_{1,n}(W_{k}^{R})-Cn^{-r_{k}(\xi )+1/2-\xi }L_{2,n}(%
\mathbf{W}^{R})-L_{3,n}(W_{k}^{R})
\end{eqnarray*}%
so that (\ref{ineq_expecx})\ holds with the left hand side replaced by its
absolute value. By an application of Lemma \ref{lm:expecT} (a) and (b), the
right hand side of (\ref{ineq_expecx}) is $O(n^{-r_{k}(\xi )+\epsilon })$.

Furthermore, with $B_{n}^{\ast \ast }=A_{n}^{\ast }\mathbf{1}[n^{\xi
-1/2}\sigma \mathbf{X}/\sqrt{V(n^{\xi }\sigma X_{k})}\in \mathcal{H}_{j}]$
and $\tilde{M}^{\ast }=n^{-1/2}(n-k)m(n^{\xi }\sigma X_{k})/\sqrt{V(n^{\xi
}\sigma X_{k})}$, by (\ref{tvd_conv}), 
\begin{eqnarray*}
&&\left\vert \mathbb{E}\left[ B_{n}\mathbf{1}[Z\leq b_{j}\left( \frac{%
\mathbf{W}^{R}}{n^{1/2}\sqrt{V(W_{k}^{R})}}\right) +\tilde{M}]\right] -%
\mathbb{E}\left[ B_{n}^{\ast \ast }\mathbf{1}[Z\leq b_{j}\left( \frac{n^{\xi
}\sigma \mathbf{X}}{n^{1/2}\sqrt{V(n^{\xi }\sigma X_{k})}}\right) +\tilde{M}%
^{\ast }\right] \right\vert \\
&\leq &Cn^{-\delta }.
\end{eqnarray*}%
By Lemma \ref{lm:Hi} (b), replacing $B_{n}^{\ast \ast }$ by $B_{n}^{\ast }$
in this expression yields an additional approximation error of order at most 
$O(n^{-r_{k}(\xi )+\epsilon })$.

By a first order Taylor expansion $|1-V(w)^{-1/2}|\leq C|1-V(w)|$ uniformly
in $w\geq w_{0}.$ Thus, by the assumption about $b_{j}$, and using again the
fact that the derivative of $\Phi $ is uniformly bounded, 
\begin{eqnarray*}
&&\left\vert \mathbb{E}\left[ B_{n}^{\ast }\Phi \left( b_{j}\left( \frac{%
n^{\xi }\sigma \mathbf{X}}{n^{1/2}\sqrt{V(n^{\xi }\sigma X_{k})}}\right) +%
\tilde{M}^{\ast }\right) -B_{n}^{\ast }\Phi \left( b_{j}\left( n^{\xi
-1/2}\sigma \mathbf{X}\right) +n^{1/2}m^{\ast }(n^{\xi }\sigma X_{k})\right) %
\right] \right\vert \\
&\leq &Cn^{-1/2}\mathbb{E}[B_{n}^{\ast }(||n^{\xi }\mathbf{X}%
||+(n-k)|m(n^{\xi }\sigma X_{k})|)|1-V(n^{\xi }\sigma X_{k})|] \\
&&\ \ \ \ \ \ \ \ \ \ \ \ \ \ \ \ \ \ \ \ \ \ \ \ \ \ \ \ \ \ \ \ \ \ \ \ \
\ \ \ \ \ \ \ \ \ \ \ \ \ +Cn^{1/2}\mathbb{E}[B_{n}^{\ast }|m(n^{\xi }\sigma
X_{k})-m^{\ast }(n^{\xi }\sigma X_{k})|].
\end{eqnarray*}%
By the Cauchy-Schwarz inequality and Lemma \ref{lm:w} (b),%
\begin{eqnarray}
n^{-1/2}\mathbb{E}[B_{n}^{\ast }||n^{\xi }\mathbf{X}||\cdot |1-V(n^{\xi
}\sigma X_{k})|] &\leq &n^{-1/2}\mathbb{E}[||n^{\xi }\mathbf{X}||^{2}]^{1/2}%
\mathbb{E}[B_{n}^{\ast }|1-V(n^{\xi }\sigma X_{k})|^{2}]^{1/2}  \notag \\
&\leq &Cn^{3(\xi -1/2)}\mathbb{E}[||\mathbf{X}||^{2}]^{1/2}\mathbb{E}%
[X_{k}^{4-2/\xi }]^{1/2}  \label{fin1}
\end{eqnarray}%
and by the Cauchy-Schwarz inequality and Lemma \ref{lm:w} (a) and (b),%
\begin{eqnarray}
n^{1/2}\mathbb{E}[B_{n}^{\ast }|m(n^{\xi }\sigma X_{k})|\cdot |1-V(n^{\xi
}\sigma X_{k})|] &\leq &n^{1/2}\mathbb{E}[B_{n}^{\ast }|m(n^{\xi }\sigma
X_{k})|^{2}]^{1/2}\mathbb{E}[B_{n}^{\ast }|1-V(n^{\xi }\sigma
X_{k})|^{2}]^{1/2}  \notag \\
&\leq &Cn^{3(\xi -1/2)}\mathbb{E}[X_{k}^{2-2/\xi }]^{1/2}\mathbb{E}%
[X_{k}^{4-2/\xi }]^{1/2}.  \label{fin2}
\end{eqnarray}%
Finally, by Lemma \ref{lm:w} (c),%
\begin{equation}
n^{1/2}\mathbb{E}[B_{n}^{\ast }|m(n^{\xi }\sigma X_{k})-m^{\ast }(n^{\xi
}\sigma X_{k})|]\leq Cn^{-1/2+\xi -\delta }\mathbb{E}[X_{k}^{1-(1+\delta
)/\xi }]+Cn^{\xi -3/2}\mathbb{E}[X_{k}^{1-2/\xi }].  \label{fin3}
\end{equation}%
Note that $3(\xi -1/2)$, $\xi -3/2$ and $-1/2+\xi -\delta $ for $\delta \geq
r_{k}(\xi )$ are weakly smaller than $-r_{k}(\xi )$ for all $\xi \in \lbrack
1/3,1/2]$, so the result follows from applying Lemma \ref{lm:expecT} (c) and
(d) to (\ref{fin1})-(\ref{fin3}).

\subsection{\label{sec:gen2tails}Generalizing Theorem \protect\ref{thm_main}
to Two Potentially Heavy Tails}

\begin{condition}
\label{cnd:2tail}Suppose for some $\xi ^{R},\sigma ^{R},\xi ^{L},\sigma
^{L},\delta ,w_{0}>0$, $F$ admits a density for $w>w_{0}$ of the form%
\begin{equation*}
f^{R}(w)=(\xi ^{R}\sigma ^{R})^{-1}(\frac{w}{\sigma ^{R}})^{-1/\xi
^{R}-1}(1+h^{R}(w))
\end{equation*}%
and a density for $w<-w_{0}$ of the form%
\begin{equation*}
f^{L}(w)=(\xi ^{L}\sigma ^{L})^{-1}(\frac{-w}{\sigma ^{L}})^{-1/\xi
^{L}-1}(1+h^{L}(-w))
\end{equation*}%
with $|h^{J}(w)|$ uniformly bounded by $Cw^{-\delta /\xi ^{J}}$ for $J\in
\{L,R\}$ and some finite $C.$
\end{condition}

\begin{theorem}
\label{thm:2tail}Suppose Condition \ref{cnd:2tail} holds, and for $k>1$, $%
r_{k}(\xi )=\frac{3(1+k)(1-2\xi )}{2(1+k+2\xi )}\leq \delta $ where $\xi
=\max (\xi ^{L},\xi ^{R})$. Let $\varphi :%
\mathbb{R}
^{2k+1}\mapsto \{0,1\}$ be such that for some finite $m_{\varphi }$, $%
\varphi :%
\mathbb{R}
^{2k+1}\mapsto \{0,1\}$ can be written as an affine function of $\{\varphi
_{j}\}_{j=1}^{m_{\varphi }}$, where each $\varphi _{j}$ is of the form 
\begin{equation*}
\varphi _{j}(\mathbf{y}^{R},\mathbf{y}^{L},y_{0})=\mathbf{1}[(\mathbf{y}^{R},%
\mathbf{y}^{L})\in \mathcal{H}_{j}]\mathbf{1}[y_{0}\leq b_{j}(\mathbf{y}^{R},%
\mathbf{y}^{L})]
\end{equation*}%
with $b_{j}:%
\mathbb{R}
^{2k}\mapsto 
\mathbb{R}
$ a Lipschitz continuous function and $\mathcal{H}_{j}$ a Borel measurable
subset of $%
\mathbb{R}
^{2k}$ with boundary $\partial \mathcal{H}_{j}$. For $\mathbf{u}%
^{J}=(1,u_{2},\ldots ,u_{k})^{\prime }\in 
\mathbb{R}
^{k}$ with $1\geq u_{2}\geq u_{3}\geq \ldots \geq u_{k}$, let $\mathcal{I}%
_{j}(\mathbf{u}^{R},\mathbf{u}^{L})=\{s^{R},s^{L}>0:(s^{R}\mathbf{u}%
^{R},s^{L}\mathbf{u}^{L})\in \partial \mathcal{H}_{j}\}$. Assume further
that for some $L>0$, and Lebesgue almost all $\mathbf{u}^{R},\mathbf{u}^{L},$
$\mathcal{I}_{j}(\mathbf{u}^{R},\mathbf{u}^{L})$ contains at most $L$
elements in the interval $[L^{-1},\infty )^{2}$.

Then under $H_{0}:\mu =0$, for $\frac{1+k}{1+3k}<\xi <1/2$ and any $\epsilon
>0$ 
\begin{equation*}
|\mathbb{E}[\varphi (\mathbf{Y}_{n})]-\mathbb{E}[\varphi (\mathbf{Y}%
_{n}^{\ast })]|\leq Cn^{-r_{k}(\xi )+\epsilon }
\end{equation*}%
where $\mathbf{Y}_{n}$ and $\mathbf{Y}_{n}^{\ast }$ are the l.h.s.~and
r.h.s.~of (\ref{Y_twosided}), respectively.
\end{theorem}

The proof of Theorem \ref{thm:2tail} follows from the same steps as Theorem %
\ref{thm_main} and is omitted for brevity.

Tests of the form described in Section 3 reject if all of the following four
conditions hold 
\begin{eqnarray*}
h_{1}(\mathbf{y}^{\ast }) &=&\sum_{i=1}^{M}\lambda _{i}\frac{f(\mathbf{y}%
^{\ast }|\theta _{i},0)}{f_{a}^{S}(\mathbf{y}^{R})f_{a}^{S}(\mathbf{y}^{L})}%
\leq 1 \\
h_{2}(\mathbf{y}^{\ast }) &=&\exp [-5\chi (\mathbf{y}^{L})]%
\sum_{i=1}^{M^{S}}\lambda _{i}^{S}\frac{f^{S}(\mathbf{y}^{R}\mathbf{,y}%
^{L},y_{0}|\theta _{i}^{S})}{f_{a}^{S}(\mathbf{y}^{R})}\leq 1 \\
h_{3}(\mathbf{y}^{\ast }) &=&\exp [-5\chi (\mathbf{y}^{R})]%
\sum_{i=1}^{M^{S}}\lambda _{i}^{S}\frac{f^{S}(\mathbf{y}^{L}\mathbf{,y}%
^{R},y_{0}|\theta _{i}^{S})}{f_{a}^{S}(\mathbf{y}^{L})}\leq 1 \\
h_{4}(\mathbf{y}^{\ast }) &=&\frac{(y_{0}+\sum_{i=1}^{k}y_{i}^{R}-%
\sum_{i=1}^{k}y_{i}^{L})^{2}}{\func{cv}_{T}(\mathbf{y}^{\ast
})^{2}(1+\sum_{i=1}^{k}(y_{i}^{R})^{2}+\sum_{i=1}^{k}(y_{i}^{L})^{2})}\leq 1.
\end{eqnarray*}%
Inspection of the densities $f$ in (\ref{dens_joint}) and $f^{S}$ in (\ref%
{densSapprox}) shows that $h_{j}(\mathbf{y}^{\ast })$, $j=1,2,3$, viewed as
function of $y_{0}$, are a linear combination of normal densities, with
weights, means and variances a function of $(\mathbf{y}^{R},\mathbf{y}^{L})$%
. For any given $(\mathbf{y}^{R},\mathbf{y}^{L})$, the acceptance region of $%
\varphi (\mathbf{y}^{\ast })$, $\{y_{0}:\varphi (\mathbf{y}^{\ast })=0$ for $%
\mathbf{y}^{\ast }=(\mathbf{y}^{R},\mathbf{y}^{L},y_{0})\}$, is thus a
finite union of intervals, which an be written as an affine function of
appropriately defined $\varphi _{j}$. By the inverse function theorem, the
endpoints of these intervals are a function of $(\mathbf{y}^{R},\mathbf{y}%
^{L})$ with derivative equal to $-1/(\partial h_{j}(\mathbf{y}^{\ast
})/\partial y_{0})=0$ at $\mathbf{y}^{\ast }$ with $h_{j}(\mathbf{y}^{\ast
})=1$. Define the sets 
\begin{equation*}
\mathcal{E}_{j}(\varepsilon )=\{(\mathbf{y}^{R},\mathbf{y}^{L}):\exists
y_{0}\in 
\mathbb{R}
\text{ such that }h_{j}(\mathbf{y}^{\ast })=1\text{ and }|\partial h_{j}(%
\mathbf{y}^{\ast })/\partial y_{0}|<\varepsilon \}\in 
\mathbb{R}
^{2k}
\end{equation*}%
for $\varepsilon >0$, $j=1,\ldots ,4$, which have Lebesgue measure that can
be made arbitrarily small for small enough $\varepsilon $. Furthermore, for
any $\theta ^{S}=(\mu ,\sigma ,\xi )\in 
\mathbb{R}
^{3}$ and $c>0$, the equation $f_{T}(s\mathbf{u}^{J}|\theta ^{S})/f_{a}^{S}(s%
\mathbf{u}^{J})=c$ for a given $\mathbf{u}^{J}$ has at most two roots in $%
s>0 $. The modified test $\varphi _{\varepsilon }(\mathbf{y}^{\ast
})=\varphi (\mathbf{y}^{\ast })\mathbf{1}[(\mathbf{y}^{R},\mathbf{y}%
^{L})\notin \bigcup_{j=1}^{3}\mathcal{E}_{j}(\varepsilon )]$ is hence seen
to satisfy the assumptions of Theorem \ref{thm:2tail}, yet for small enough $%
\varepsilon $, $\varphi _{\varepsilon }$ and $\varphi $ have nearly
indistinguishable rejection probabilities.

\subsection{~Small Sample Results for $\protect\alpha =0.01$}

\renewcommand\floatpagefraction{1.0}
\renewcommand\topfraction{1}
\renewcommand\textfraction{0}%
\begin{table}[h!] 
\begin{center}
\small%
\caption{Small Sample Results in Inference for the Mean}\bigskip

$%
\begin{tabular}{lrrrrrrr}
\hline\hline
& N(0,1) & LogN & F(4,5) & \ \ t(3) & P(0.4) & Mix 1 & Mix 2 \\ \hline
\multicolumn{8}{c}{$n=50$} \\ \hline
\textsc{t-stat} & 0.9\TEXTsymbol{\vert}\textbf{1.00} & 5.1\TEXTsymbol{\vert}%
0.62 & 6.8\TEXTsymbol{\vert}0.58 & 0.7\TEXTsymbol{\vert}\textbf{1.04} & 7.8%
\TEXTsymbol{\vert}0.53 & 2.7\TEXTsymbol{\vert}0.82 & 10.8\TEXTsymbol{\vert}%
0.60 \\ 
\textsc{sym-boot} & 0.9\TEXTsymbol{\vert}\textbf{1.02} & 2.9\TEXTsymbol{\vert%
}1.04 & 4.0\TEXTsymbol{\vert}1.31 & 0.5\TEXTsymbol{\vert}\textbf{1.19} & 4.1%
\TEXTsymbol{\vert}1.29 & 2.6\TEXTsymbol{\vert}1.12 & 10.3\TEXTsymbol{\vert}%
1.58 \\ 
\textsc{asym-boot} & 1.0\TEXTsymbol{\vert}\textbf{1.02} & 2.1\TEXTsymbol{%
\vert}0.85 & 2.6\TEXTsymbol{\vert}0.97 & 1.9\TEXTsymbol{\vert}\textbf{1.10}
& 2.7\TEXTsymbol{\vert}0.95 & 2.6\TEXTsymbol{\vert}0.96 & 9.4\TEXTsymbol{%
\vert}1.08 \\ 
\textsc{new default} & 0.3\TEXTsymbol{\vert}\textbf{1.39} & 0.6\TEXTsymbol{%
\vert}\textbf{0.79} & 1.1\TEXTsymbol{\vert}\textbf{0.69} & 0.2\TEXTsymbol{%
\vert}\textbf{1.57} & 1.8\TEXTsymbol{\vert}\textbf{0.63} & 0.5\TEXTsymbol{%
\vert}\textbf{1.10} & 3.9\TEXTsymbol{\vert}0.70 \\ \hline
\multicolumn{8}{c}{$n=100$} \\ \hline
\textsc{t-stat} & 1.1\TEXTsymbol{\vert}\textbf{0.99} & 3.6\TEXTsymbol{\vert}%
0.74 & 5.3\TEXTsymbol{\vert}0.65 & 0.8\TEXTsymbol{\vert}\textbf{1.02} & 5.9%
\TEXTsymbol{\vert}0.62 & 2.6\TEXTsymbol{\vert}0.82 & 9.3\TEXTsymbol{\vert}%
0.57 \\ 
\textsc{sym-boot} & 1.1\TEXTsymbol{\vert}\textbf{0.99} & 2.2\TEXTsymbol{\vert%
}1.05 & 3.5\TEXTsymbol{\vert}1.22 & 0.7\TEXTsymbol{\vert}\textbf{1.09} & 3.5%
\TEXTsymbol{\vert}1.25 & 2.4\TEXTsymbol{\vert}1.03 & 8.5\TEXTsymbol{\vert}%
1.23 \\ 
\textsc{asym-boot} & 1.1\TEXTsymbol{\vert}\textbf{0.99} & 1.6\TEXTsymbol{%
\vert}\textbf{0.90} & 2.4\TEXTsymbol{\vert}0.94 & 1.9\TEXTsymbol{\vert}%
\textbf{1.04} & 2.5\TEXTsymbol{\vert}0.95 & 2.4\TEXTsymbol{\vert}0.92 & 7.3%
\TEXTsymbol{\vert}0.90 \\ 
\textsc{new default} & 0.7\TEXTsymbol{\vert}\textbf{1.08} & 0.2\TEXTsymbol{%
\vert}\textbf{1.22} & 0.5\TEXTsymbol{\vert}\textbf{0.97} & 0.4\TEXTsymbol{%
\vert}\textbf{1.67} & 0.5\TEXTsymbol{\vert}\textbf{0.92} & 0.5\TEXTsymbol{%
\vert}\textbf{1.31} & 3.7\TEXTsymbol{\vert}0.73 \\ \hline
\multicolumn{8}{c}{$n=500$} \\ \hline
\textsc{t-stat} & 1.0\TEXTsymbol{\vert}\textbf{1.00} & 1.9\TEXTsymbol{\vert}%
\textbf{0.88} & 3.0\TEXTsymbol{\vert}0.80 & 0.8\TEXTsymbol{\vert}\textbf{1.02%
} & 3.5\TEXTsymbol{\vert}0.78 & 1.5\TEXTsymbol{\vert}\textbf{0.93} & 4.2%
\TEXTsymbol{\vert}0.69 \\ 
\textsc{sym-boot} & 1.1\TEXTsymbol{\vert}\textbf{1.00} & 1.5\TEXTsymbol{\vert%
}\textbf{0.99} & 2.1\TEXTsymbol{\vert}1.09 & 0.8\TEXTsymbol{\vert}\textbf{%
1.04} & 2.6\TEXTsymbol{\vert}1.13 & 1.4\TEXTsymbol{\vert}\textbf{1.01} & 3.0%
\TEXTsymbol{\vert}1.04 \\ 
\textsc{asym-boot} & 1.1\TEXTsymbol{\vert}\textbf{1.00} & 1.3\TEXTsymbol{%
\vert}\textbf{0.93} & 1.7\TEXTsymbol{\vert}\textbf{0.94} & 1.6\TEXTsymbol{%
\vert}\textbf{1.02} & 1.8\TEXTsymbol{\vert}\textbf{0.95} & 1.4\TEXTsymbol{%
\vert}\textbf{0.97} & 2.3\TEXTsymbol{\vert}0.86 \\ 
\textsc{new default} & 1.0\TEXTsymbol{\vert}\textbf{1.01} & 0.6\TEXTsymbol{%
\vert}\textbf{1.33} & 0.7\TEXTsymbol{\vert}\textbf{1.34} & 0.5\TEXTsymbol{%
\vert}\textbf{1.32} & 0.7\TEXTsymbol{\vert}\textbf{1.30} & 0.5\TEXTsymbol{%
\vert}\textbf{1.37} & 0.4\TEXTsymbol{\vert}\textbf{1.14}%
\end{tabular}%
$

\end{center}
\linespread{1.00}\selectfont
\begin{small}
\linespread{1.00}\selectfont%

Notes: Entries are the null rejection probability in percent, and the
average length of confidence intervals relative to average length of
confidence intervals based on size corrected t-statistic (bold if null
rejection probability is smaller than 2\%) of nominal 1\% level tests. Based
on 20,000 replications.%
\linespread{1.00}\selectfont
\end{small}
\linespread{1.00}\selectfont%
\end{table}%
\begin{table}[tbp] 
\begin{center}
\small%

\caption{Small Sample Results of New Methods for Inference for the  Mean}

$%
\begin{tabular}{lrrrrrrr}
\hline\hline
& N(0,1) & LogN & F(4,5) & \ \ t(3) & P(0.4) & Mix 1 & Mix 2 \\ \hline
\multicolumn{8}{c}{$n=25$} \\ \hline
\textsc{def}: $k=8,n_{0}=50$ & 0.6\TEXTsymbol{\vert}\textbf{1.07} & 6.7%
\TEXTsymbol{\vert}0.57 & 8.9\TEXTsymbol{\vert}0.50 & 0.4\TEXTsymbol{\vert}%
\textbf{1.10} & 10.3\TEXTsymbol{\vert}0.46 & 2.1\TEXTsymbol{\vert}0.84 & 6.6%
\TEXTsymbol{\vert}0.68 \\ 
$k=4,n_{0}=50$ & 0.3\TEXTsymbol{\vert}\textbf{1.34} & 3.4\TEXTsymbol{\vert}%
0.60 & 4.9\TEXTsymbol{\vert}0.53 & 0.3\TEXTsymbol{\vert}\textbf{1.22} & 6.4%
\TEXTsymbol{\vert}0.47 & 1.2\TEXTsymbol{\vert}\textbf{0.88} & 3.0\TEXTsymbol{%
\vert}0.70 \\ 
$k=12,n_{0}=50$ & 0.0\TEXTsymbol{\vert}\textbf{0.00} & 0.0\TEXTsymbol{\vert}%
\textbf{0.00} & 0.0\TEXTsymbol{\vert}\textbf{0.00} & 0.0\TEXTsymbol{\vert}%
\textbf{0.00} & 0.0\TEXTsymbol{\vert}\textbf{0.00} & 0.0\TEXTsymbol{\vert}%
\textbf{0.00} & 0.0\TEXTsymbol{\vert}\textbf{0.00} \\ 
$k=4,n_{0}=25$ & 0.1\TEXTsymbol{\vert}\textbf{1.71} & 1.8\TEXTsymbol{\vert}%
\textbf{0.65} & 2.3\TEXTsymbol{\vert}0.58 & 0.1\TEXTsymbol{\vert}\textbf{1.54%
} & 3.0\TEXTsymbol{\vert}0.53 & 0.5\TEXTsymbol{\vert}\textbf{1.04} & 1.4%
\TEXTsymbol{\vert}\textbf{0.77} \\ \hline
\multicolumn{8}{c}{$n=50$} \\ \hline
\textsc{def}: $k=8,n_{0}=50$ & 0.3\TEXTsymbol{\vert}\textbf{1.39} & 0.6%
\TEXTsymbol{\vert}\textbf{0.79} & 1.1\TEXTsymbol{\vert}\textbf{0.69} & 0.2%
\TEXTsymbol{\vert}\textbf{1.57} & 1.8\TEXTsymbol{\vert}\textbf{0.63} & 0.5%
\TEXTsymbol{\vert}\textbf{1.10} & 3.9\TEXTsymbol{\vert}0.70 \\ 
$k=4,n_{0}=50$ & 0.3\TEXTsymbol{\vert}\textbf{1.50} & 0.5\TEXTsymbol{\vert}%
\textbf{0.83} & 0.9\TEXTsymbol{\vert}\textbf{0.69} & 0.2\TEXTsymbol{\vert}%
\textbf{1.65} & 1.3\TEXTsymbol{\vert}\textbf{0.64} & 0.5\TEXTsymbol{\vert}%
\textbf{1.15} & 3.4\TEXTsymbol{\vert}0.70 \\ 
$k=12,n_{0}=50$ & 0.1\TEXTsymbol{\vert}\textbf{1.50} & 1.4\TEXTsymbol{\vert}%
\textbf{0.75} & 2.2\TEXTsymbol{\vert}0.65 & 0.2\TEXTsymbol{\vert}\textbf{1.54%
} & 3.3\TEXTsymbol{\vert}0.59 & 0.4\TEXTsymbol{\vert}\textbf{1.03} & 2.6%
\TEXTsymbol{\vert}0.70 \\ 
$k=4,n_{0}=25$ & 0.4\TEXTsymbol{\vert}\textbf{1.59} & 0.3\TEXTsymbol{\vert}%
\textbf{0.97} & 0.5\TEXTsymbol{\vert}\textbf{0.78} & 0.2\TEXTsymbol{\vert}%
\textbf{2.02} & 0.6\TEXTsymbol{\vert}\textbf{0.73} & 0.5\TEXTsymbol{\vert}%
\textbf{1.34} & 3.5\TEXTsymbol{\vert}0.79 \\ \hline
\multicolumn{8}{c}{$n=100$} \\ \hline
\textsc{def}: $k=8,n_{0}=50$ & 0.7\TEXTsymbol{\vert}\textbf{1.08} & 0.2%
\TEXTsymbol{\vert}\textbf{1.22} & 0.5\TEXTsymbol{\vert}\textbf{0.97} & 0.4%
\TEXTsymbol{\vert}\textbf{1.67} & 0.5\TEXTsymbol{\vert}\textbf{0.92} & 0.5%
\TEXTsymbol{\vert}\textbf{1.31} & 3.7\TEXTsymbol{\vert}0.73 \\ 
$k=4,n_{0}=50$ & 0.6\TEXTsymbol{\vert}\textbf{1.22} & 0.3\TEXTsymbol{\vert}%
\textbf{1.06} & 0.6\TEXTsymbol{\vert}\textbf{0.85} & 0.3\TEXTsymbol{\vert}%
\textbf{1.85} & 0.7\TEXTsymbol{\vert}\textbf{0.81} & 0.5\TEXTsymbol{\vert}%
\textbf{1.27} & 2.2\TEXTsymbol{\vert}0.69 \\ 
$k=12,n_{0}=50$ & 0.6\TEXTsymbol{\vert}\textbf{1.14} & 0.3\TEXTsymbol{\vert}%
\textbf{1.22} & 0.4\TEXTsymbol{\vert}\textbf{1.00} & 0.5\TEXTsymbol{\vert}%
\textbf{1.56} & 0.5\TEXTsymbol{\vert}\textbf{0.92} & 0.3\TEXTsymbol{\vert}%
\textbf{1.35} & 3.3\TEXTsymbol{\vert}0.75 \\ 
$k=4,n_{0}=25$ & 0.7\TEXTsymbol{\vert}\textbf{1.20} & 0.4\TEXTsymbol{\vert}%
\textbf{1.21} & 0.5\TEXTsymbol{\vert}\textbf{0.98} & 0.3\TEXTsymbol{\vert}%
\textbf{2.13} & 0.6\TEXTsymbol{\vert}\textbf{0.95} & 0.4\TEXTsymbol{\vert}%
\textbf{1.47} & 2.5\TEXTsymbol{\vert}0.76 \\ \hline
\multicolumn{8}{c}{$n=500$} \\ \hline
\textsc{def}: $k=8,n_{0}=50$ & 1.0\TEXTsymbol{\vert}\textbf{1.01} & 0.6%
\TEXTsymbol{\vert}\textbf{1.33} & 0.7\TEXTsymbol{\vert}\textbf{1.34} & 0.5%
\TEXTsymbol{\vert}\textbf{1.32} & 0.7\TEXTsymbol{\vert}\textbf{1.30} & 0.5%
\TEXTsymbol{\vert}\textbf{1.37} & 0.4\TEXTsymbol{\vert}\textbf{1.14} \\ 
$k=4,n_{0}=50$ & 1.1\TEXTsymbol{\vert}\textbf{0.99} & 0.6\TEXTsymbol{\vert}%
\textbf{1.45} & 0.7\TEXTsymbol{\vert}\textbf{1.20} & 0.5\TEXTsymbol{\vert}%
\textbf{1.72} & 0.8\TEXTsymbol{\vert}\textbf{1.20} & 0.6\TEXTsymbol{\vert}%
\textbf{1.49} & 0.5\TEXTsymbol{\vert}\textbf{0.95} \\ 
$k=12,n_{0}=50$ & 1.1\TEXTsymbol{\vert}\textbf{0.99} & 0.5\TEXTsymbol{\vert}%
\textbf{1.29} & 0.7\TEXTsymbol{\vert}\textbf{1.20} & 0.5\TEXTsymbol{\vert}%
\textbf{1.30} & 0.6\TEXTsymbol{\vert}\textbf{1.22} & 0.5\TEXTsymbol{\vert}%
\textbf{1.31} & 0.4\TEXTsymbol{\vert}\textbf{1.23} \\ 
$k=4,n_{0}=25$ & 0.9\TEXTsymbol{\vert}\textbf{1.01} & 0.6\TEXTsymbol{\vert}%
\textbf{1.54} & 0.7\TEXTsymbol{\vert}\textbf{1.34} & 0.5\TEXTsymbol{\vert}%
\textbf{1.84} & 0.7\TEXTsymbol{\vert}\textbf{1.34} & 0.6\TEXTsymbol{\vert}%
\textbf{1.64} & 0.5\TEXTsymbol{\vert}\textbf{1.08}%
\end{tabular}%
$

\end{center}
\linespread{1.00}\selectfont
\begin{small}
\linespread{1.00}\selectfont%

Notes: See Table 5.%
\linespread{1.00}\selectfont
\end{small}
\linespread{1.00}\selectfont%
\end{table}%

\begin{figure}[h] 
\begin{center}
\small%
\caption{Small Sample Results for HMDA Populations}\bigskip
\includegraphics[width=5.8356in,keepaspectratio]{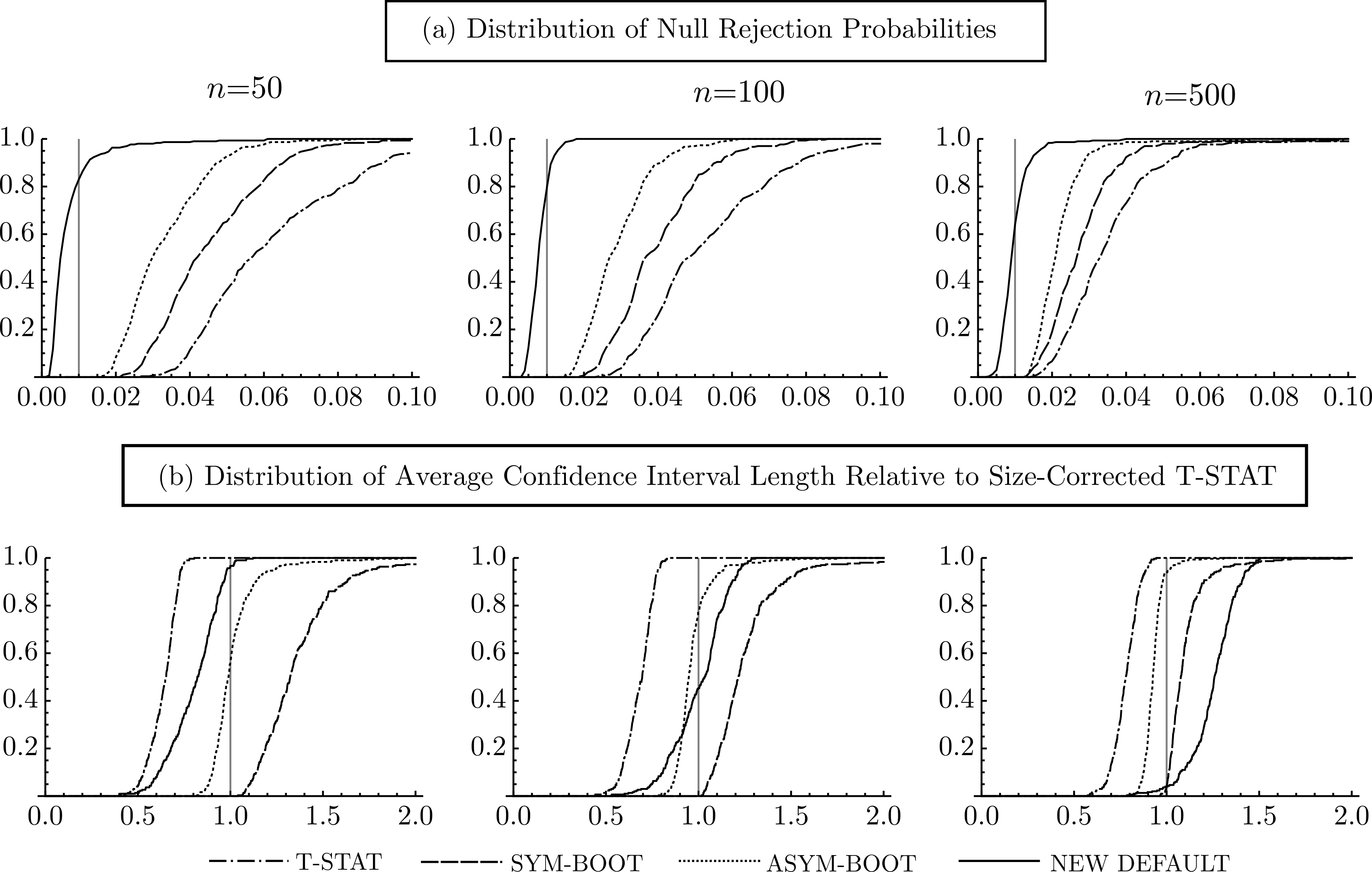}

\end{center}
\begin{small}%
\end{small}%
\end{figure}%

\begin{table}[h!] 
\begin{center}
\small%
\caption{Small Sample Results for Difference of Population Means}\bigskip

$%
\begin{tabular}{lrrrrrrr}
\hline\hline
& N(0,1) & LogN & F(4,5) & \ \ t(3) & P(0.4) & Mix 1 & Mix 2 \\ \hline
\multicolumn{8}{c}{$n=50$} \\ \hline
\textsc{t-stat} & 1.4\TEXTsymbol{\vert}\textbf{0.94} & 3.9\TEXTsymbol{\vert}%
0.75 & 3.3\TEXTsymbol{\vert}0.78 & 1.0\TEXTsymbol{\vert}\textbf{0.99} & 3.4%
\TEXTsymbol{\vert}0.78 & 2.1\TEXTsymbol{\vert}0.87 & 3.1\TEXTsymbol{\vert}%
0.81 \\ 
\textsc{sym-boot} & 1.4\TEXTsymbol{\vert}\textbf{0.95} & 3.6\TEXTsymbol{\vert%
}0.97 & 3.1\TEXTsymbol{\vert}1.02 & 0.9\TEXTsymbol{\vert}\textbf{1.08} & 3.3%
\TEXTsymbol{\vert}1.05 & 2.1\TEXTsymbol{\vert}1.02 & 3.1\TEXTsymbol{\vert}%
1.12 \\ 
\textsc{asym-boot} & 1.6\TEXTsymbol{\vert}\textbf{0.94} & 3.2\TEXTsymbol{%
\vert}0.83 & 3.1\TEXTsymbol{\vert}0.87 & 2.3\TEXTsymbol{\vert}1.01 & 3.2%
\TEXTsymbol{\vert}0.89 & 2.7\TEXTsymbol{\vert}0.92 & 3.2\TEXTsymbol{\vert}%
0.93 \\ 
\textsc{new default} & 0.1\TEXTsymbol{\vert}\textbf{1.46} & 0.7\TEXTsymbol{%
\vert}\textbf{1.02} & 0.7\TEXTsymbol{\vert}\textbf{1.07} & 0.1\TEXTsymbol{%
\vert}\textbf{1.49} & 0.7\TEXTsymbol{\vert}\textbf{1.07} & 0.5\TEXTsymbol{%
\vert}\textbf{1.25} & 0.9\TEXTsymbol{\vert}\textbf{1.06} \\ \hline
\multicolumn{8}{c}{$n=100$} \\ \hline
\textsc{t-stat} & 1.0\TEXTsymbol{\vert}\textbf{0.99} & 3.1\TEXTsymbol{\vert}%
0.78 & 3.1\TEXTsymbol{\vert}0.80 & 1.0\TEXTsymbol{\vert}\textbf{1.00} & 3.0%
\TEXTsymbol{\vert}0.81 & 1.8\TEXTsymbol{\vert}\textbf{0.89} & 3.6\TEXTsymbol{%
\vert}0.80 \\ 
\textsc{sym-boot} & 1.1\TEXTsymbol{\vert}\textbf{0.99} & 2.8\TEXTsymbol{\vert%
}0.99 & 3.0\TEXTsymbol{\vert}1.08 & 0.9\TEXTsymbol{\vert}\textbf{1.05} & 2.9%
\TEXTsymbol{\vert}1.10 & 1.8\TEXTsymbol{\vert}\textbf{1.03} & 3.6\TEXTsymbol{%
\vert}1.15 \\ 
\textsc{asym-boot} & 1.1\TEXTsymbol{\vert}\textbf{0.99} & 2.5\TEXTsymbol{%
\vert}0.87 & 2.6\TEXTsymbol{\vert}0.92 & 2.0\TEXTsymbol{\vert}1.01 & 2.5%
\TEXTsymbol{\vert}0.93 & 2.1\TEXTsymbol{\vert}0.94 & 3.6\TEXTsymbol{\vert}%
0.94 \\ 
\textsc{new default} & 0.3\TEXTsymbol{\vert}\textbf{1.37} & 0.5\TEXTsymbol{%
\vert}\textbf{1.23} & 0.9\TEXTsymbol{\vert}\textbf{1.21} & 0.4\TEXTsymbol{%
\vert}\textbf{1.74} & 0.9\TEXTsymbol{\vert}\textbf{1.21} & 0.5\TEXTsymbol{%
\vert}\textbf{1.41} & 1.9\TEXTsymbol{\vert}\textbf{1.06} \\ \hline
\multicolumn{8}{c}{$n=500$} \\ \hline
\textsc{t-stat} & 1.0\TEXTsymbol{\vert}\textbf{1.00} & 1.7\TEXTsymbol{\vert}%
\textbf{0.92} & 2.0\TEXTsymbol{\vert}0.87 & 0.9\TEXTsymbol{\vert}\textbf{1.01%
} & 2.3\TEXTsymbol{\vert}0.85 & 1.4\TEXTsymbol{\vert}\textbf{0.94} & 3.3%
\TEXTsymbol{\vert}0.78 \\ 
\textsc{sym-boot} & 1.1\TEXTsymbol{\vert}\textbf{0.99} & 1.5\TEXTsymbol{\vert%
}\textbf{1.02} & 1.9\TEXTsymbol{\vert}\textbf{1.09} & 0.9\TEXTsymbol{\vert}%
\textbf{1.03} & 2.2\TEXTsymbol{\vert}1.08 & 1.3\TEXTsymbol{\vert}\textbf{1.02%
} & 3.1\TEXTsymbol{\vert}1.06 \\ 
\textsc{asym-boot} & 1.1\TEXTsymbol{\vert}\textbf{1.00} & 1.5\TEXTsymbol{%
\vert}\textbf{0.96} & 1.9\TEXTsymbol{\vert}\textbf{0.96} & 1.5\TEXTsymbol{%
\vert}\textbf{1.01} & 1.9\TEXTsymbol{\vert}\textbf{0.95} & 1.5\TEXTsymbol{%
\vert}\textbf{0.97} & 2.9\TEXTsymbol{\vert}0.90 \\ 
\textsc{new default} & 0.9\TEXTsymbol{\vert}\textbf{1.01} & 0.5\TEXTsymbol{%
\vert}\textbf{1.44} & 0.7\TEXTsymbol{\vert}\textbf{1.40} & 0.7\TEXTsymbol{%
\vert}\textbf{1.32} & 0.7\TEXTsymbol{\vert}\textbf{1.37} & 0.6\TEXTsymbol{%
\vert}\textbf{1.43} & 1.2\TEXTsymbol{\vert}\textbf{1.38}%
\end{tabular}%
$

\end{center}
\linespread{1.00}\selectfont
\begin{small}
\linespread{1.00}\selectfont%
Notes: See Table 5.%
\linespread{1.00}\selectfont
\end{small}
\linespread{1.00}\selectfont%
\end{table}%

\begin{figure}[h] 
\begin{center}
\small%
\caption{Small Sample Results for Two Samples from HDMA Populations}\bigskip
\includegraphics[width=5.8356in,keepaspectratio]{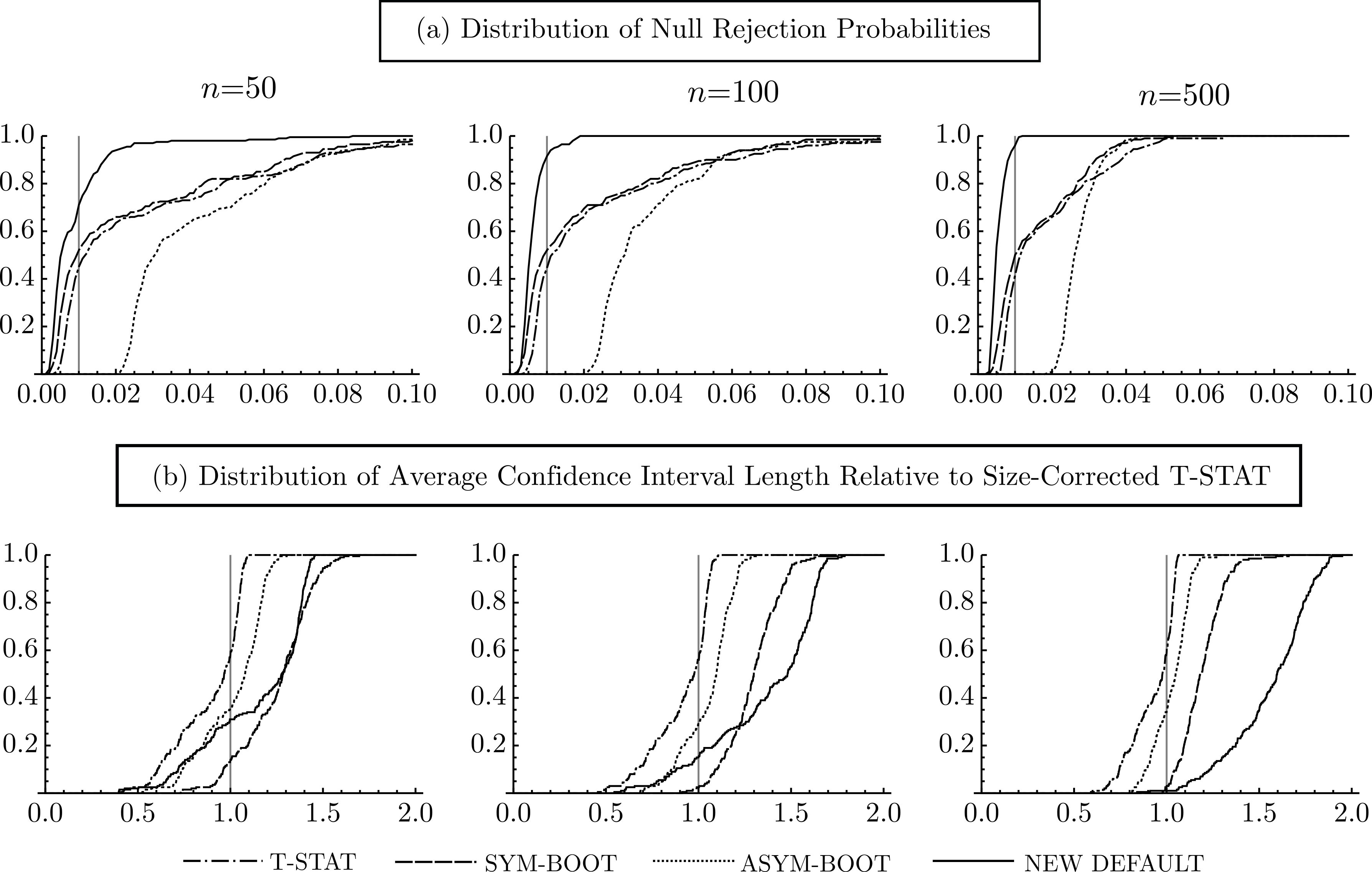}

\end{center}
\begin{small}%
\end{small}%
\end{figure}%

\begin{table}[h!] 
\begin{center}
\small%
\caption{Small Sample Results in Clustered Regression Design}\bigskip

$%
\begin{tabular}{lrrrrrrr}
\hline\hline
& N(0,1) & LogN & F(4,5) & \ \ t(3) & P(0.4) & Mix 1 & Mix 2 \\ \hline
\multicolumn{8}{c}{$n=50$} \\ \hline
\textsc{STATA} & 1.0\TEXTsymbol{\vert}\textbf{1.00} & 3.8\TEXTsymbol{\vert}%
0.72 & 4.3\TEXTsymbol{\vert}0.72 & 0.9\TEXTsymbol{\vert}\textbf{1.02} & 4.9%
\TEXTsymbol{\vert}0.71 & 2.1\TEXTsymbol{\vert}0.86 & 4.8\TEXTsymbol{\vert}%
0.74 \\ 
\textsc{Im-Ko} & 1.0\TEXTsymbol{\vert}\textbf{1.01} & 3.7\TEXTsymbol{\vert}%
0.73 & 4.1\TEXTsymbol{\vert}0.73 & 0.8\TEXTsymbol{\vert}\textbf{1.03} & 4.7%
\TEXTsymbol{\vert}0.72 & 2.0\TEXTsymbol{\vert}\textbf{0.88} & 4.6\TEXTsymbol{%
\vert}0.75 \\ 
\textsc{CGM} & 0.9\TEXTsymbol{\vert}\textbf{1.02} & 3.7\TEXTsymbol{\vert}0.68
& 4.1\TEXTsymbol{\vert}0.66 & 1.0\TEXTsymbol{\vert}\textbf{0.99} & 4.7%
\TEXTsymbol{\vert}0.65 & 2.0\TEXTsymbol{\vert}0.83 & 4.6\TEXTsymbol{\vert}%
0.65 \\ 
\textsc{new default} & 0.2\TEXTsymbol{\vert}\textbf{1.55} & 0.5\TEXTsymbol{%
\vert}\textbf{0.91} & 0.7\TEXTsymbol{\vert}\textbf{0.91} & 0.3\TEXTsymbol{%
\vert}\textbf{1.49} & 0.7\TEXTsymbol{\vert}\textbf{0.89} & 0.5\TEXTsymbol{%
\vert}\textbf{1.18} & 1.4\TEXTsymbol{\vert}\textbf{0.92} \\ \hline
\multicolumn{8}{c}{$n=100$} \\ \hline
\textsc{STATA} & 1.1\TEXTsymbol{\vert}\textbf{0.99} & 3.3\TEXTsymbol{\vert}%
0.78 & 3.9\TEXTsymbol{\vert}0.76 & 0.8\TEXTsymbol{\vert}\textbf{1.03} & 4.2%
\TEXTsymbol{\vert}0.73 & 2.1\TEXTsymbol{\vert}0.88 & 5.7\TEXTsymbol{\vert}%
0.70 \\ 
\textsc{Im-Ko} & 1.1\TEXTsymbol{\vert}\textbf{0.99} & 3.2\TEXTsymbol{\vert}%
0.78 & 3.8\TEXTsymbol{\vert}0.76 & 0.8\TEXTsymbol{\vert}\textbf{1.03} & 4.1%
\TEXTsymbol{\vert}0.74 & 2.0\TEXTsymbol{\vert}\textbf{0.88} & 5.6\TEXTsymbol{%
\vert}0.71 \\ 
\textsc{CGM} & 1.1\TEXTsymbol{\vert}\textbf{1.00} & 3.3\TEXTsymbol{\vert}0.73
& 3.9\TEXTsymbol{\vert}0.69 & 1.0\TEXTsymbol{\vert}\textbf{1.00} & 4.2%
\TEXTsymbol{\vert}0.66 & 2.2\TEXTsymbol{\vert}0.84 & 5.7\TEXTsymbol{\vert}%
0.61 \\ 
\textsc{new default} & 0.5\TEXTsymbol{\vert}\textbf{1.33} & 0.4\TEXTsymbol{%
\vert}\textbf{1.22} & 0.7\TEXTsymbol{\vert}\textbf{1.13} & 0.4\TEXTsymbol{%
\vert}\textbf{1.74} & 0.7\TEXTsymbol{\vert}\textbf{1.09} & 0.6\TEXTsymbol{%
\vert}\textbf{1.39} & 2.2\TEXTsymbol{\vert}0.95 \\ \hline
\multicolumn{8}{c}{$n=500$} \\ \hline
\textsc{STATA} & 1.1\TEXTsymbol{\vert}\textbf{0.99} & 1.9\TEXTsymbol{\vert}%
\textbf{0.90} & 2.7\TEXTsymbol{\vert}0.82 & 0.9\TEXTsymbol{\vert}\textbf{1.02%
} & 2.7\TEXTsymbol{\vert}0.82 & 1.6\TEXTsymbol{\vert}\textbf{0.92} & 4.0%
\TEXTsymbol{\vert}0.74 \\ 
\textsc{Im-Ko} & 1.1\TEXTsymbol{\vert}\textbf{0.99} & 1.9\TEXTsymbol{\vert}%
\textbf{0.90} & 2.7\TEXTsymbol{\vert}0.82 & 0.9\TEXTsymbol{\vert}\textbf{1.02%
} & 2.6\TEXTsymbol{\vert}0.82 & 1.6\TEXTsymbol{\vert}\textbf{0.92} & 4.0%
\TEXTsymbol{\vert}0.74 \\ 
\textsc{CGM} & 1.1\TEXTsymbol{\vert}\textbf{1.00} & 2.0\TEXTsymbol{\vert}%
\textbf{0.87} & 2.8\TEXTsymbol{\vert}0.77 & 0.9\TEXTsymbol{\vert}\textbf{1.00%
} & 2.8\TEXTsymbol{\vert}0.77 & 1.6\TEXTsymbol{\vert}\textbf{0.90} & 4.2%
\TEXTsymbol{\vert}0.68 \\ 
\textsc{new default} & 1.0\TEXTsymbol{\vert}\textbf{1.01} & 0.6\TEXTsymbol{%
\vert}\textbf{1.39} & 0.7\TEXTsymbol{\vert}\textbf{1.34} & 0.6\TEXTsymbol{%
\vert}\textbf{1.36} & 0.7\TEXTsymbol{\vert}\textbf{1.35} & 0.6\TEXTsymbol{%
\vert}\textbf{1.38} & 0.7\TEXTsymbol{\vert}\textbf{1.28}%
\end{tabular}%
$

\end{center}
\linespread{1.00}\selectfont
\begin{small}
\linespread{1.00}\selectfont%
Notes: Entries are the null rejection probability in percent, and the
average length of confidence intervals relative to average length of
confidence intervals based on size corrected \textsc{STATA} (bold if null
rejection probability is smaller than 2\%) of nominal 1\% level tests.%
\linespread{1.00}\selectfont
\end{small}
\linespread{1.00}\selectfont%
\end{table}%

\bigskip

\begin{figure}[h!] 
\begin{center}
\small%
\caption{Small Sample Results for CPS Clustered Regressions}

\bigskip
\includegraphics[width=5.8356in,keepaspectratio]{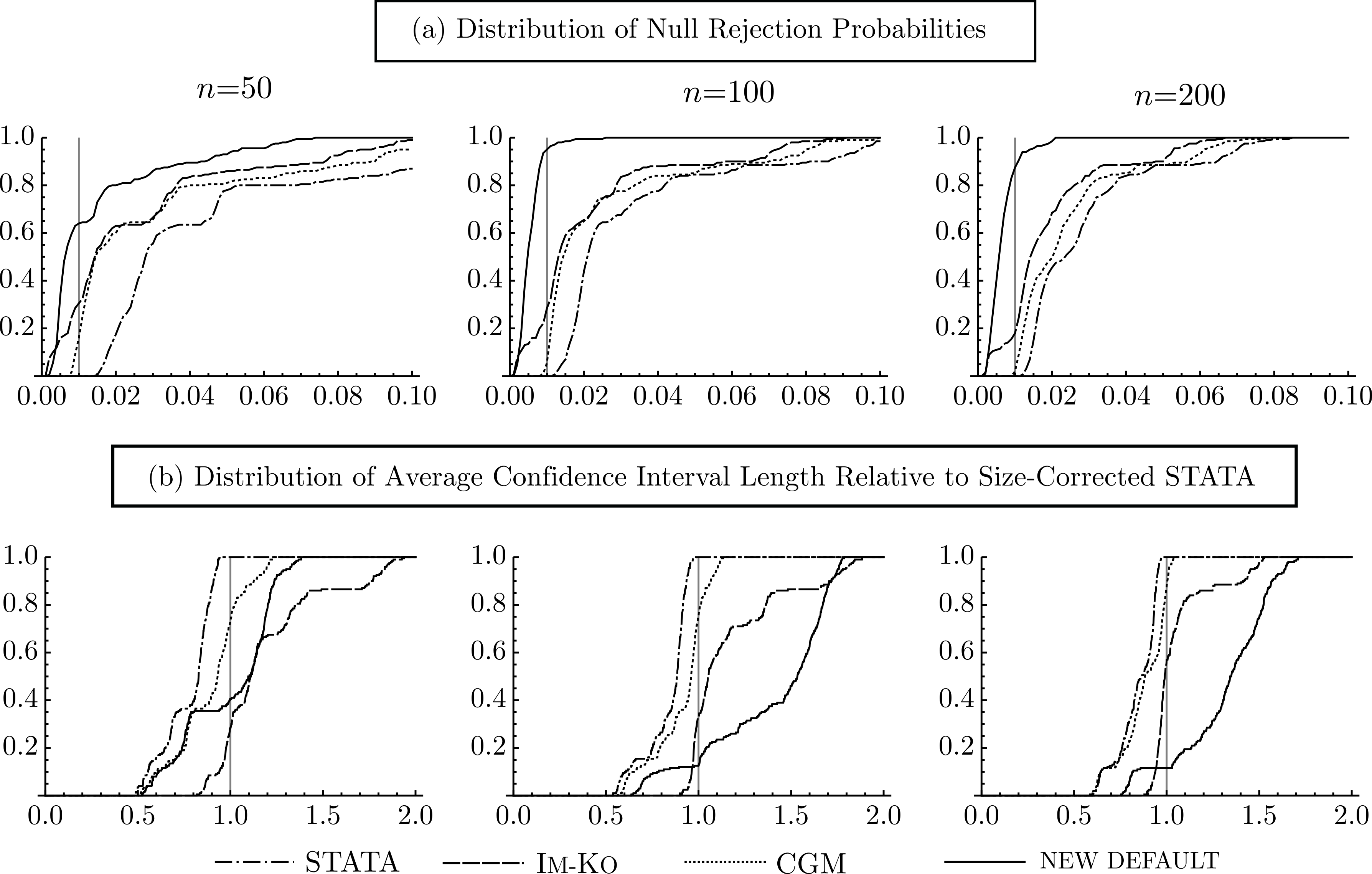}

\end{center}
\begin{small}%
\end{small}%
\end{figure}%
\bigskip

\end{small}%
\clearpage%

\newpage

\bibliographystyle{econometrica}
\bibliography{diss}

\end{document}